\documentclass[11pt,a4paper]{amsart}
\usepackage[latin1]{inputenc}
\usepackage{amsthm,amsmath,amsfonts,dsfont}
\usepackage{mathrsfs} %\mathscr
\usepackage{color}
\usepackage{a4wide}
\usepackage{amssymb}
\usepackage{hyperref}
\usepackage{calc} % simple arithmetics in latex
\usepackage{txfonts}
\usepackage{braket}
\usepackage[normalem]{ulem}
\usepackage{todonotes}
\usepackage{parskip} % no indent at start of paragraph

\newcommand{\comment}[1]{}

\theoremstyle{plain}
\newtheorem{lemma}{Lemma}[section]
\newtheorem{proposition}[lemma]{Proposition}
\newtheorem{theorem}[lemma]{Theorem}
\newtheorem{corollary}[lemma]{Corollary}

\theoremstyle{definition}
\newtheorem{assumption}[lemma]{Assumption}

\newtheorem{example}[lemma]{Example}

\theoremstyle{remark}

\newtheorem{remark}[lemma]{Remark}

\newcommand{\bC}{{\bf C}}

\newcommand{\bN}{{\bf N}}
\newcommand{\bP}{{\bf P}}
\newcommand{\bM}{{\bf M}}
\newcommand{\bL}{{\bf L}}

\newcommand{\R}{\mathbb{R}}

\newcommand{\N}{\mathbb{N}}
\newcommand{\F}{\mathcal{F}}
\newcommand{\Z}{\mathbb{Z}}

\newcommand{\dx}{\Delta x}

\newcommand{\dv}{\Delta v}
\newcommand{\dZ}{\Delta Z}

\newcommand{\barv}{\overline{v}}
\newcommand{\barA}{\overline{A}}
\newcommand{\barB}{\overline{B}}

\newcommand{\barV}{\overline{V}}
\newcommand{\barg}{\overline{g}}
\newcommand{\barY}{\overline{Y}}
\newcommand{\barZ}{\overline{Z}}
\newcommand{\barF}{\overline{\mathcal{F}}}

\newcommand{\tV}{\widetilde{V}}
\newcommand{\tU}{\widetilde{U}}
\newcommand{\hatV}{\widehat{V}}
\newcommand{\hatv}{\widehat{v}}

\newcommand{\tN}{\widetilde{N}}
\newcommand{\txi}{\widetilde{\xi}}
\newcommand{\ttau}{\widetilde{\tau}}
\newcommand{\tg}{\widetilde{g}}

\newcommand{\abs}[1]{\left\lvert#1\right\rvert} % absolute value
\newcommand{\norm}[1]{\left\lVert#1\right\rVert} % norm
\newcommand{\indic}[1]{\mathbf{1}_{#1}} % indicator function
\newcommand{\ip}[2]{\left\langle #1\,,#2\right\rangle} % inner product (< , >)
\newcommand{\sip}[2]{\left[ #1\,,#2\right]} % quadratic variation: [ , ]
\newcommand{\floor}[1]{\left\lfloor#1\right\rfloor} % floor command
\newcommand{\ceil}[1]{\left\lceil#1\right\rceil} % ceiling command
\newcommand{\f}[2]{\frac{#1}{#2}}
\newcommand{\qvar}[1]{\left[ #1 \right]}
\newcommand{\var}{\operatorname{var}}
\newcommand{\diag}{\operatorname{diag}}

\makeatletter
\newcommand{\xRightarrow}[2][]{\ext@arrow 0359\Rightarrowfill@{#1}{#2}}
\makeatother

\numberwithin{equation}{section}

\allowdisplaybreaks % Allows pagebreaks in equations.
\author{Christian Bayer}
\address{Weierstrass Institute\\Mohrenstr.~39\\10117 Berlin\\Germany}
\email{christian.bayer@wias-berlin.de}

\author{Ulrich Horst}
\address{Humboldt University Berlin\\Department of Mathematics\\Unter den
  Linden 6\\10099 Berlin\\Germany}
\email{horst@math.hu-berlin.de}

\author{Jinniao Qiu}
\address{University of Michigan\\Department of Mathematics\\
East Hall, 530 Church Street\\
Ann Arbor, MI 48109-1043\\USA}
\email{qiujinn@gmail.com}

\title{A Functional Limit Theorem for Limit Order Books with State Dependent Price Dynamics$^\dag$}\thanks{$^\dag$Financial support from the SFB 649 ``Economic Risk" is gratefully acknowledged. We thank seminar and conference participants at various institutions for helpful comments and suggestions. The paper was finished while Horst was visiting the Center for Interdisciplinary Research at Bielefeld University. Grateful acknowledgment is made for hospitality. An earlier version of this paper was entitled {\sl A Functional Limit Theorem for Limit Order Books}}

\begin{document}

\maketitle

\begin{abstract}
We consider a stochastic model for the dynamics of the two-sided limit order book (LOB). Our model is flexible enough to allow for a dependence of the price dynamics on volumes. For the joint dynamics of best bid and ask prices and the standing buy and sell volume densities, we derive a functional limit theorem, which states that our LOB model converges in distribution to a fully coupled SDE-SPDE system when the order arrival rates tend to infinity and the impact of an individual order arrival on the book as well as the tick size tends to zero. The SDE describes the bid/ask price dynamics while the SPDE describes the volume dynamics.
\end{abstract}
\vspace{4mm}


\vspace{2mm}


\vspace{6mm}

\thispagestyle{empty}

\renewcommand{\baselinestretch}{1.05}\normalsize

\section{Introduction}
\label{sec:introduction}

In modern financial markets almost all transactions are settled through Limit
Order Books (LOBs). An LOB is a record -- maintained by an exchange or
specialist -- of unexecuted orders awaiting execution. Unexecuted (standing)
orders are executed against incoming market orders according to a set of
precedence rules. Most exchanges give orders at better price levels priority
over orders submitted at less competitive price levels (``price priority'')
and orders with the same price-priority are typically (though not always)
executed on a first-in-first-out basis (``time-priority'')\footnote{We note that some exchanges also use matching algorithms based on pro-rata allocations.} From a
mathematical perspective, LOBs can thus be viewed as high-dimensional complex
priority queuing systems.  In this paper, we present a probabilistic framework
within which to derive functional scaling limits for LOBs from individual
order arrival dynamics. We assume that orders arrivals and cancellations follows a occur according to a Poisson dynamics relative to the best bid and ask prices. With our choice of scaling, prices follow a diffusion process while volume density functions can be described by an infinite dimensional SDE, that is coupled with the price process. As a special case we obtain law-of-large-numbers-type scaling with absolutely continuous (in time) volume  density functions. 

\subsection{Literature review}

There is a substantial economic and econometric literature on LOBs
\cite{Biais, Cebiroglu-Horst, Hautsch-Huang, EasleyOHara,GlostenMilgrom,Rosu}
that puts a lot of emphasis on the realistic modeling of the working of the
LOB. At the same time, only few authors have analyzed LOB dynamics from a more
probabilistic perspective. Kruk \cite{Kruk2003} studied a queuing theoretic
LOB model with finitely many price levels. For the special case of two price
levels, in his model the scaled number of standing buy and sell orders at the
top of the book converges weakly to a semimartingale reflected two-dimensional
Brownian motion in the first quadrant. Cont, Stoikov and Talreja
\cite{ContStoikovTalreja} proposed an LOB model with finitely many submission
price levels where the LOB dynamics follows an ergodic Markov process. Cont
and DeLarrard \cite{ContDeLarrard2012b} established a scaling limit for a
Markovian limit order market in which the state of the book is represented by
the best bid and ask prices along with the liquidity standing at these prices
(``top of the book'').  Under heavy traffic conditions their bid and ask queue
lengths are given by a two-dimensional Brownian motion in the first quadrant
with reflection to the interior at the boundaries, similar to the diffusion
limit for two price levels in \cite{Kruk2003}.

When scaling limits of financial price fluctuations %\cite{BayraktarHorstSircar,DuffieProtter,FoellmerSchweizer,Garman,HorstRothe} 
or joint price and volume fluctuations at selected price levels \cite{ContDeLarrard2012b,Kruk2003} are studied, the limit can naturally be described by ordinary differential equations or finite-dimensional diffusion processes, depending on the choice of scaling. The mathematical analysis is more challenging when the dynamics of the full book is considered. To the best of our knowledge, Osterrieder \cite{Osterrieder} was the first to model LOBs as measure-valued diffusions. Horst and Paulsen \cite{Horst-Paulsen} were the first to prove a scaling limit for the full order book. With their choice of scaling the joint dynamics of volumes and prices converges to a coupled system of two PDEs that describe the limiting volume dynamics and two ODEs that describe the limiting price dynamics. A related model with state-dependent prices in the approximating sequence but constant prices in the limit is analyzed in \cite{Gao} where the scaling limit is also empirically tested against real LOB data. Lakner et al \cite{Lakner1} derived a scaling limit for a one-sided limit order book model under the assumption that average investors place their limit orders above the current best ask price. The opposite case when orders are placed in the spread with higher probability is analyzed in \cite{Lakner2}, where the authors use a coupling between a simple one-sided limit order book model and a branching random walk to characterize the diffusion limit.  Lasry and Lions \cite{LasryLions}, starting from a mean-field game perspective, also describe the LOB by a coupled PDE model with the interface given by the price; see also \cite{lehalle}. Keller-Ressel and M\"{u}ller \cite{KRM} describe the LOB as a coupled system of SPDEs separated by a random interface (Stochastic Stefan problem) that can again be interpreted as the price.  
\\[4pt]
Despite the considerable empirical evidence that the state of the order book, especially order imbalance at the top of the book, has a noticeable impact on order dynamics (see \cite{Biais,Cebiroglu-Horst,Hautsch-Huang} and references therein) the order flow in most limit order book models either follows independent Poisson dynamics or depends on the price process only as in \cite{Gao,Horst-Paulsen,Lakner1,Lakner2}. Notable exceptions are the papers by Abergel and Jeddi \cite{AJ}, where Hawkes-type dynamics are used, Huang et al \cite{Rosenbaum1} and Huang and Rosenbaum \cite{Rosenbaum2} where the ergodicity of a general Markovian order book model is studied and the diffusivity of the rescaled price process in this general framework is derived, and Horst and Kreher \cite{Horst-Kreher} who obtained a deterministic scaling limit for LOBs with fully state dependent event dynamics. In this paper we consider a diffusion limit for the full LOB dynamics, both prices and volumes, where the price dynamics depends on standing volumes. 

\subsection{Our contribution}
\label{sec:our-contribution}

 As in \cite{Horst-Paulsen} our limit result requires two time scales: a fast time scale for cancellations and limit order placements outside the spread and a comparably slow time scale for market order arrivals and limit order placements in the spread. The different time scales account for the well-documented fact that placements and cancellations occur much more
frequently than price changes. For instance, using LOBSTER data for Jan 2, 2014 Horst and Paulsen \cite{Horst-Paulsen}  computed the empirical probabilities for an incoming order to trigger price change for Apple (0.016), Ebay (0.02), Facebook (0.02), and Microsoft (0.002). Estimates of similar order are reported in \cite{Gao} for the stock Bank of America. 

In our model, market orders and limit order placements in the spread trigger
price changes. We refer to these order types as {\sl active orders}. The
probability of an active order being a market order or spread placement at the
bid or ask side of the book depends on the standing volume. Limit order
placements outside the spread and cancellations of standing volume do not lead
to price changes. We refer to these order types as {\sl passive
  orders}. Passive orders arrive according to an independent Poisson at random
distances from the best bid and ask price for random amounts (placements) and
propositions (cancellations), respectively.  

In this framework, after suitable scaling the price processes follow diffusion
processes whose coefficients depend on standing volumes, and the volume
density functions (in absolute coordinates) are deterministic and absolutely
continuous (in time) functions of the price process. In particular, all
fluctuations in standing volumes result from fluctuations in the price
process. While such a scaling already results in a realistic limiting LOB
dynamics, it seems desirable to us to allow for additional fluctuations in
standing volumes that do not originate from price fluctuations. Our framework
is flexible enough to allow for such fluctuations. In a model with both
positive and negative placements (additive cancellations), we may allow
placements to be correlated on a common factor that translates into an
additive martingale part driving the volume dynamics. While the ``common
factor extension'' should be viewed as a mostly mathematical extension it does
shed further light on the importance of time scales in our model. Our analysis
suggests that even the simple case of correlated additive volume fluctuations
requires some form of ``common factor'' upon which to condition volume
fluctuations that changes on a much slower time scale than individual order
arrival dynamics and cancellations. Of course, many other approaches to
modelling volume fluctuations are perceivable.

Our main result states that when the rate of active order arrivals scales by a
factor $n$, the rate of passive order arrivals scales by a factor $n^2$, the
tick size scales by a factor $1/\sqrt{n}$, the sizes (proportions) of incoming
orders (canceled volumes) scale by a factor $1/n^2$ and the impact of the
common factor scales by a factor $1/n$, then the price processes converge to
an SDE and the volume density functions in absolute coordinates converge to an
infinite dimensional SDE (SPDE in relative coordinates) as $n \to \infty$. The
convergence concept we use is weak convergence in the class of c\`{a}dl\`{a}g
stochastic processes with sample paths in $\mathbb{R}^2 \times ( H^{-1} )^2$
where $H^{-1}$ denotes the Sobolev space of order $-1$. The main challenge is
to prove convergence of the $\left(H^{-1}\right)^2$-valued volume processes.
To prove tightness we decompose the volume processes into three components
describing the aggregate placements, the proportionality of the cancellations
and the impact of the common factor at the various price levels, respectively.
We establish norm-bounds for each of these processes from which we then deduce
that the volume process as a whole satisfies a standard tightness criteria.
To characterize the limit we first prove joint convergence of prices and the
martingale part of the volume processes. Subsequently, we identify the limits
of aggregate placements and cancellations and use C-tightness
{(i.e., tightness with continuous limit processes)} of the
price and martingale part to prove joint convergence of all the processes to
the desired limit.

The key observation is that tightness of the volume processes is guaranteed
under mild assumptions on the price process; it does not require any
particular dependence of prices on volumes. In particular, it does {\sl not}
require the price process to be independent from volumes. The {\sl
  characterization} of the limiting volume dynamics requires joint convergence
of volumes and prices along a subsequence. This is guaranteed if the price
process is C-tight, a condition which, too, does not require any particular
assumptions on the interplay between prices and volumes. If the limiting price
process is known upfront, either because the approximating price process is
independent from volumes as in \cite{Horst-Paulsen} or the limiting price
process is state-independent as in \cite{Gao}, then the limiting volume
process exists and the joint dynamics of prices and volumes is fully
specified. The added difficulty under state dependence is the identification
of the limiting price/volume process. To this end, we first characterize the
limiting volume dynamics as a function of the---unknown yet existing---weak
accumulation point of the price process. Based on this partial
characterization result, we then fully characterize the joint evolution of
prices and volumes. 

The remainder of  this paper is organized as follows. In Section \ref{sec:model-main-results} we define a sequence of limit order books in terms of our scaling parameters, state the main result and give an outline of the proof. Section \ref{sec:scaling-limit-volume} is devoted to the analysis of the volume dynamics. Section \ref{new-model} characterizes the joint limit of the price/volume process. A general result on the characterization of stochastic process limits, general tightness results and some technical proofs are collected in an appendix.

\textit{Notational conventions.} For any (deterministic or random) function
$u: [0,\infty) \times \R \to \R$ we denote by $u(t): \R \to \R$ the function
$x \mapsto u(t,x)$ for $t \in [0,\infty)$. Unless otherwise stated,
$\left(L^p,\,\|\cdot\|_{L^p}\right)$ ($p\in[1,\infty]$) refers to the space
$L^p\left( \R, \mathcal{B}(\R), dx \right)$. $L^2$ is equipped with the usual inner product $\langle \cdot, \cdot \rangle$. For $\sigma$-algebras $\mathcal{G} \subset \mathcal{F}$ we shall write  $E_{\mathcal{G}}\left[ \cdot \right] \coloneqq E\left[\cdot \, | \, \mathcal{G} \right]$. Further, all random variables are defined
on a common probability space $\left(\Omega,\mathcal{F},\mathds{P}\right)$. We may write $X(t)$ or $X_t$ for the value of a stochastic process $X$ at time $t \geq 0$.   

\section{Model and main results}
\label{sec:model-main-results}

%%%%%%%%%%%%%%%%%%%%%%%%%%%
%%%%%%%%%%%%%%%%%%%%%%%%%%%
%%%%%%%%%%%%%%%%%%%%%%%%%%%
%%%%%%%%%%%%%%%%%%%%%%%%%%%
%%%%%%%%%%%%%%%%%%%%%%%%%%%
%%%%%%%%%%%%%%%%%%%%%%%%%%%

\subsection{The discrete model}
\label{sec:discrete-model}

In this section we introduce a sequence of continuous time order book models with state-dependent price dynamics.
The set of price levels at which orders can be submitted in the $n$-th model is
$\{x^{n}_j\}_{j \in \mathbb{Z}}$. 
We put
$x^{n}_j := j \cdot \Delta x^{n}$ for each $j \in {\mathbb Z}$ where $\Delta
x^{n}$ is the {\sl tick size}, i.e. the minimum difference between two
consecutive price levels. 

The {\sl state} of the order book at time $t \geq 0$ is given by a pair $\big(B^n_t,A^n_t \big)$ with $B^n_t \leq A^n_t$ of best bid and ask prices together with the buy and sell limit order volumes standing at the different price levels. We identify volumes at the best bid and ask side of the book with step functions $v_{b/a}: [0,\infty) \to \mathbb{R}$,
\begin{equation*}
	v^{n}_{b}(t,x):=\sum_{j \in \mathbb{Z}} v^{n,j}_{b,t} \mathds{1}_{[x^{n}_j,
          x^{n}_{j+1})}( x), \quad
	v^{n}_{a}(t,x):=\sum_{j \in \mathbb{Z}} v^{n,j}_{a,t} \mathds{1}_{[x^{n}_j,
          x^{n}_{j+1})}(x) \quad (x \in \R)
\end{equation*}
with the interpretation that the liquidity available for selling $j \in \mathbb{N}$ ticks {\sl below} the best bid price at time $t \geq 0$ is given by 
\[
	\int_{B^n_t + j\dx^{n}}^{B^n_t + (j+1)\dx^{n}} {
          v^{n}_{b}(x)} dx  =
	\dx^{n} \cdot v^{n,B^n_t/\dx^n +
              j}_{b}, 
\]
while the liquidity available for buying $j \in \mathbb{N}$ ticks {\sl above} the best ask price at that time is given by 
\[
	\int_{A^n_t + j\dx^{n}}^{A^n_t + (j+1)\dx^{n}} {
          v^{n}_{a}(x) dx}  =
	\dx^{n} \cdot v^{n,A^n_t/\dx^n + j}_{a}.
\]
Our choice of notation allows to treat both sides of the books symmetrically
and hence simplifies the presentation of the results.\footnote{We acknowledge
  that the choice of notation for the bid side is not intuitive as it implies
  that the volume standing at price level $x$ at time $t$ is given by
  $v_b(t,2B^n_t - x)$. However, it greatly unifies the presentation of the results and proofs.} 
We are mainly interested in the volume density functions in {\sl relative coordinates}, denoted 
\[
	u^n_b(t,x) := v^n_b(t,B^n_t + x) \quad \mbox{and} \quad u_a(t,x) := v^n_a(t,A^n_t + x)
\]	
respectively. That is, $u^n_b(t,j\cdot \Delta x^n)$ denotes the liquidity standing $j$ ticks below the best bid and $u^n_a(t,j\cdot \Delta x^n)$ denotes the liquidity standing $j$ ticks above the best ask. 

We call $\{u^n_{b/a}(t,x) : x \geq 0\}$ the {\sl visible book} and
$\{u^n_{b/a}(t,x) : x < 0\}$ the {\sl shadow book} at time $t \geq 0$ of the
bid $(b)$, respectively the ask $(a)$ side of the book.  The visible book
collects the orders awaiting execution. The shadow book specifies the volumes
that will be placed into the spread, should such an event occur next. Since
several consecutive spread placements may occur the shadow book is defined on
the whole negative half-line. It will undergo random fluctuations similar to
the visible book and is just convenient tool to describe spread
placements. Its precise working will be further described in Section
\ref{section-active} below. See also \cite{Horst-Kreher,Horst-Paulsen} for a
discussion of the shadow book.

Throughout, indices $b$ and $a$ refer to bid and ask side volumes, respectively. We often use the index $r$ to refer to either side of the book. Occasionally, we drop the index altogether and write for instance just $v(t,x)$ if we give generic arguments that apply to both sides of the book. In both cases, we use $R^n(t)$ or $R^n_t$ to denote either the best bid $(r=b)$ or the best ask $(r=a)$ price.   

\begin{assumption}\label{ass-initial}
The sequence of initial data $(A_0^n,B_0^n, v^n_{a}(0,\cdot),v^n_{b}(0,\cdot))$ converges to $(A_0,B_0,v_{a,0}(\cdot),v_{b,0}(\cdot))$ in both $\mathbb{R}^2 \times L^2 \times L^2$ and $\mathbb{R}^2 \times L^{\infty} \times L^{\infty}$.
\end{assumption}

There are eight events -- labeled $\bM_{r}, \bL_{r}, \bC_{r}, \bP_{r}$ ($r=a,b$) -- that change the state of the book.  The events $\bM_b, \bL_{b}, \bC_{b}, \bP_{b}$ affect the bid side of the book:
\begin{align*}
	\textbf{M}_b&\ldots \text{market sell order}& \textbf{L}_b& \ldots
        \text{buy limit order placed in the spread}\\ 
	\textbf{C}_b&\ldots \text{cancellation of buy volume}& \textbf{P}_b&
        \ldots \text{buy limit order not placed in spread}
\end{align*}
The events $\bM_a, \bL_{a}, \bC_{a}, \bP_{a}$ affect the ask side of the book:
\begin{align*}
	\textbf{M}_a&\ldots \text{market buy order}& \textbf{L}_a& \ldots \text{sell limit order placed in the spread}\\
	\textbf{C}_a&\ldots \text{cancellation of sell volume}& \textbf{P}_a&
        \ldots \text{sell limit order not placed  in the spread}. 
\end{align*}
In the sequel we specify how different order types change the state of the book. 

%%%%%%%%%%%%%%%%%%%%%%%%%%%%%%%%%%%
%%%%%%%%%%%%%%%%%%%%%%%%%%%%%%%%%%%
%%%%%%%%%%%%%%%%%%%%%%%%%%%%%%%%%%%

\subsubsection {Active orders and price dynamics}  \label{section-active}

We assume that market order arrivals (Events $\textbf{M}_{b/a}$) and placements of limit orders in the spread (Events $\textbf{L}_{b/a}$) lead to price changes. In fact, a market order that does not lead to a price change is equivalent to a cancellation at the top of the book. We refer to these order types as {\sl active orders}.

\begin{assumption}
  \label{ass:Poisson}
  Active orders arrive according to a Poisson process $\tN^n$
  with intensity $\mu^n$. The corresponding jump times
  $\left(\ttau^n_{i} \right)_{i=1}^\infty$ will be called active order times.
\end{assumption}

In our model market orders match precisely against the volume at the
top of the book. In other words, a market order arriving at time $\ttau^n_{i}$
is good for $v^n_r(\ttau^n_{i}-,R^n(\ttau^n_{i}-)) \cdot \Delta x^n$
shares. We further assume that limit orders placed into the spread are placed
at the first best price increment and that their sizes are determined by the
shadow book. Specifically, a limit order placed into the spread at time
$\ttau^n_{j}$ is placed at the price level $R^n(\ttau^n_{j}-) - \Delta x^n$
and its size is $v^n_r(\ttau^n_{j}-,R^n(\ttau^n_{j}-)-\Delta x^n) \cdot
  \Delta x^n$. If another limit order placement occurs at the next active order time 
$\ttau^n_{j+1}$, then the order is placed at $R^n(\ttau^n_{j}-) - 2\Delta x^n$ and its size is {$v^n_r(\ttau^n_{j+1}-,R^n(\ttau^n_{j}-)-2\Delta x^n) \cdot \Delta x^n$}. In between two active orders cancellations and placements may occur in the shadow book so typically $v^n_r(\ttau^n_{j+1}-,R^n(\ttau^n_{j}-)-2\Delta x^n) \neq v^n_r(\ttau^n_{j}-,R^n(\ttau^n_{j}-)-2\Delta x^n)$; cf. Section \ref{section-passive} and Figures 1 and 2 below. 

\begin{figure}[h]\label{fig1}
\hspace{-3cm}
\begin{minipage}{0.49\textwidth}
	\centering
	\includegraphics[width=15cm, height=14cm]{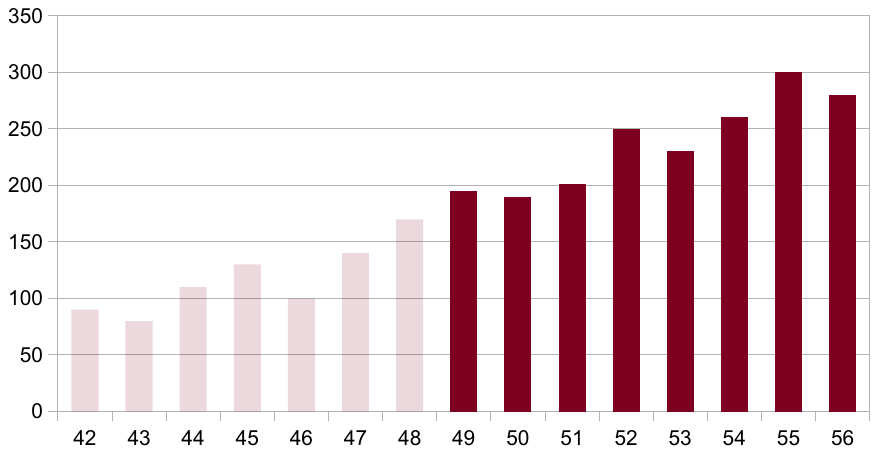}
\end{minipage}
\begin{minipage}{0.49\textwidth}
	\centering
	\includegraphics[width=15cm, height=14cm]{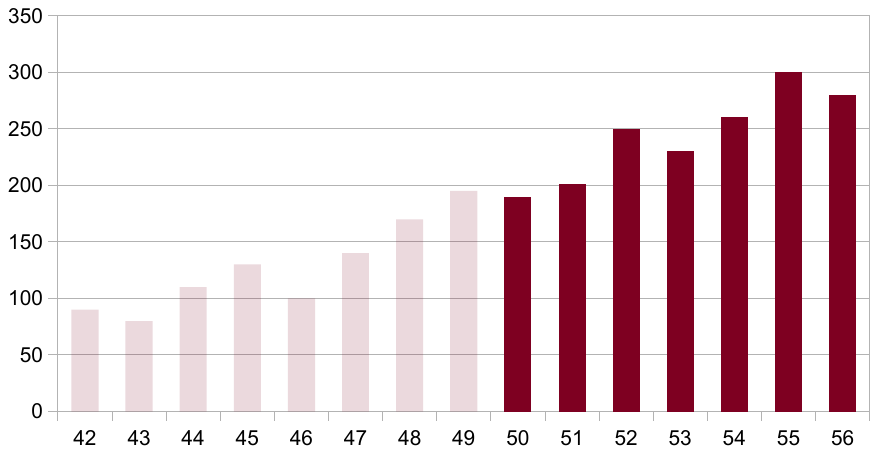}
\end{minipage}
\vspace{-6.5cm}
\caption{Ask--side volume function at time $\ttau^n_{i}-$ (left) and $\ttau^n_{i}$ (right) of the visible (dark coloured) and shadow book (light coloured) when a market order arrives at $\ttau^n_{i}$.}
%\end{figure}
%
%\begin{figure}
\hspace{-3cm}
\begin{minipage}{0.49\textwidth}
	\centering
	\includegraphics[width=15cm, height=14cm]{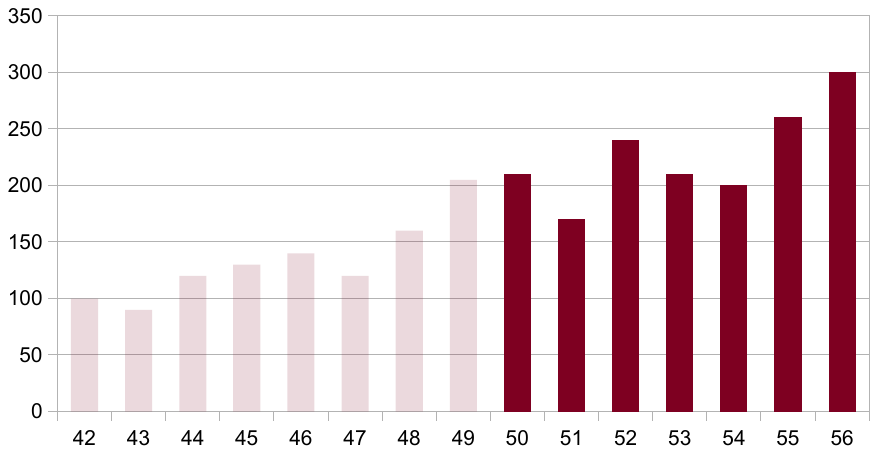}
\end{minipage}
%\caption{bla}
\begin{minipage}{0.49\textwidth}
	\centering
	\includegraphics[width=15cm, height=14cm]{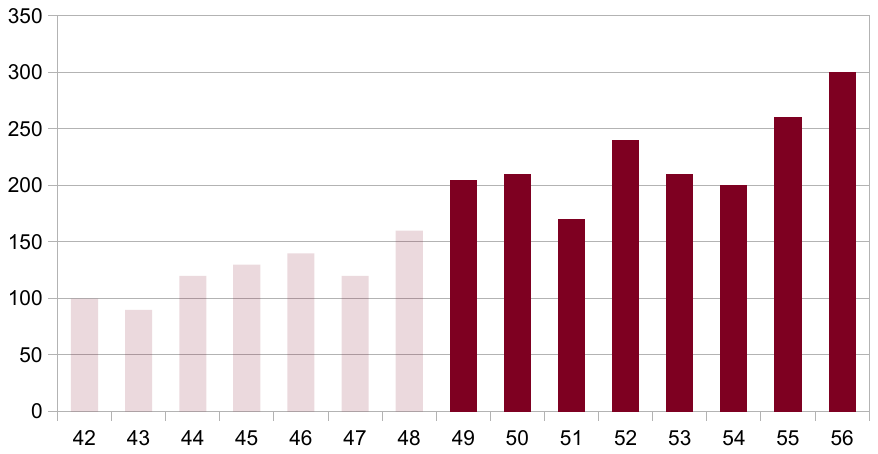}
\end{minipage}
\vspace{-6.5cm}
	\caption{Ask--side volume function at time $\ttau^n_{i+1}-$ (left) and $\ttau^n_{i+1}$ (right) of the visible and shadow book when a spread placement occurs at $\ttau^n_{i+1}$.}
\end{figure}

We allow the probabilities of price changes to depend on standing volumes. To this end, we fix smooth non-negative functions $\varphi^{r}: \R \to \R$ and put 
\[
Y^{r,n}_t:=\langle v_{r}^n(t,\cdot),\,\varphi^r(\cdot - R^n(t)) \rangle \quad
(r=a,b).
\]	
We interpret $Y^{r,n}_t$
as a measure for the volume standing at the top of the book or the total bid $(r=b)$ or ask $(r=a)$ side volume, depending on the choice of $\varphi^r$. 
The price dynamics is now defined as
\begin{equation}
  \label{eq:price-dynamics-general}
    dR^n(t) = {\Delta x^n} \xi^n_{r,\tN^n(t)}\,
    d\tN^n(t)
\end{equation}
where the random variables $\xi^n_{r,\tN^n(t)}$ take values in $\{ 0,
  -1, +1\}$. Their distribution will depend on the bid
  and ask price and on the state of the volumes placed. Hence, the
  model considered is not of zero-intelligence type. More precisely, we work under the following assumption. 
  
\begin{assumption}\label{ass-ODE}
Let $\big(\mathcal{F}^n_t \big)$ denote the filtration generated by the $n$-th model (the precise definition is given in (\ref{filtration}) below). For $r=a,b$ there exist functions $b^n_{r}\in C^1(\mathbb{R}^4)$ and
  $\sigma^n_r \in C^1(\mathbb{R}^4;\mathbb{R}^{2\times 1})$ such that 
  \begin{align}
    &E_{\mathscr{F}^n_{t-}\vee
      \sigma(\tN^n(t))} \left[\xi^n_{r,\tN^n(t)}\right] =
    \frac{1}{\sqrt{n}}b^n_{r}(B^n_{t-},A^n_{t-},Y^{b,n}_{t-},Y^{a,n}_{t-}),
    \label{ass-mean}\\
    &\text{Cov}_{\mathscr{F}^n_{t-}\vee \sigma(\tN^n(t))}
    \left[
    \begin{pmatrix} 
      	\xi^n_{b,\tN^n(t)} \\ \xi^n_{a,\tN^n(t)}
     \end{pmatrix}
     \right]
    = \begin{pmatrix} \sigma^n_b \\ \sigma^n_a \end{pmatrix} \cdot
    \begin{pmatrix} \sigma^n_b \\ \sigma^n_a \end{pmatrix}^\top (B^n_{t-},A^n_{t-},Y^{b,n}_{t-},Y^{a,n}_{t-}) 
    , \label{ass-var}
  \end{align}
  for any $t>0$, and $(b^n_{r},\sigma^n_r)$ converges to $(b_{r},\sigma_r)$ in
  $C(\mathbb{R}^4)\times C(\mathbb{R}^4;\mathbb{R}^{2\times 1})$
  (uniformly) such that the matrix
  $\begin{pmatrix} \sigma_b \\ \sigma_a \end{pmatrix}$ is invertible and
  $(b_{r},\sigma_r) \in C^1(\mathbb{R}^4) \times
  C^1(\mathbb{R}^4;\mathbb{R}^{2\times 1})$ and the limiting objects are
  uniformly bounded.
\end{assumption}

\begin{lemma}
  \label{lem:price-tightness}
  The sequence of price processes $(A^n, B^n)$ is C-tight.
\end{lemma}
\begin{proof}
  Immediately by Theorem \ref{thr:aldous} and Lemma \ref{lem-c-tight} as price increments are bounded by $\Delta x^n$.
\end{proof}

It is clearly desirable to avoid crossing of bid and ask prices. One possibility is to introduce a
reflection term and to scale the prices such that they converge to reflected Brownian motion in the limit as in \cite{KangWilliams2007}.  Another is to consider short time scales as illustrated by the following example. 

\begin{example}
Let us assume that
\begin{eqnarray*}
	\mathds{P}_{\mathscr{F}^n_{t-}\vee
      \sigma(\tN^n(t))}[\xi^n_{b/a,\tN^n(t)} = \pm 1] = g^n_{r}(\pm 1, B^n_{t-},A^n_{t-},Y^{b,n}_{t-},Y^{a,n}_{t-})
\end{eqnarray*}
for smooth functions $g^n_{r}(\pm 1, \cdot)$ that satisfy for any $(y^1,y^2,y^3,y^4)\in\mathbb{R}^4$,
\[
	g^n_{b}(+ 1, y^1,y^2,y^3,y^4) = g^n_{a}(- 1, y^1,y^2,y^3,y^4) = 0 , \quad\text{if }y^2 - y^1 < \epsilon \,\text{ for some } \epsilon >0
\]
%if $y^2 - y^1 < \epsilon$ for some $\epsilon > 0$
and
\[
	g^n_{r}(+1, y^1,y^2,y^3,y^4) - g^n_{r}(-1, y^1,y^2,y^3,y^4)
	= \frac{1}{\sqrt{n}}b^n_{r}(y^1,y^2,y^3,y^4).
\]
Then Assumption \ref{ass-ODE} is satisfied up to some stopping time. If we further assume that $\xi^n_{b,\tN^n(t)}  \cdot \xi^n_{a,\tN^n(t)} = 0$, then at most one price moves at any active order time.
\end{example}

Notice that the above example makes sense for short times. The next example avoids this limitation.

\begin{example}
  For simplicity, we give an example where price dynamics do not depend on
  $Y$. It is, however, simple to extend the example. Let $\zeta^n_{l,i}$,
  $l=1,2$, $i \in \N$, be the increments (indexed by $i$) of two independent
  Donsker type approximations of two independent geometrical Brownian motions
  denoted by $S^l_t$, $l=1,2$, both of which are constructed such that
  positivity of the cumulative sum of $(\zeta_{l,i}^n)_{i \in \N}$ is ensured
  for any $n$---for instance, by reflection. We may suppose that
  $\zeta^n_{l,i}$ takes values in $\{0,\pm 1\}$. Define
  $\xi^n_{b,i} \coloneqq \zeta^n_{1,i}$ and
  $\xi^{n}_{2,i} \coloneqq \zeta^n_{1,i} + \zeta^n_{2,l}$. Hence,
  $\xi^n_{r,i}$ take values in $\{0,\pm1,\pm2\}$---an
    inconsequential violation of the assumption that $-1 \le \xi \le 1$.  In
  the limit we have $B_t = S^1_t$ and $A_t = S^1_t + S^2_t \ge
  B_t$. 
\end{example}

%%%%%%%%%%%%%%%%%%%%%%%%%
%%%%%%%%%%%%%%%%%%%%%%%%%
%%%%%%%%%%%%%%%%%%%%%%%%%

\subsubsection{Passive orders and volume changes} \label{section-passive}

Limit order placements outside the spread and cancellations of standing volume
do not change prices. We refer to these order types as {\sl passive
  orders}. In our model cancellations (Events $\textbf{C}_{b/a}$) occur for
random {\sl proportions} of the standing volume while
limit order placements outside the spread (Events $\textbf{P}_{b/a}$)
occur for random {\sl volumes} at random price levels.

\begin{assumption}
  \label{ass:Poisson-passive}
 Passive orders arrive according to independent Poisson processes $N^n_b$ and $N^n_a$ that are independent of $\widetilde N^n$ with intensities $\lambda^n_b$ and $\lambda^n_a$ at the bid and ask side of the book, respectively. The corresponding jump times $\left( \tau^n_{b/a,i} \right)_{i=1}^\infty$ will be called passive order times.
\end{assumption}

The {\sl submission and cancellation price levels} are chosen
{\sl relative} to the best prices. Specifically, we assume that the distances to the best prices are chosen according to
a sequence of i.i.d.random variables $\left( \pi_{i} \right)_{i=0}^\infty$
where each $\pi_i$ is of the form:
\begin{equation}
  \pi_{i} = \left( \pi^{\bC_{b}}_{i}, \pi^{\bC_{a}}_{i}, \pi^{\bP_{b}}_{i},
    \pi^{\bP_{a}}_{i}, \pi_i^{{\bN_{b}}}, \pi_i^{{\bN_{a}}} \right).
\end{equation}
The entries take values in an interval $[-M,M]$,
for some $M >0$; positive values indicate changes in the visible book while
negative values indicate changes in the shadow book. Superscripts indicate
event types and $`\bN$' stands for `noise'. For instance, $\pi_i^{\bC_{a}} \in
[0,\Delta x^n)$ means that if the $i$-th event is a passive order, then it triggers a ask-side cancellation at the top of the visible book while $\pi_i^{\bC_{a}} \in [-\Delta x^n,0)$ corresponds to an ask-side cancellation one tick below the best ask, i.e. a cancellation in the shadow book. The precise meaning of the entries will become clear in \eqref{eq:density-dynamics-d1} below.

For $r=a,b$ passive order sizes are described by a sequence of
i.i.d.~random variables $\left( \omega_{i} \right)_{i=0}^\infty$ where each
$\omega_i$ is of the form
\begin{equation}
\label{variables}
	\omega_i = \left( \omega_{i}^{\bC_{b}}, \omega_{i}^{\bC_{a}},
          \omega_{i}^{\bP_{b}}, \omega_{i}^{\bP_{a}}, \omega^{\bN_{b}}_i,
          \omega^{\bN_{a}}_i \right).
\end{equation}	
The random variables $\omega_{i}^{\bP_{r}}$ take values in $[0,\infty)$;
they describe the {\sl sizes} of order placements. Likewise, the random
variables $\omega_{i}^{\bC_{r}}$ take values in $[0,1]$ and describe the
{\sl proportions} of cancellations. We notice that 
  $\omega_{i}^{\bC_{r}} = 1$ corresponds to a wipe-out of the orders at the
  corresponding price level that is, in principle, not forbidden. The resulting dynamics of the buy and sell side volume density functions satisfies: 
  \begin{equation}
  \label{eq:density-dynamics-d1}
  \begin{split}
    dv^n_r(t, \cdot) &=
    \bigg[ \indic{I^n\left( R^n(\ttau^n_{{\tN}^n(t-)}) +
        \pi^{\bP_r}_{N^n_r(t)} \right)}(\cdot) \omega^{\bP_r}_{N^n_r(t-)}
    \frac{\dv^n}{\dx^n}\\
    & \quad - \indic{I^n\left( R^n(\ttau^n_{{\tN}^n(t-)}) +
        \pi^{\bC_r}_{N^n_r(t)} \right)}(\cdot) \omega^{\bC_r}_{ N^n_r(t-)}
    v^n_r(\tau^n_{r,{N}^n_{r}(t-)}, \cdot) \frac{\dv^n}{\dx^n}\bigg] 
      \end{split}
\end{equation}
where $\dv^n$ is a scaling parameter that measures the impact of an individual
order on the state of the book and {$I^n(y)$ is the subinterval
corresponding to tick-size $\dx^n$ that $y$ belongs to, i.e.,}
\begin{equation}
  \label{eq:indicator-I}
  \indic{I^n(y)}(x) \coloneqq \sum_{j \in \Z} \indic{[x^n_j, x^n_{j+1}[}(y)
  \indic{[x^n_j, x^n_{j+1}[}(x).
\end{equation}
The specific structure of the dependence of the volume density functions on the bid and ask price as well as the random submission price levels reflects the fact that submission and cancellation price levels
are chosen relative to the best bid/ask price. 

\begin{remark} \label{rem-example}
In real-world markets only one event (market order arrival, cancellation,
placement) happens at a time. Within our framework this corresponds to the
special case where
% \[
% 	\pi^{\bC_{r}}_{i}=\pi^{\bP_{r}}_{i} \quad (r=a,b; i \in \mathbb{N})
% \]
% and
only one of the four random variables $\omega^{{\bC}_{b/a}},
\omega^{{\bP}_{b/a}}$ is different from zero. Our mathematical framework is
flexible enough to allow for such {a dependence structure}. 
\end{remark}

Within the framework described thus far, {(random)}
fluctuations in limiting volumes will originate entirely from fluctuations in
prices through the price-dependent order arrival and cancellation
dynamics.\footnote{Loosely speaking, the scaling of price is of CLT-type while the scaling of placements and cancellations is of LLN-type.} Our mathematical framework is flexible enough to also allow
fluctuations in volumes to originate directly from order placements if we
allow for a second type of placements that are correlated on a ``common factor'' rather than the price process.
In the simplest case the ``common factor'' dynamics is specified by sequences of i.i.d.~random variables $(\txi_{r,i})_{i=0}^\infty$
$(r=a,b)$. For the scaling limit it will be important that this common factor
changes at the same rate as prices do. To simplify the analysis, we assume
that it actually stays constant between two active order times and specify our
volume dynamics as:
  \begin{equation}
  \label{eq:density-dynamics-d}
  \begin{split}
     dv^n_r(t, \cdot) &=
    \bigg[ \indic{I^n\left( R^n(\ttau^n_{{\tN}^n(t-)}) +
        \pi^{\bP_r}_{N^n_r(t)} \right)}(\cdot) \omega^{\bP_r}_{N^n_r(t-)}
    \frac{\dv^n}{\dx^n}\\
    & \quad - \indic{I^n\left( R^n(\ttau^n_{{\tN}^n(t-)}) +
        \pi^{\bC_r}_{N^n_r(t)} \right)}(\cdot) \omega^{\bC_r}_{ N^n_r(t-)}
    v^n_r(\tau^n_{r,{N}^n_{r}(t-)}, \cdot) \frac{\dv^n}{\dx^n} \\
        & \quad + \indic{I^n\left( R^n(\ttau^n_{{\tN}^n(t-)}) +
            \pi^{\bN_r}_{N^n_r(t)} \right)}(\cdot)
        \omega^{\bN_r}_{N^n_r(t-)}\txi_{r, \tN^n(t-)}
      \sqrt{\dv^n} \bigg] dN^n_r(t). 
      \end{split}
\end{equation}
We notice that the common factor is modulated by the non-negative i.i.d.~noise
variables $\omega^{\bN_{r}}_i$ that change between two consecutive passive
orders. {We motivate the particular choice of the noise terms
  after the main result is formulated, below Corollary~\ref{cor:general}.}

We assume that the following condition holds. 
    
\begin{assumption}\label{assumption:scaling}\ 
\begin{itemize}
\item The random variables $\left( \pi_i^{\mathbf{T}} \right)_{{\bf
      T}=\bC_{r},\bP_{r}, \bN_{r},\,r=a,b}$, $i \in \N,$ are i.i.d.~with
  Lipschitz continuous densities $f^{\bf T}$ on some compact interval $[-M,M]$
  and independent of the Poisson processes.
\item The variables $\left( \omega_i^{\mathbf{T}} \right)_{{\bf
      T}=\bC_{r},\bP_{r}, \bN_{r},\,r=a,b}$, $i \in \N,$ are i.i.d.,
  independent of the Poisson processes and have a finite fourth moment. 
\item The variables $\widetilde \xi_{r,i}$ are i.i.d., independent of all
  other random variables and take the values $\pm 1$ with equal probability. 
\end{itemize}    
\end{assumption}    

For future use, we also introduce the filtration $\mathcal{F}^n$
generated by the $n$-th model. More precisely, we set
\begin{multline} \label{filtration}
  \mathcal{F}^n_t \coloneqq \sigma\biggl( \left(\tN^n_s\right)_{0\le s \le t},
  \left(\xi^n_{a,k}\right)_{k=1}^{\tN^n(t)},
  \left(\xi^n_{b,k}\right)_{k=1}^{\tN^n(t)}, \left(N^n_{a}(s)\right)_{0\le s
    \le t}, \left(N^n_{b}(s)\right)_{0\le s \le t}, \left(\omega^{\bC_a}_k,
    \omega^{\bP_a}_k, \omega^{\bN_a}_k \right)_{k=1}^{N^n_a(t)},\\
  \left(\omega^{\bC_b}_k,
    \omega^{\bP_b}_k, \omega^{\bN_b}_k \right)_{k=1}^{N^n_b(t)}, \left(\pi^{\bC_a}_k,
    \pi^{\bP_a}_k, \pi^{\bN_a}_k \right)_{k=1}^{N^n_a(t)}, \left(\pi^{\bC_b}_k,
    \pi^{\bP_b}_k, \pi^{\bN_b}_k
  \right)_{k=1}^{N_b^n(t)},\left(\txi^n_{b,k}\right)_{k=1}^{\tN^n_a(t)},
  \left(\txi^n_{a,k}\right)_{k=1}^{\tN^n_b(t)} \biggr).
\end{multline}

%%%%%%%%%%%%%%%%%%%%%%%%%%%%%%%%%%%%%%%%%%%%%%%%%
%%%%%%%%%%%%%%%%%%%%%%%%%%%%%%%%%%%%%%%%%%%%%%%%%
%%%%%%%%%%%%%%%%%%%%%%%%%%%%%%%%%%%%%%%%%%%%%%%%%

\subsection{The main result}
\label{sec:main-result}

We prove below that our LOB model converges to a  continuous time limit if
the order arrival rates tend to infinity and the impact of an individual order
arrival on the book as well as the tick size tends to zero in a particular
way. In order to make the convergence concept precise, and to state
the main result, we need to introduce further notation.
For $m\in (-\infty,\infty)$, we denote by $(H^{m},\|\cdot\|_m)$ the space of Bessel
potentials equipped with the usual Sobolev norm and inner product. Set
\[
\mathcal{E}' = \cup_m H^{-m} \supset \cdots \supset H^{-1} \supset L^2 \supset H^1 \supset \cdots \supset \cap_m H^{m} = \mathcal{E}.
\]
It is well known that $H^0 = L^2$ and that $\mathcal{E}$ is a complete
separable metric space. Sobolev's embedding theorem indicates that each
element of $\mathcal{E}$ is an infinitely differentiable function. In what
follows, denote the dual between $\mathcal{E}'$ and $\mathcal{E}$ by
$\langle\cdot,\,\cdot \rangle$, which is consistent with the inner product of
$H^0=L^2$.

The convergence concept we use is weak convergence in the Skorokhod space
$\mathcal{D}:=\mathcal{D}([0,\infty);\R^2\times H^{-1}\times H^{-1})$
of all c\`{a}dl\`{a}g functions on $[0,\infty)$ taking values in the space
$\R^2\times H^{-1} \times H^{-1}$.
The space $\mathcal{D}$ is equipped with the usual Skorokhod metric (see Jacod and Shiryaev \cite{JacodShiryaev2002}).

We are now ready to state the main result of this paper. {The main assumptions
and the assertions of the theorem are discussed below}. 
The proof is carried out in the subsequent sections. 

\begin{theorem} \label{thm-general} Let
  Assumptions~\ref{ass-initial}--\ref{assumption:scaling} be satisfied and assume that the scaling parameters $\lambda^n_{b/a}$ (arrival rate of passive orders),
$\mu^n$ (arrival rate of active orders), $\dv^n$ (order sizes) and
$\dx^n$ (tick size) satisfy the following conditions:
\[
	\lambda^n_{b/a}= n^{2}; \quad \mu^n=n; \quad \Delta v^n=n^{-2};
        \quad \Delta x^n=n^{-1/2}.
\]
Then  there
  are three independent Wiener processes $\widetilde{W}$, $W_a$ and $W_b$
  {($\widetilde{W}$ being two-dimensional)} such
  that the sequence $(A^n,B^n,v^n_a,v^n_b)$ of stochastic processes converges
  in distribution in $\mathcal{D}([0,\infty);\R^2\times H^{-1}\times H^{-1})$
  to $(A,B,v_a,v_b)$. Here $(A,B)$ is a two-dimensional diffusion process
  satisfying the SDE
 \begin{align*}
    dA_t=& b_a(B_t,A_t, Y^b_t, Y^a_t) dt + \sigma_a(B_t,A_t, Y^b_t, Y^a_t) d
    \widetilde{W}_t;\quad A_0=a_0;\\
    dB_t=& b_b(B_t,A_t, Y^b_t, Y^a_t) dt + \sigma_b(B_t,A_t, Y^b_t, Y^a_t) d
    \widetilde{W}_t; \quad B_0=b_0;
  \end{align*}
  with $\sigma_a=(\sigma^{11},\sigma^{12})$, $\sigma_b=(\sigma^{21},\sigma^{22})$, $Y^a_t = \ip{v_a(t,\cdot)}{\varphi^a(\cdot-A_t)}$ and $Y^b_t =
  \ip{v_b(t,\cdot)}{\varphi^b(\cdot-B_t)}$, respectively.
  Moreover, the volume density processes satisfy the infinite-dimensional SDE
  \begin{align*}
    v_b(t,\cdot) =& v_{b,0}(\cdot)+
    \int_0^t \left(E[\omega^{\bP_b}_{1}]f^{\bP_b}(\cdot-B_s)-E[\omega^{\bC_b}_1]
      f^{\bC_b}(\cdot-B_s)v_b(s,\cdot)\right)\,ds \\
    & +\sqrt{2} E\left[\omega^{\bN_b}_1\right]\int_0^t
    f^{\bN_b}(\cdot-B_s)\,dW_b(s),\quad t\geq 0;\\
    v_a(t,\cdot) =& v_{a,0}(\cdot)+
    \int_0^t \left(E[\omega^{\bP_a}_{1}]f^{\bP_a} (\cdot-A_s)-E[\omega^{\bC_a}_1]
      f^{\bC_a} (\cdot-A_s) v_a(s,\cdot)\right)\,ds \\
    & +\sqrt{2}E\left[\omega^{\bN_a}_1\right] \int_0^t
    f^{\bN_a}(\cdot-A_s)\,dW_a(s),\quad t\geq 0.
  \end{align*}
  If $\omega^{\bN_r}_1 = 0$ (no common factor), then the volume density functions are absolutely continuous in time. 
 \end{theorem}

 For any $T\in(0,\infty)$, existence and uniqueness of an
 adapted solution to the above coupled SDE system in
 $L^2(\Omega;C([0,T];\mathbb R^2)) \times L^2(\Omega;C([0,T];L^2(\mathbb R;\mathbb R^2)))$ is obvious; see \cite{DaPrato1992} for a
 general theory on stochastic equations in infinite dimensions. If the model
 parameters are sufficiently smooth, then the density functions are smooth as
 well. The following corollary is a consequence of the It\^{o}-Kunita formula.

\begin{corollary}
  \label{cor:general}
  If $v_{r,0}$ and the densities $f^{\bf T}$ belong to $H^m$ with $m>3$,
  then $v_{r}(t)$ take values in $H^m$ and hence by embedding, in
  $C^2(\R)$. Then the relative volume processes {$u_b(t,x) =
    v_b(t,B_t+x)$, $u_a(t,x) = v_a(t, A_t+x)$}
  satisfy the non-local stochastic partial differential equations
  \begin{align*}
  du_a(t,x)=&\left[
    E[\omega^{\bP_a}_1]f^{\bP_a}(x) -
    E[\omega^{\bC_a}_1]f^{\bC_a}(x)u_a(t,x)+D u_a(t,x)b_a(B_t,A_t,\langle u_b(t),\varphi^b\rangle,\langle u_a(t),\varphi^a\rangle)\right] dt\\
  &+
  \frac{1}{2}\text{tr}\left\{
  \sigma_a\sigma_a^\top(B_t,A_t,\langle u_b(t),\varphi^b\rangle,\langle u_a(t),\varphi^a\rangle)D^2u_a(t,x)\right\}\,dt
  +\sqrt{2} E\left[ \omega^{\bN_a}_1 \right]
  f^{\bN_a}(x)dW_a(t)\\
  &+D u_a(t,x)\sigma_a(B_t,A_t,\langle u_b(t),\varphi^b\rangle,\langle u_a(t),\varphi^a\rangle)\,d\widetilde{W}(s),\quad t\geq 0;\\
  u_a(0,x)=&v_{a,0}(x+a_0);\\
   du_b(t,x)=&\left[E[\omega_{1}^{\bP_b}]f^{\bP_b}(x)
   -E[\omega_{1}^{\bC_b}]f^{\bC_b}(x)u_b(t,x)+D u_b(t,x)b_b(B_t,A_t,\langle u_b(t),\varphi^b\rangle,\langle u_a(t),\varphi^a\rangle)
   \right]\,dt
   \\
   &+
  \frac{1}{2}\text{tr}\left\{
  \sigma_b\sigma_b^\top(B_t,A_t,\langle u_b(t),\varphi^b\rangle,\langle u_a(t),\varphi^a\rangle)D^2u_a(t,x)
  \right\}\,dt
  + \sqrt{2} E\left[ \omega^{\bN_b}_1 \right]
  f^{\bN_b}(x)dW_b(t)
  \\
  &+D u_b(t,x)\sigma_b(B_t,A_t,\langle u_b(t),\varphi^b\rangle,\langle u_a(t),\varphi^a\rangle)\,d\widetilde{W}(s),\ \,t\geq 0;\\
  u_b(0,x)=&v_{b,0}(x+b_0)
\end{align*}
which are coupled with the SDE for the price system given in Theorem~\ref{thm-general}.
\end{corollary}

Some comments on our scaling assumptions are in order. The assumption that market orders match precisely against the standing volume at the top of the book and that market orders of smaller size are viewed as cancellation is made for mathematical convenience. There is some empirical evidence, though, that this assumption is not too restrictive. In an empirical study the authors of \cite{Farmer} found that in their data sample around 85\% of the sell market orders that lead to price changes match exactly the size of the volume standing at the best bid price.

The assumptions that market orders and limit order placements in the spread occur at the same rate and that the liquidity at the top of the book is an indicator for volumes placed in the spread are key to our analysis. They await empirical verification.\footnote{To the best of our knowledge spread placement dynamics have not yet been extensively investigated in the financial econometrics literature. In any case, the assumption that market orders and limit order placements in the spread occur at the same rate implies that orders that are placed in the spread and almost immediately canceled (``ping orders'') are not allowed in our model as they do not really provide liquidity.} 

The assumption that $\frac{\mu^n}{\lambda^n} \to 0$ has also been made in
\cite{Horst-Kreher, Horst-Paulsen}. It states that passive events happen much
more frequently than active ones. There is strong empirical evidence
supporting the assumption that spread placements . For instance, Figure 3
(left) shows the intraday evolution of the proportion of spread placements
among all orders for all NASDAQ traded stocks for the month of March 2016. The
proportions of spread placements is particularly low for very liquid stocks
such as APPL,  MSFT or BAC; see \cite{Horst-Paulsen} and references
therein. Moreover, it is well known, that many spread placements have very
short lifetimes. As an example, Figure 3 (right) displays the cumulative
distribution function of the time to cancellation of spread placements for
APPL (consolidated NASDAQ data; March 2016). As we can see, more than 60\% of
all spread placements are cancelled after less than 5
milli-seconds\footnote{We thank Michael No\'{e} for the data analysis and
  Nikolaus Hautsch for data provision.}. Of  course, our model can not
reasonably account for such ping-orders.

\begin{figure}[h]\label{fig2}
\begin{minipage}{0.49\textwidth}
	\centering
	\includegraphics[width=6.5cm]{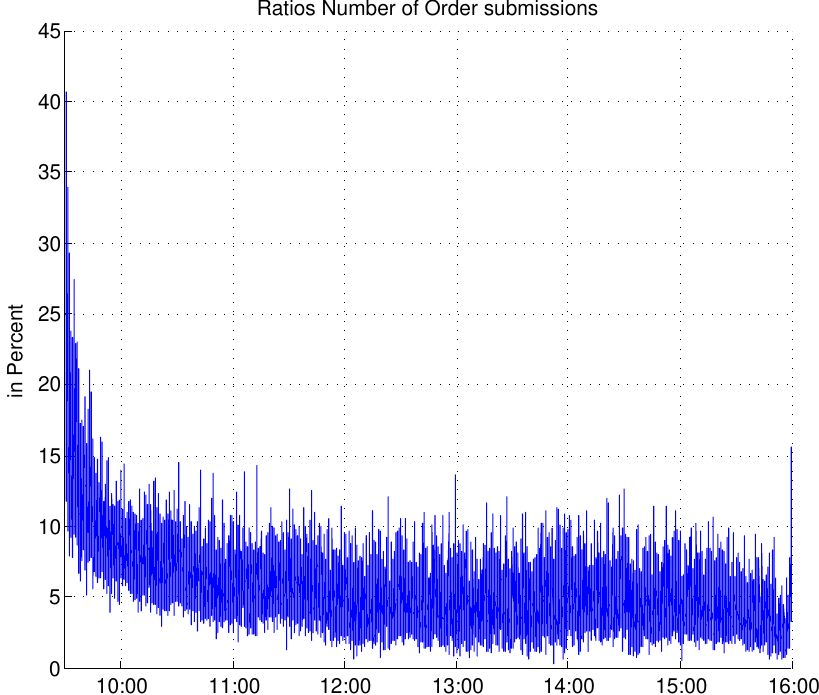}
\end{minipage}
\begin{minipage}{0.49\textwidth}
	\centering
	\includegraphics[width=6.5cm]{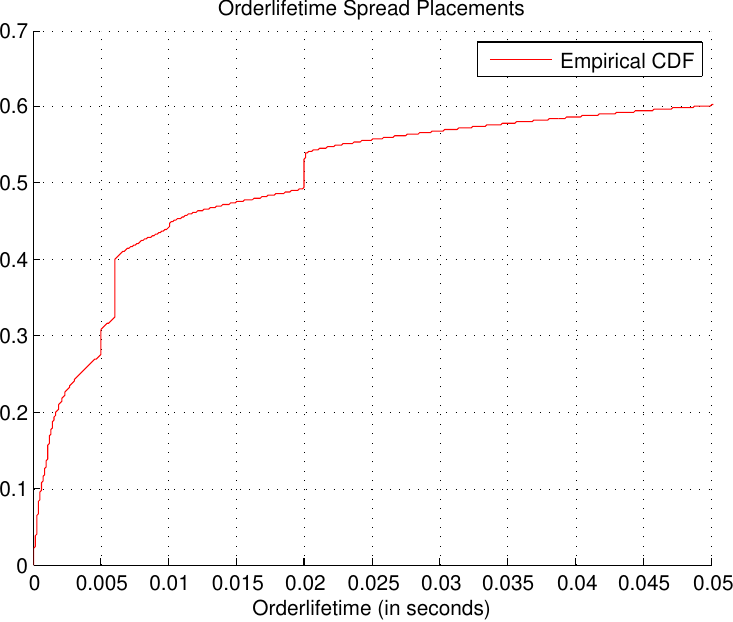}
\end{minipage}
\caption{Average percentage of spread placement per second (left) and empirical lifetime distribution of spread placements for APPL (right).}
\end{figure}

The scaling
assumptions $\mu^n \sim (\Delta x^n)^{-1}$ as in
\cite{Horst-Kreher,Horst-Paulsen}, respectively our assumption $\sqrt{\mu^n}
\sim (\Delta x^n)^{-1}$ are standard to obtain an ODE, respectively, diffusion
approximations of the price process. The assumption $\lambda^n_r \sim \Delta
v^n$ is as in \cite{Horst-Kreher, Horst-Paulsen}. It guarantees that the order
of magnitude of aggregate placements and cancellations over a given period of
time does not change with the model index $n$. Furthermore, from  the proof of
Lemma~\ref{lem-limit-L2-Vn12} we see that our proof requires $\mu^n \sim
\lambda^n (\dx^n)^2$. This is again the same condition as in
\cite{Horst-Kreher, Horst-Paulsen} taking into account that the different
scaling of the tick size (with our choice of rates, $\Delta p^n = \dx^n =
n^{-1}$ in \cite{Horst-Kreher, Horst-Paulsen}). Altogether, this explains the
absolutely continuous part of the limiting volume process; it describes the
expected volume placement and cancellation activity; see \cite{Horst-Kreher,
  Horst-Paulsen} for details. Summarizing, the absolutely
continuous part requires the scaling conditions
\[
	\mu^n \sim \Delta x^n, \quad \lambda^n_r \sim \Delta v^n, \quad \mu^n \sim \lambda_r^n (\dx^n)^{2}.
\]
The specific choice $\mu^n \sim n, \lambda_r^n \sim n^2$ and $\dx^n \sim n^{-1/2}$ was made for notational convenience. 

The diffusion part of the limiting volume density function is a direct consequence of the noise term
$\widetilde{\xi}^n_{r,i}$ in~\eqref{eq:density-dynamics-d} that does not change in-between price changes. The intuition is that in between two consecutive price changes a law of large numbers applies to the volume density function whose increment can hence be approximated by its expected value plus a random term of order $\sqrt{\dv^n}$ that translates into a Brownian motion as $n \to \infty$. If the scaling constant $\sqrt{\dv^n}$ is replaced by a smaller one, then the dynamics of the limiting volume density function will take the form of an (infinite-dimensional) ODE in a random environment generated by the price process. The SPDE dynamics of the volume process in relative coordinates is a direct consequence of the diffusive limiting price process and does not depend on the scaling of the noise terms. 

The requirement of a single common factor driving the noise along all passive events can easily be relaxed. 
Indeed, suppose that we have finitely or infinitely many factors with weights depending on the location of
 the passive event. This would result in limiting dynamics of the same form
 as above, except that the single driving Brownian motions were to be replaced by sums of the form $\sum_{i} e^i_{a/b}(\cdot) dW^i_{a/b}(s)$. As long as the
 (coloured) noise $\sum_i e^i_{a/b}(x) dW^{i}_{a/b}(s)$ exists in a suitable
 space of square integrable smooth functions (in $x$), the analysis should
 essentially stay the same. 
 
%%%%%%%%%%%%%%%%%%%%%%%%%%%%%%%%%%%%%%%%%%%%%%%%%

\subsection{Outline of the proof}
\label{sec:outline-proof}

The proof of Theorem~\ref{thm-general} is
carried out in the following sections. The main challenge is convergence, especially tightness of the volume densities.
Since the price process is C-tight by construction tightness of the volume process implies tightness of the 
price-volume process and hence existence of an accumulation point. 

We split the dynamics of the volume density functions
into the three processes $V^{n,i}_r(t,\cdot)$ $(i=1,2,3)$ that we are going to handle
separately, before finally pasting them back together to obtain the limiting dynamics. 
From equation (\ref{eq:density-dynamics-d})
we identify the following three processes which drive the evolution of the
%bid-side
 volume density function $(r=a,b)$:
\begin{subequations}
\label{eq:Vi-def}
\begin{gather}
  \label{eq:V1-def}
  V^{n,1}_r(t,x) = \sum_{i=1}^{N_r^n(t)}
  \indic{I^n\left(R^n(\ttau^n_{\tN^n({\tau_{r,i}^n})})+\pi^{\bP_r}_{i}\right)}(x)
  \omega_{i}^{\bP_r} \frac{\dv^n}{\dx^n},\\
  \label{eq:V2-def}
  V^{n,2}_r(t,x) = \sum_{i=1}^{N_r^n(t)} \indic{I^n\left(
      R^n(\ttau^n_{\tN^n({\tau_{r,i}^n})})+\pi^{\bC_r}_{i} \right)}(x)
  \omega_{i}^{\bC_r}
  \frac{\dv^n}{\dx^n},\\
  \label{eq:V3-def}
  V^{n,3}_r(t,x)=\sum_{i=1}^{N_r^n(t)}
  \indic{I^n\left( R^n(\ttau^n_{\tN^n({\tau_{r,i}^n})})+\pi_{i}^{\bN_r} \right)}(x)
  \omega_i^{\bN_r} \txi_{r,\tN^n({\tau_{r,i}^n})+1} \sqrt{\dv^n},
\end{gather}
\end{subequations}
corresponding to the volume changes due to incoming order placements
($V^{n,1}_r$), the proportional cancellations of standing volume ($V^{n,2}_r$)
and aggregated random fluctuations ($V^{3,n}_r$).  In the limit the increasing
functions (in time) $V^{n,1}_r$ and $V^{n,2}_r$ will translate into the
integrals w.r.t. the functions $f^{\bP_r}$ and $f^{\bC_r}$. The process
$V^{n,3}_r$ will contribute the martingale part.  \footnote{Note that
  $V^{n,3}_b$ itself is {\sl not} a martingale (in the filtration
  $\mathcal{F}^n$ generated by the full model), as the fluctuations $\txi$ are
  constant between two active order times.}

Unfortunately, these processes are not convenient for characterizing the
limit process. They are not Markov chains, and $V^{n,3}_r$ is not a martingale. The `markovization' is achieved by registering changes to the order book only along active order times and by  
considering the process as if these times were deterministic. More precisely, we define time-changes together with their inverses by
\begin{align}
  \overline{\eta}^n_u \coloneqq &   \,\,\ttau^n_{\floor{nu}}, \quad u \in
  [0,\infty);\nonumber\\
  \label{eq:bareta-def}
  \eta^n_u \coloneqq & \,\inf \{t:\,t>0,\,\overline{\eta}^n_t>u\}-\frac{1}{n},\quad u\in
  [0,\infty).
\end{align}
and introduce the following processes:
\begin{subequations}
  \label{eq:bar-def}
  \begin{align}
    \label{eq:barA-def}
    \barA^n(u) &\coloneqq A^n_0+{\Delta x^n}\sum_{i=1}^{\lfloor nu\rfloor} \xi^n_{a,i} \\
      \label{eq:barB-def}
      \barB^n(u) &\coloneqq B^n_0+{\Delta x^n}\sum_{i=1}^{\lfloor nu\rfloor} \xi^n_{b,i}
      \\
      \label{eq:barV1-def}
      \barV^{n,1}_r(u,x) &\coloneqq \sum_{i=1}^{N_r^n(\ttau_{\lfloor
          nu\rfloor}^n )} \omega_{i}^{\bP_r} 		
      \indic{I^n\left(\bar R^n\left( \eta^n_{\tau_{r,i}^n}
          \right)+\pi^{\bP_r}_{i}\right)}(x) \frac{\dv^n}{\dx^n},\\
      \label{eq:barV2-def}
      \barV^{n,2}_r(u,x) &\coloneqq \sum_{i=1}^{N_r^n(\ttau_{\lfloor
          nu\rfloor}^n )} \omega_{i}^{\bC_r} \indic{I^n \left(
          \bar R^n\left(\eta^n_{\tau_{r,i}^n}\right) +
          \pi^{\bC_r}_{i}\right)}(x) \frac{\dv^n}{\dx^n},\\
      \label{eq:barV3-def}
      \barV^{n,3}_r(u,x) &\coloneqq \sum_{i=1}^{N_r^n(\ttau_{\lfloor
          nu\rfloor}^n )} \omega_i^{\bN_r}  \indic{I^n \left(
          \bar R^n\left(\eta^n_{\tau_{r,i}^n}\right) + \pi_{i}^{\bN_r}
        \right)}(x) \txi_{r,\tN^n({\tau_{r,i}^n})+1}
      \sqrt{\dv^n},\\
      \label{eq:barv-def}
      \barv^{n}_r(u,x) &\coloneqq v^n_r(0,x) + \barV^{n,1}_r(u,x) +
      \barV^{n,3}_{r}(u,x)\\
      &\quad - \sum_{i=1}^{N_r^n(\ttau_{\lfloor nu\rfloor}^n )} \omega^{\bC_r}_{i}
      \indic{I^n \left(\bar R^n \left( \eta^n_{\tau_{r,i}^n} \right) +
          \pi_{i}^{\bC_r}\right)}(x) v^n_r(\tau_{r,i}^n,x)\frac{\dv^n}{\dx^n},
      \nonumber
  \end{align}
\end{subequations}

In a first step we prove in Section~\ref{sec:tightness-barv} tightness of each of the processes $\barV^{n,i}_{r}$ and of
$\barv^{n}_{r}$ in the distributional sense indicated above. For this part,
we heavily rely on Mitoma's theorem (Theorem~\ref{thr:mitoma}) together with
Kurtz's criterion (Theorem~\ref{thr:aldous}).   
Extending the tightness result from $\barv^n_{r}$ to $v^n_{r}$ requires $C$-tightness of
$\barv^n_{r}$. Hence, in Section~\ref{sec:char-limit}, we first
characterize the limit $\barv_{r}$ of $\barv^n_{r}$, depending on the yet
  unknown limiting price process $(A,B)$. Convergence of the placement term
is standard; convergence of the martingale term follows from a general result
on the convergence of stochastic process limits, given in Appendix A. The
challenge is to prove convergence of {\sl aggregate}
cancellations\footnote{The process $\barV^{n,2}_{r}$ only describes the
  proportionality of cancellation but not the actual volumes.}. In
Section~\ref{sec:limit-volume-density} we extend our tightness result to the
process 
\[
	\widehat v^n_r := \bar v^n_r \circ \eta^n
\]
that accounts for the random event times. 
As a byproduct we obtain that the limits of all the processes $\barv^n_{r}$, $\hatv^n_{r}$ and $v^n_{r}$ coincide. 
More precisely, we first use $C$-tightness of the sequence
$\barv^n_{r}$ to establish the joint convergence $\left( \barv^n_{r}\, ,\, \eta^n \right) \xrightarrow{n\to \infty} \left(
    \barv_{r} \, , \, \operatorname{id} \right)$ (in a weak sense). 
By Lemma \ref{lem:billingsley}, this
implies that
\begin{equation*}
  \lim_{n \to \infty} \hatv^n_{r} = \lim_{n\to\infty} \barv^n_{r} \circ
  (\eta^n) = \barv_{r}.
\end{equation*}
Subsequently we prove the tightness of $v^n_{r}$ and further verify that $\hatv^n_{r} - v^n_{r}$ converges
to $0$ in an $L^2(\Omega; L^2(\R))$-sense (this is where we need $\hatv^n_{r}$), thereby implying that
\begin{equation*}
  \lim_{n \to \infty} v^n_{r} = \lim_{n\to\infty} \hatv^n_{r} = \barv_{r}.
\end{equation*}
At this stage, we have only treated the convergence of each of the individual
sequences of processes $(A^n, B^n, v^n_{b})$ and $(A^n, B^n, v^n_{a})$ to some
limiting processes. However, as all these limiting processes are actually
continuous, joint tightness and, finally, joint weak convergence of $\left(
  A^n, B^n, v^n_b, v^n_a \right)$ follows by
Corollary~\ref{cor:C-tight-tight}. The last step, performed in
Section~\ref{new-model}  is then to characterize the limit of the price
processes, and consequently, of the full model.

%%%%%%%%%%%%%%%%%%%%%%%%%%%%%
%%%%%%%%%%%%%%%%%%%%%%%%%%%%%
%%%%%%%%%%%%%%%%%%%%%%%%%%%%%

\section{The scaling limit of the volume density}
\label{sec:scaling-limit-volume}

In this section, we prove weak convergence in a distributional sense
of the volume density function. While we do not yet know at this point whether
there is a unique accumulation point of the sequence of processes $(A^n,
B^n)$, we \emph{do} know that there are such accumulation points and all these
points are processes, which are continuous in time, see
Lemma~\ref{lem:price-tightness}. By choosing a proper sub-sequence, we can,
therefore, assume that $(A^n, B^n)$ does converge to a continuous limiting
process $(A,B)$, and we will often do so in this section.

Throughout, we use the symbol $C$ for deterministic constants which may change
from occurrence to occurrence.

\subsection{Tightness of the auxiliary process $\boldsymbol{\barv^n_{r}}$}
\label{sec:tightness-barv}

We first prove tightness of the processes $\barv^n_{r}$. The arguments are
the same for the bid and ask side of the book. We shall therefore
drop the index indicating of the bid/ask side and write $R^n$ or $\bar R^n$ for the price process in what follows. Further, where appropriate we drop the index $n$ and denote the random location of any activity in the book simply by $\pi$ or $\pi_i$ and its size by $\omega$ or $\omega_i$, disregarding the type (placement, cancellation, noise). 

We start with an elementary auxiliary lemma on the distribution of a Poisson
process as seen from a second, independent Poisson process. The lemma will be
key to compute the distribution of passive order arrivals between two
consecutive active order times.

\begin{lemma}
  \label{lem:aux-poisson-1}%\textbf{[$\mu = \lambda_1, \lambda =\lambda_2$]}
  Let $N_1$ and $N_2$ be two independent Poisson processes with intensities
  $\lambda_1$ and $\lambda_2$, respectively. Moreover, let $T_i$, $i=1,
  \ldots$, denote the jump times of the Poisson process $N_1$.
    For any $\alpha = 1,2, \ldots$, the random variable $N_2(T_\alpha)$
    has a negative binomial (NB) distribution with parameters $r = \alpha$ and
    $p = \f{\lambda_2}{\lambda_1 + \lambda_2}$, i.e., we have
    \begin{equation*}
      P\left( N_2(T_\alpha) = l \right) = \binom{l + \alpha -1}{\alpha-1}
      \left( \f{\lambda_2}{\lambda_1 + \lambda_2} \right)^l
      \left(\f{\lambda_1}{\lambda_1+\lambda_2} \right)^\alpha, \quad l=0, 1,
      \ldots
    \end{equation*}
    In particular, the moment-generating function reads
    $$
    Ee^{t N_2(T_{\alpha})}
    =\bigg(\frac{1-p}{1-pe^t}\bigg)^{\alpha},\quad\textrm{for }t<-\log\, p,
    $$
    and
    {
    \begin{equation*}
      E\left[ \prod_{i=0}^{k-1} (N_2(T_\alpha)-i) \right] = \left(
        \prod_{i=0}^{k-1}(\alpha + i) \right) \frac{\lambda_2^k}{\lambda_1^k},
      \quad k=1, \ldots, 4.
    \end{equation*}
  }
\end{lemma}

In what follows we denote by $\barF^n$ the filtration generated by the processes
$\barV^{n,1/2/3}_r$ and $\barv^{n}_r$ ($r=a,b$). 

In the next two lemmas we provide $L^p$ estimates for the processes $\barV^{n,1/2}$ and $\barV^{n,3}$, respectively. The arguments for $\barV^{n,1}$ and $\barV^{n,2}$ are the same. The arguments for $\barV^{n,3}$ are similar. However, since the scaling for $\barV^{n,3}$ is much smaller we need to take advantage of the martingale-difference structure in order to avoid mixed terms. 

\begin{lemma}
  \label{lem:tightness-barV1/2}
  There is a constant $C > 0$ (independent of $n, s, t$) such that for any
  $0<s\leq t$ we have
  \begin{align*}
    E_{\barF^n_s}\left[ \norm{\barV^{n,1/2}(t,\cdot) -
        \barV^{n,1/2}(s, \cdot)}^2_{L^2} \right] &\le C\left( (t-s)^2 +
      \frac{|t-s|}{n} \right), \\
    \sup_{x \in \R} E_{\barF^n_s} \left[ \left( \barV^{n,1/2}(t,x) -
        \barV^{n,1/2}(s,x) \right)^2 \right] &\le C\left( (t-s)^2 +
      \frac{|t-s|}{n} \right),\\
      E_{\barF^n_s}\left[ \norm{\barV^{n,1/2}(t,\cdot) -
        \barV^{n,1/2}(s, \cdot)}^4_{L^4} \right] &\le C\left( (t-s)^4 +
        \frac{|t-s|^3}{n}
        +\frac{|t-s|^2}{n^2}+
      \frac{|t-s|}{n^3} \right),\\
      \sup_{x \in \R} E_{\barF^n_s} \left[ \left( \barV^{n,1/2}(t,x) -
        \barV^{n,1/2}(s,x) \right)^4 \right] &\le C\left( (t-s)^4 +
        \frac{|t-s|^3}{n}
        +\frac{|t-s|^2}{n^2}+
      \frac{|t-s|}{n^3} \right).
  \end{align*}
\end{lemma}
\begin{proof}
  We drop the superscripts. Without any loss of generality, we can choose
  $s=0$. Let $\alpha \coloneqq \floor{nt}$ and consider
  \begin{equation*}
    E\left[ \barV(t,x)^2 \right] = E\left[ \left(
        \sum_{i=1}^{N(\ttau_\alpha)} \indic{I\left(\bar R(\eta_{\tau_i} ) +
          \pi_i\right)}(x) \omega_i  \right)^2 \right] \left( \f{\dv}{\dx}
    \right)^2.
  \end{equation*}
   Using the fact that the
  random variables $\omega_i$ are i.i.d.~and independent of the Poisson processes,
%  all the other random terms above, 
  we get
  \begin{align*}
    E\left[ \barV(t,x)^2 \right] &= E\Biggl[ \sum_{i<j;i,
      j=1}^{N(\ttau_\alpha)}2 E_{\mathcal{F}_{\tau_i}\vee \sigma(\pi_i,\omega_i,\overline{R}_{\eta_{\tau_i}} ) } \left[  \omega_j \indic{I\left( \bar R( \eta_{\tau_j}  ) + \pi_j\right)}(x)\right]\omega_i
    \indic{I\left( \bar R( \eta_{\tau_i}  ) + \pi_i\right)}(x)
     +\\
    &\quad\quad + \sum_{i=1}^{N(\ttau_\alpha)}
    E \left[\omega_i^2\right] \indic{I\left( \bar R( \eta_{\tau_i}
         ) + \pi_i\right)}(x) \Biggr] \left( \f{\dv}{\dx} \right)^2 .
  \end{align*}
  {
  As the random variable $\pi$ has a density $f$ with support in $[-M,M]$, for
  any deterministic $y$ we can bound
  \begin{equation}
    E\left[ \indic{I\left( y + \pi_i\right)}(x) \right] = \sum_{j \in \Z}
    \indic{[x_j,x_{j+1}[}(x) \int_{x_j-y}^{x_{j+1}-y} f(z) dz
    \le
    \norm{f}_{L^\infty} \dx \indic{[y-M-\dx, y+M+\dx]}(x).\label{eq:density-bound}
  \end{equation}
  }
Conditioning on the $\sigma$-algebra generated by all sources of
  randomness \emph{except} $(\pi_i)_{i\in\N}$, these bounds enable us to
  estimate:
  \begin{align*}
    E\left[ \barV(t,x)^2 \right]
    \le & E\Biggl[ 2 E\left[\omega_1\right]^2
    	\norm{f}_{L^\infty}^2 \dx^2 \sum_{i<j;i, j = 1}^{N(\ttau_\alpha)}
    	\indic{[ \bar R(\eta_{\tau_i} ) - M-\dx,\,  \bar R(\eta_{\tau_i} )
         + M+\dx]}(x) +\\
    & + E\left[\omega_1^2\right] \norm{f}_{L^\infty} \dx \sum_{i =
      1}^{N(\ttau_\alpha)} \indic{\left[ \bar R(\eta_{\tau_i} )-M-\dx,\,
        \bar R(\eta_{\tau_i} )+M +\dx \right]}(x) \Biggr] \left(\f{\dv}{\dx}
    \right)^2.
  \end{align*}
  At this stage, we can easily bound $\barV$ both in $L^2(\R)$ and as a
  supremum in $x$. More precisely, we have
  \begin{align*}
    E\left[ \norm{\barV(t)}_{L^2}^2 \right] + \sup_{x \in \R} E\left[
      \barV(t,x)^2 \right] \le (4(M+\dx) + 1) \biggl( E[\omega_1]^2
      \norm{f}_{L^\infty}^2 \dx^2 E\left[ N(\ttau_\alpha) \left(
          N(\ttau_\alpha) - 1 \right) \right] \\
       + E\left[ \omega_1^2 \right]
      \norm{f}_{L^\infty} \dx E\left[ N(\ttau_\alpha) \right] \biggr) \left(
      \f{\dv}{\dx} \right)^2.
  \end{align*}
  Finally, inserting the moment formulas given in
  Lemma~\ref{lem:aux-poisson-1} and applying the trivial estimate $\alpha =
  \floor{nt} \le nt$ together with Assumption~\ref{assumption:scaling}, we
  arrive at
  \begin{align*}
    E\left[ \norm{\barV(t)}_{L^2}^2 \right] + \sup_{x \in \R} E\left[
      \barV(t,x)^2 \right] 
    & \le C n^{-7/2} \left\{ n^{-1/2}
      nt(1+nt) \f{n^4}{n^2} + nt \f{n^2}{n} \right\} \\ & = C \left( t^2 + (n^{-1}
      + n^{-3/2}) t \right) \le C\left(t^2 + \frac{t}{n}\right).
  \end{align*}

  The estimate for the fourth moment follows analogously and is therefore
  skipped.
\end{proof}

\begin{lemma}
  \label{lem:tightness-barV3}
  There is a constant $C$ (independent of $n$, $s$, $t$) such that for every
  $0 < s\leq t$
  \begin{align}
    E_{\barF^n_s}\left[ \sup_{s \le u \le t} \norm{\barV^{n,3}(u) -
        \barV^{n,3}(s)}^2_{L^2} \right]
        +
        \sup_{x\in\R}E_{\barF^n_s}\left[ \sup_{s \le u \le t} \left|{\barV^{n,3}(u,x) -
        \barV^{n,3}(s,x)}\right|^2 \right]
        &\le C|t-s|,\label{est:diffu-L2}\\
     E_{\barF^n_s}\left[ \sup_{s \le u \le t} \norm{\barV^{n,3}(u) -
        \barV^{n,3}(s)}^4_{L^4} \right]
        +
        \sup_{x\in\R}E_{\barF^n_s}\left[ \sup_{s \le u \le t} \left|{\barV^{n,3}(u,x) -
        \barV^{n,3}(s,x)}\right|^4 \right]
        &\le C\left(  (t-s)^2+\frac{|t-s|}{n} \right).
  \end{align}
\end{lemma}
\begin{proof}
  Again, we restrict ourselves to proving the case $s=0$ and drop all
  superscripts from the notation. Re-writing $\barV$ in a form more clearly
  expressing its martingale structure, we consider
  \begin{equation*}
    \barV(t) = \sum_{i=1}^{N(\ttau_\alpha)} \indic{I\left(\bar R(\eta^n_{\tau^n_i})
        + \pi_i \right)}(x) \omega_i \txi_{\tN(\tau_i)} \sqrt{\dv} =
    \sum_{j=0}^{\alpha-1} \sum_{i=N(\ttau_j)+1}^{N(\ttau_{j+1})}
    \indic{I\left(\bar R(j/n) + \pi_i \right)}(x) \omega_i
    \txi_{j} \sqrt{\dv},
  \end{equation*}
  where we again use the short-hand notation $\alpha = \floor{tn}$.
  Using Doob's inequality and the fact that $E\left[\txi_i \txi_{j} \right] =
  \delta_{ij}$ with $\txi_i^2 = 1$, we have
  \begin{align*}
    E\left[ \sup_{0\le u\le t} |{\barV(u,x)}|^2 \right] &\le 4 E\left[
      |{\barV(t,x)}|^2 \right] \\
    &= 4 \dv  E\left[ \left( \sum_{j=0}^{\alpha-1} \txi_j
        \sum_{i=N(\ttau_j)+1}^{N(\ttau_{j+1})} \indic{I\left(\bar R(j/n)
            + \pi_i \right)}(x) \omega_i \right)^2 \right] \\
    &= 4 \dv  E\left[ \sum_{j=0}^{\alpha-1} \left(
        \sum_{i=N(\ttau_j)+1}^{N(\ttau_{j+1})} \indic{I\left(\bar R(j/n)
            + \pi_i \right)}(x) \omega_i \right)^2 \right].
  \end{align*}
Next, we estimate the contribution of the random
  locations $\pi$ as in~\eqref{eq:density-bound}.
  We have:
  \begin{align*}
    E\left[ \sup_{0\le u\le t} \norm{\barV(u)}^2_{L^2} \right]
    &\le 4 \dv
    E\Biggl[ \sum_{j=0}^{\alpha-1} \Biggl\{ \sum_{i\neq i' =
      N(\ttau_j)+1}^{N(\ttau_{j+1})} \omega_i \omega_{i'} \int_{\R}
     \indic{I\left(\bar R(j/n) + \pi_i \right)}(x)
      \indic{I\left(\bar R(j/n) + \pi_{i'} \right)}(x)\,dx  \Biggr\}
    \Biggr]  +\\
    &\quad \quad + \sum_{i = N(\ttau_j)+1}^{N(\ttau_{j+1})} \omega_i^2
    \int_{\R}  \indic{I\left( \bar R(j/n) + \pi_i
        \right)}(x) \,  dx\Biggr] \\
    &\le 4 \dv E\left[ \sum_{j=0}^{\alpha-1} \left\{ \sum_{i\neq i' =
          N(\ttau_j)+1}^{N(\ttau_{j+1})} E[\omega_1]^2
        \norm{f}_{L^\infty}^2 \dx^2 (2M) + \sum_{i =
          N(\ttau_j)+1}^{N(\ttau_{j+1})} E[\omega_1^2] \norm{f}_{L^\infty} \dx
        (2M) \right\} \right],
  \end{align*}
  and similarly,
  \begin{equation*}
  \sup_{x\in\R}E\left[ \sup_{0\le u\le t} |{\barV(u,x)}|^2 \right]
  \le 4  \dv E\left[ \sum_{j=0}^{\alpha-1} \left\{ \sum_{i\neq i' =
        N(\ttau_j)+1}^{N(\ttau_{j+1})} E[\omega_1]^2
      \norm{f}_{L^\infty}^2 \dx^2  + \sum_{i =
        N(\ttau_j)+1}^{N(\ttau_{j+1})} E[\omega_1^2] \norm{f}_{L^\infty} \dx
      \right\} \right].
  \end{equation*}
  Since the distribution of the increments
  $N(\ttau_{j+1})-N(\ttau_j)$
  does not depend on $j$, we see that
  \begin{multline*}
    E\left[ \sup_{0\le u\le t} \norm{\barV(u)}^2_{L^2} \right]
    +\sup_{x\in\R}E\left[ \sup_{0\le u\le t} |{\barV(u,x)}|^2 \right]\\
    \le \,C \,\dv
    E\left[ \alpha \left\{ E[\omega_1]^2 \norm{f}_{L^\infty}^2 (\dx)^2
        N(\ttau_1) \left(N(\ttau_1) - 1 \right) + E\left[\omega_1^2\right]
        \norm{f}_{L^\infty} \dx N(\ttau_1) \right\} \right].
  \end{multline*}
  Again appealing to Lemma~\ref{lem:aux-poisson-1} (with $\alpha = 1$)
  together with Assumption~\ref{assumption:scaling}, we obtain
  \begin{equation*}
    E\left[ \sup_{0\le u\le t} \norm{\barV(u)}^2_{L^2} \right]
    +\sup_{x\in\R}E\left[ \sup_{0\le u\le t} |{\barV(u,x)}|^2 \right]
     \le C
    \f{1}{n^2} nt \left\{ \f{2}{n} \f{n^4}{n^2} + \f{1}{\sqrt{n}} \f{n^2}{n}
    \right\} = Ct\{ 2 + 1/\sqrt{n} \} \le C t.
  \end{equation*}
  As in the proof of Lemma~\ref{lem:tightness-barV1/2}, the estimate for the
  fourth moment follows by the similar arguments.
\end{proof}

At this stage we can patch together the estimates in
Lemmas~\ref{lem:tightness-barV1/2} and~\ref{lem:tightness-barV3} to obtain a
similar one for the process $\barv^n$. The proof is based on an
event-by-event decomposition of the limit order book dynamics. More precisely,
in terms of the increments (again, we drop indices indicating the order book side)
\begin{gather*}
  h^{n,1}_{i}(x) \coloneqq \omega_{i}^{\bP} \indic{I^n\left(\bar R^n\left(
        \eta^n_{\tau_{i}^n} \right) + \pi^{\bP}_{i}\right)}(x)
  \frac{\dv^n}{\dx^n},\\
  h^{n,2}_{i}(x) \coloneqq \omega_{i}^{\bC} \indic{I^n \left(
      \bar R^n\left(\eta^n_{\tau_{i}^n}\right) +
      \pi^{\bC}_{i}\right)}(x) \frac{\dv^n}{\dx^n},\\
  h^{n,3}_{i}(x) \coloneqq \indic{I^n \left(
          \bar R^n\left(\eta^n_{\tau_{i}^n}\right) + \pi_{i}^{\bN}
        \right)}(x) \omega_i^{\bN} \txi_{a,\tN^n({\tau_{i}^n})+1}
      \sqrt{\dv^n}\\
\end{gather*}
of the processes $\barV^{n,j}$ $(j=1,2,3)$ one has the following generic decomposition,
\begin{multline}
  \label{eq:event-by-event}
  \barv^n(t,x) = \prod_{i=1}^{N^n(\ttau^n_{\floor{nt}})}
  \left(1 - h^{n,2}_{i}(x) \right) \barv^{n}(0,x) +\\
  + \prod_{i=1}^{N^n_{a/b}(\ttau^n_{\floor{nt}})} \left(1 -
    h^{n,2}_{i}(x) \right) \left[
    \sum_{i=1}^{N^n(\ttau^n_{\floor{nt}})} \f{1}{ \prod_{m=1}^{i} \left(1 -
    h^{n,2}_{m}(x) \right) } \left( h^{n,1}_{i}(x) +
  h^{n,3}_{i}(x) \right) \right].
\end{multline}

\begin{lemma}
  \label{lem:tightness-barv}
  There exists a sequence of non-negative adapted process $C_t^n$ and a
  deterministic constant $C$ such that for $p\in\{2,4\}$
  \begin{align*}
    E_{\barF^n_s}\left[\sup_{s\le r\le t} \norm{\barv^{n}(r) - \barv^n(s)}_{L^p}^p
    \right]
   & \le C_{s}^n \left( (t-s)^p + (t-s) \right),
    \\
    E \left[\sup_{ r\le t} \norm{\barv^{n}(r) }_{L^p}^p
    \right]
   & \le C \left( t^p + t +1\right),
   \\
    \sup_{x\in\R}  E \left[\sup_{ r\le t}    \left|   { \barv^{n}(r,x) }    \right|^p
    \right]
    &\le 
    C\left( t^p + t +1\right),
  \end{align*}
  with 
 \begin{align} 
 &C^n_s \le   C \left(  \norm{\barv^{n}(s) }_{L^4}^4  +     \norm{\barv^{n}(s) }_{L^2}^2  +1\right),  \nonumber\\
  &	\sup_n E\left[ \sup_{0 \le s \le t} C^n_{s} \right] 
	\le C (t^4+t+1). \label{estimate}
\end{align}
\end{lemma}

\begin{proof}
  We may again drop the dependence on $n$ from the notation and
  w.l.o.g.~assume $s = 0$. Note that $0 \le 1-h_i^{2}(x) \le 1$ and
  \begin{equation*}
    \abs{ \prod_{i=1}^{N(\ttau_{\floor{nt}})} (1-h^2_i(x)) - 1} \le
    \sum_{i=1}^{N(\ttau_{\floor{nt}})} h_i^2(x) = \barV^2(t,x).
  \end{equation*}
  Hence, \eqref{eq:event-by-event} together with
  Lemma~\ref{lem:tightness-barV1/2} and~\ref{lem:tightness-barV3} implies that
  for $p\in \{2,4\}$,
  \begin{align}
    &E\left[ \left|{\barv(t,x) - \barv(0,x)}\right|^p\right] \\
    &=  E
    \left[
    \left| {
        \left( \prod_{i=1}^{N(\ttau_{\floor{nt}})}
          \left(1 - h^{2}_{i} \right) -1 \right) \barv(0,x) +
        \prod_{i=1}^{N(\ttau_{\floor{nt}})} \left(1 -
          h^{2}_{i} \right) \left(
          \sum_{i=1}^{N(\ttau_{\floor{nt}})} \f{1}{ \prod_{m=1}^{i} \left(1 -
              h^{2}_{i} \right) } \left( h^{1}_{i} +
            h^{3}_{i} \right) \right)}\right|^p 
            \right]
            \nonumber\\
    &\le C \left\{\left| {\barv(0,x)}\right|^p %\sup_{x \in \R}
       \sup_{x\in\mathbb R}E \left[ \left(\barV^2(t,x)\right)^p \right] + E \left[ \left| {\barV^1(t,x)}\right|^p +
      \sup_{0 \le s \le t} \left|{\barV^3(s,x)}\right|^p \right]\right\} . \label{ineq-lem-lem:tightness-barv}
  \end{align}
  It follows for $p\in\{2,4\}$ that,
  \begin{align*}
%    \sup_{x\in\R}E\left[\sup_{0\le u\le t} \left|{\barv(u,x) -
 %        \barv(0,x)}\right|^p\right]
 %   +
    &E\left[\sup_{0\le u\le t} \norm{\barv(u) - \barv(0)}^p_{L^p}\right]
    \leq
    C\left( 
     %\sup_{x\in\R} \left|{\barv(0,x)}\right|^p
    %+
    \norm{\barv(0)}^p_{L^p} +1   \right)(t^{p}+t).
  \end{align*}
  For a general $s \in [0,t]$, this proves the estimate for a
  $\barF^n_s$-measurable random variable $C_s^n$ that depends in an affine way
  on $\norm{\barv(s)}^p_{L^p}+\norm{\barv(s)}^2_{L^2}$. Note, however, that 
  %the above estimate also implies 
  it follows in a similar way that 
  for $p\in\{2,4\}$
  \begin{equation*}
    \sup_{n\in\N^+} \left(
    E\left[ \sup_{0 \le s \le t} \norm{\barv^n(s)}_{L^p}^p
    \right]
    + \sup_{x\in\R}E\left[ \sup_{0 \le s \le t} \left|{\barv^n(s,x)}\right|^p
    \right]\right)< C(t^{4}+t+1),
  \end{equation*}
  so that we can, indeed, find a deterministic constant $C$ which is
  independent of $s$, $t$ and $n$ and bounds $E\left[ \sup_{0 \le s \le t}
    C^n_{s} \right] \le C(t^4+t+1)$.
\end{proof}

\begin{remark}\label{rmk-difference-vak}
  Using the same arguments as in the above proof, we obtain for $p\in\{2,4\}$ and
  $k=0,1,2,\cdots$,
  \begin{equation*}
    E\left[\sup_{i\in [N^n(\ttau^n_{k}),\,N^n(\ttau^n_{k+1})]\cap\N^+}
    \|v^n(\tau_{a,i})-v^n(\eta_k^n)\|_{L^p}^p\right]
    %+\sup_{x\in\R} E\left[\sup_{i\in [N^n(\ttau^n_{k}),\,N^n(\ttau^n_{k+1})]\cap\N^+}
    %\left|v^n(\tau_{i},x)-v^n(\ttau_k^n,x)\right|^p\right]
    \leq C \f{t+t^p}{n},
  \end{equation*}
  where the constant $C$ is independent of $n$, $k$ and $t$.
\end{remark}

We are now ready to state and prove the main result of this section.

\begin{proposition}
  \label{prop:tightness-barv}
  The processes $\barv^n_{r}$ and $\barV^{n,i}_{r}$ $(r=a,b; i=1,2,3)$ are tight as processes with paths in
  $\mathcal{D}\left([0,\infty); H^{-1} \right)$.
\end{proposition}
\begin{proof}

  Let $X^n\in \left\{ \barv^n_{r},
  \barV^{n,1}_{r}, \barV^{n,2}_{r}, \barV^{n,3}_{r} \right\}$. 
%  By definition, the tightness of $X^n$ is equivalent to that of $((1+t)^{-1}X^n(t))_{t\in [0,\infty)}$ that we denote again by $X^n$. {The reason why we scale the processes this way is % that estimate (\ref{estimate}) prevents us from applying Theorem \ref{thr:aldous} directly to the original processes}.
By Mitoma's theorem (see Theorem~\ref{thr:mitoma}), we need to prove
  tightness of the processes $\ip{X^n}{\phi}$ for any test function $\phi \in
  \mathcal{E} \subset L^2(\R)$, for which we,
  in turn, will appeal to Kurtz's criterion (see
  Theorem~\ref{thr:aldous}). Hence, we need to estimate
  \begin{equation*}
    E_{\F^n_s}\left[ \left| \ip{X^n(t)-X^n(s)}{\phi}\right|^2 \right].
  \end{equation*}
 % where $\rho(a,b)=|a-b|\wedge 1$ for any $a, b\in \mathbb R$, which is a bounded metric on $\mathbb R$.
  As $X^n$ takes values in $L^2$, the bracket
  $\ip{X^n}{\phi}$ is equal to the $L^2$ inner product
  $\ip{X^n}{\phi}_{L^2}$. By  Lemmas~\ref{lem:tightness-barV1/2}, \ref{lem:tightness-barV3}
  and~\ref{lem:tightness-barv}, for each $T > 0$ and $0\leq s < t\leq T$, 
  \begin{align*}
    E_{\F^n_s}\left[ \ip{X^n(t)-X^n(s)}{\phi}^2 \right] & \le E_{\F^n_s}\left[
      \norm{X^n(t) - X^n(s)}_{L^2}^2 \right] \norm{\phi}_{L^2}^2 \\
      & \le C^n_s \left[(t-s)^2 + (t-s)\right]\norm{\phi}_{L^2}^2 \\
       %& \le E_{\F^n_s} \left[\sup_{0 \leq \tau \leq T} C^n_{\tau} \right] \left[(t-s)^2 + (t-s)\right]\norm{\phi}_{L^2}^2
  \end{align*}
  for some sequence of adapted processes $C^n_t$ with
  \[
  	\sup_n E\left[ \sup_{0 \le \tau \le T} C^n_{\tau} \right] < \infty.
  \]
Hence, the second condition of Theorem \ref{thr:aldous} follows with
  $\gamma_n(\delta) = \sup_{\tau \in [0,T]} C^n_\tau (\delta^2 + \delta)$. 
 The first condition, tightness of the sequence of random variables
  $\ip{X^n(t)}{\phi}$ for each rational $t$, 
  follows from uniform boundedness of the sequence of random variables
  $\ip{X^n(t)}{\phi}$ in $L^2(\Omega, \F, P)$.
  
  Furthermore, again by Lemmas~\ref{lem:tightness-barV1/2}, \ref{lem:tightness-barV3}
  and~\ref{lem:tightness-barv},
  $$
  \sup_n E \left[\sup_{t\in[0,T]}\|X^n(t)\|_{L^2}^2 \right] \leq C (T+T^{2}),
  $$
  for some constant $C$ that is independent of $n$ and $T$. As a result, it follows from the Markov inequality that 
  $$
  \sup_n P\left(
  \sup_{t\in[0,T]}
  \|X^n(t)\|_{L^2}^2>N
  \right)\leq \frac{C(T+T^2)}{N}\rightarrow 0,\quad \textrm{as }N\rightarrow\infty.
  $$
  Thus, by Mitoma's theorem $X^n$ is tight as sequences of processes with paths in
  $\mathcal{D}\left([0,\infty); H^{-1} \right)$.
\end{proof}
\begin{remark}
  The preceding proof \emph{almost} gives us tightness in $\mathcal{D}\left([0,\infty);
    L^2(\R) \right)$ for $L^2(\R)$ equipped with the weak topology. Note,
  however, that $L^2(\R)$ is not a metric space when equipped with the weak
  topology. Hence we cannot use Kurtz's criterion as it does not apply to
  non-metric state spaces.
\end{remark}

%%%%%%%%%%%%%%%%%%%%%%%%%%%%%%%%%%%
%%%%%%%%%%%%%%%%%%%%%%%%%%%%%%%%%%%
%%%%%%%%%%%%%%%%%%%%%%%%%%%%%%%%%%%

\subsection{Characterization of the limit of $\boldsymbol{\barv^n_{r}}$}
\label{sec:char-limit}

In this section, we characterize the limit of the sequence $\barv^n_{r}$. 
Again, we drop indices where appropriate. 
We start with establishing joint convergence in
distribution of bid/ask prices along with the aggregate fluctuations of
standing volumes on one side of the book.

\begin{proposition}\label{prop-volume-limit-V3}
    For $r=a,b$, $(\barA^n,\barB^n,\barV^{n,3}_r)\Rightarrow (A,B,\barV^3_r)$,
    with $(A,B)$ being a two-dimensional continuous process, and for any
    choice $\phi_1, \ldots, \phi_l \in \mathcal{E}$ the $l$-dimensional
    process $\left( \ip{V^3_r}{\phi_1}, \ldots, \ip{V^3_r}{\phi_l} \right)$ is
    a martingale w.r.t.~the filtration generated by the process
    $(A,B,\barV^3_r)$ with quadratic co-variation
    \begin{gather*}
      \sip{\ip{V^3_r}{\phi_i}}{\ip{V^3_r}{\phi_j}}_t = \int_0^t
      \sigma(\phi_i)(R_s) \sigma(\phi_j)(R_s) ds, \quad t \ge 0,\ 1, \le i,j
      \le l,\\
      \sigma(\phi)(y) \coloneqq \sqrt{2} E[\omega_1^{\mathbf{N}}] \int_{\R}
      f^{\mathbf{N}_r}(x - y) \phi(x) dx.
    \end{gather*}
\end{proposition}

\begin{proof}
  Combining Proposition \ref{prop:tightness-barv}, Corollary
  \ref{cor:C-tight-tight} and C-tightness of the price process
  (Lemma~\ref{lem:price-tightness}), we conclude that
  $(\barA^n,\barB^n,\barV^{n,3})$ is tight as a sequence of processes with
  sample paths in $\mathcal{D}([0,\infty);\mathbb{R}^2\times H^{-1})$ and that
  $(\barA^n,\barB^n)$ converges in distribution to a two-dimensional
  continuous process $(A,B)$ along a sub-sequence.

  Since the sequence of price processes is C-tight and converges to $(A,B)$ it
  is sufficient to characterize the weak accumulation point $\barV^3$.  To
  this end, we assume w.l.o.g. that 
  $E[\omega_1^{\bN}] > 0$. We now proceed in several steps.

i) First, we define, for any $\phi\in \mathcal{E}$,
  \begin{align*}
    \barY^n_t(\phi)=\langle \phi,\,\barV^{n,3}(t)\rangle, \quad t\in[0,\infty) ,
  \end{align*}
  and denote by $\mathcal {G}^n$ the filtration generated by the processes
  $\big( \barA^n_t,\barB^n_t,\barV^{n,3}(t) \big)$.
  Note that the sequence
  $(\barA^n,\barB^n,\barY^{n}(\phi))$ converges in distribution to
  $(A,B,\bar Y(\phi))$ where $\bar Y(\phi) := \langle\phi,\,\barV^3\rangle$ as a sequence of processes
  whose sample paths belong to $\mathcal{D}(0,\infty;\mathbb{R}^3)$.

  We are now going to use Lemma~\ref{lem:C1} % together with
  % Corollary~\ref{cor:C11}
  % to show that $Y(\phi)$ admits a representation in
  % terms of a Brownian integral.
  {to verify the claimed form of the quadratic variation. For
    simplicity, we start with the special case $l = 1$.}
  For this, we assume that $\phi\geq 0$;
  otherwise, we make the decomposition $\phi=\phi^+ -\phi^-$ and consider
  $\phi^+$ and $\phi^-$ respectively, just noting that both $\phi^+$ and
  $\phi^-$ belong to $H^1$.  Let
  \begin{align*}
    a^n_0(\phi)(R) &\coloneqq \Big(\sum_j \int_{x_j^n}^{x_{j+1}^n}
    f(x-R)\,dx \int_{x_j^n}^{x_{j+1}^n} \phi(x)\,dx \Big)^2(\Delta
    x^n)^{-2} E\left[ \omega_1^{\bN} \right]^2,
    \\
    a^n_1(\phi)(R) &\coloneqq\sum_j \int_{x_j^n}^{x_{j+1}^n} f(x-R)\,dx ~
    \Big|\int_{x_j^n}^{x_{j+1}^n} \phi(x)\,dx \Big|^2(\Delta x^n)^{-2} E\left[
    (\omega_1^{\bN})^2 \right]
    \\
    \sigma^n(\phi)(R)
    &\coloneqq \left(2a^n_0(\phi)(R)+\frac{1}{n} a^n_1(\phi)(R) \right)^{1/2}.
  \end{align*}

  Note that for any deterministic $y$ and any random variable $\pi$ with
  density $f$,~\eqref{eq:indicator-I} implies
  \begin{align*}
    E\left[ \int \indic{I^n(y+\pi)}(x) \phi(x) dx \right] = & \sum_{j} \int
    \indic{[x_j^n,x_{j+1}^n[}(y+z) f(z) dz \int \indic{x_j^n,x_{j+1}^n[}(x) \phi(x)
    dx \\
    = & \sum_j \int \indic{[x^n_j,x^n_{j+1}[}(x) f(x-y) dx \int
    \indic{[x^n_j,x^n_{j+1}[}(x) \phi(x) dx.
  \end{align*}
  Since the number of passive order arrivals $\left( N^n_{\ttau^n_{k}} -
    N^n_{\ttau^n_{k-1}}\right)$ on $[\frac{k-1}{n}, \frac{k}{n})$ follows a
  negative binomial distribution $\mathrm{NB}\left(1,
    \f{\lambda^n}{\lambda^n+\mu^n} \right)$ (see Lemma
  \ref{lem:aux-poisson-1}), we have (using \eqref{eq:barV3-def}):
  \begin{align}
    &E_{\mathcal{G}_{\frac{k-1}{n}}^n}\left[
      |\barY^n_{\frac{k}{n}}(\phi)-\barY^n_{\frac{k-1}{n}}(\phi)|^2    \right]
    \nonumber\\
    &= \Delta v^n \Bigg\{
    E\left[\left( N^n_{\ttau^n_{k}}-N^n_{\ttau^n_{k-1}}\right)
    \left( N^n_{\ttau^n_{k}}-N^n_{\ttau^n_{k-1}}-1\right)\right]
    \bigg(\sum_j \int_{x_j^n}^{x_{j+1}^n} f(x-\bar R^n_{\f{k-1}{n}})\,dx
    \int_{x_j^n}^{x_{j+1}^n} \phi(x)\,dx  \bigg)^2 E\left[\omega_1^{\bN}
    \right]^2
    \nonumber\\
    &\ \,~ +
    E\left[\left( N^n_{\ttau^n_{k}} - N^n_{\ttau^n_{k-1}}\right)\right]
    \sum_j\int_{x_j^n}^{x_{j+1}^n} f(x-\bar R^n_{\f{k-1}{n}})\,dx
    \Big|\int_{x_j^n}^{x_{j+1}^n} \phi(x)\,dx  \Big|^2 E\left[\left(
        \omega_1^{\bN} \right)^2 \right]\Bigg\}
        \nonumber\\
    &=
    \Delta v^n	(\Delta x^n)^2
    \Bigg\{
    E\left[\left( N^n_{\ttau^n_{k}}-N^n_{\ttau^n_{k-1}}\right)
    \left( N^n_{\ttau^n_{k}}-N^n_{\ttau^n_{k-1}}-1\right)\right]
    a_0^n(\phi)(\bar R^n_{\f{k-1}{n}})+
    E\left[\left( N^n_{\ttau^n_{k}}-N^n_{\ttau^n_{k-1}}\right)\right]
    a_1^n(\phi)\left(\bar R^n_{\f{k-1}{n}}\right)
    \Bigg\}
    \nonumber\\
    &=
    \frac{1}{n^3} \left(2 n^2 a^n_0(\phi)
      +n a^n_1(\phi)
    \right)\left(\bar R^n_{\f{k-1}{n}}\right)
    \nonumber\\
    & = \frac{1}{n} \left(\sigma^n(\phi)\left(\bar R^n_{\f{k-1}{n}}\right) \right)^2.\label{est-prop-martg}
  \end{align}
  Set
  \begin{align*}
    \sigma(\phi)(R)=\sqrt{2}\int_{\mathbb{R}}f(x-R)\phi(x)\,dx
    E\left[ \omega_1^{\bN} \right],\quad
    t\in[0,\infty).
  \end{align*}
  Note that $\sigma \geq 0$ since $\phi$ is non-negative.

    ii) We claim that $\sigma^n(\phi) \to \sigma$ uniformly. First note that
    $\norm{a_1^n(\phi)}_{L^\infty} \le \norm{\phi}_\infty^2
    E[\omega_1^{\bN}]$. Hence, $\f{1}{n} a^n_1(\phi) \to 0$ uniformly, and we
    may ignore the second term in the definition of $\sigma^n$. Further note
    that
    \begin{align*}
      & \abs{\sum_j \int_{x_j^n}^{x_{j+1}^n} f(x-R) dx \int_{x_j^n}^{x_{j+1}^n}
        \phi(x) dx \f{1}{\dx^n} - \int_{\R} f(x-R) \phi(x) dx } \\
        \le & \sum_j
      \int_{x_j^n}^{x_{j+1}^n} f(x-R) \abs{\f{1}{\dx^n}
        \int_{x_j^n}^{x_{j+1}^n} \phi(y) dy - \phi(x)} dx.
    \end{align*}
    By the mean value theorem, there exists $y \in [x_j^n, x_{j+1}^n]$ with
    $\f{1}{\dx^n} \int_{x_j^n}^{x_{j+1}^n} \phi(y) dy = \phi^\prime(y)$. For
    $x < y$, $\abs{x-y} \le \dx^n$ we have
    \begin{equation*}
      \abs{\phi(x) - \phi(y)} = \int_{\R} \indic{[x,y]}(z) \phi^\prime(z) dz
      \le \sqrt{\dx^n} \norm{\phi}_{H^1}.
    \end{equation*}
    Therefore,
    \begin{equation*}
      \norm{\sigma^n(\phi) - \sigma(\phi)}_{L^{\infty}} \le \sqrt{2}
      E[\omega_1^{\bN}] \sqrt{\dx^n} \norm{\phi}_{H^1} + o(1),
    \end{equation*}
    and we have established uniform convergence of $\sigma^n(\phi)$ to
    $\sigma(\phi)$. 
  
  {iii) We need to verify the conditions of Lemma~\ref{lem:C1}
    outlined in Assumption~\ref{ass:1}, i.e.:
    \begin{gather}
      \tag{A.1}
      \sup_n \norm{\sigma^n}_{L^{\infty}} < \infty,\\
      \tag{A.2}
      E\sum_{k=1}^{\lfloor nt \rfloor+1}
      |\barY_{\f{k}{n}}^n(\phi)-\barY_{\f{k-1}{n}}^n(\phi)|^4 
      \rightarrow 0,\\
      \tag{A.3}
      \sup_n E\left[ \sup_{k \le \floor{nt}} \abs{\barY^n_{k/n} -
          \barY^n_{(k-1)/n}} \right] < \infty.
  \end{gather}
  Note that~\eqref{eq:1} follows immediately from boundedness of $f^{\bN_r}$
  and of $\phi$.  \eqref{eq:4} is a direct consequence of
  Lemma~\ref{lem:tightness-barV3}.  As for~\eqref{eq:3},
  \begin{align*}
    E_{\mathcal{G}_{\frac{k-1}{n}}^n}\left[
      |\barY^n_{\f{k}{n}}(\phi)-\barY^n_{\f{k}{n}}(\phi)|^4  \right] &\le
    C (\Delta v^n)^2(\Delta x^n)^4
    E\left[\left| N^n_{\ttau^n_{k}}-N^n_{\ttau^n_{k-1}}\right|^4\right]
    \\
    &\leq\,C \frac{1}{n^6} \left[n^4+ n \right]
    \\
    &\leq\,C \frac{1}{n^2},
  \end{align*}
  where $C$ is a positive constant which is independent of $n$ and may vary from
  line to line. Thus, for any $t \in(0,\infty)$,
  \begin{align*}
    E\sum_{k=1}^{\lfloor nt \rfloor+1} |\barY_{\f{k}{n}}^n(\phi)-\barY_{\f{k-1}{n}}^n(\phi)|^4
    \leq \, C(nt+1) \frac{1}{n^2} \rightarrow 0 \text{ as }
    n\rightarrow \infty.
  \end{align*}
}
 
iv) The previous arguments easily extend to the finite dimensional case. For
each $l\in\mathbb N^+$ and any family of non-negative functions
$\phi_1,\dots,\phi_l$, the process $(\barY^n(\phi_1),\dots,\barY^n(\phi_l))$
converges jointly to $(\barY(\phi_1),\dots,\barY(\phi_l))$ in distribution.% ,
We
compute for $i,j=1,\dots,l$,
  \begin{align*}
E_{\mathcal{G}_{\frac{k-1}{n}}^n}\left[
      \left( \barY^n_{\frac{k}{n}}(\phi_j)-\barY^n_{\frac{k-1}{n}}(\phi_j)\right)  \left( \barY^n_{\frac{k}{n}}(\phi_i)-\barY^n_{\frac{k-1}{n}}(\phi_i)\right)   \right]
      &=\frac{(\sigma^n(\phi_j+\phi_i))^2-(\sigma^n(\phi_j-\phi_i))^2}{4n}\left(\bar R^n_{\f{k-1}{n}}\right)
      \\
      &=\frac{1}{n}\sigma^n(\phi_j)\sigma^n(\phi_i)\left(\bar R^n_{\f{k-1}{n}}\right).
\end{align*}
% This shows that the first column of the matrix $\bar \sigma(\cdot)$ in
% Lemma~\ref{lem:C1} equals $\left( \sigma(\phi_1)(\cdot), \ldots,
%   \sigma(\phi_l)(\cdot)\right)'$ while all the other entries of this matrix
% are zero.  In particular, if one of the test function is strictly positive,
% the arguments given in step (iii) show that all the processes $\barY(\phi_i)$
% are driven by the same Wiener process.
Since $\mathcal{E}$ is dense in $H^1$,
this completes the proof.
\end{proof}

The previous proposition characterizes the quadratic variation of the limiting
volume density processes. Next we are going to study the limiting
dynamics of aggregate order placements and cancellations, disregarding the
random fluctuations. As we expect order placements and cancellations to
contribute to the drift part of the limiting model, we find it helpful to
re-write their dynamics in the form of an integral in time. That is, if we
write
\begin{gather*}
  \barV^{n,2}(t,x) = \int_0^{\f{\lfloor nt\rfloor}{n}}g^n(s,x)ds,\\
  \barV^{n,1}(t,x) = \int_0^{\f{\lfloor nt\rfloor }{n}}\tg^n(s,x)ds,
\end{gather*}
it is clear that we can identify the limiting drift term by studying the
limits of $g^n$ and $\tg^n$. Comparing with~\eqref{eq:bar-def}, we have
\begin{align*}
  g^n(t,x) &\coloneqq \sum_{k=1}^{\infty}
  \sum_{i=N^n(\ttau^n_{k-1})+1}^{N^n(\ttau^n_{k})}
  \indic{I^n\left(\pi^\bC_{i}+\bar R^n_{\f{k-1}{n}}\right)}(x) \omega_{i}^\bC
  \indic{[\f{k}{n},\f{k+1}{n})}(t) \f{\dv^n}{\dx^n} n,
  \\
  \tg^n(t,x) &\coloneqq \sum_{k=1}^{\infty}
  \sum_{i=N^n(\ttau^n_{k-1})+1}^{N^n(\ttau^n_{k})}
  \indic{\left(\pi^\bP_{i}+\bar R^n_{\f{k-1}{n}}\right)}(x) \omega_{i}^\bP
  \indic{[\f{k}{n},\f{k+1}{n})}(t) \f{\dv^n}{\dx^n} n.
\end{align*}

With regards to aggregate cancellations, $g^n$ only captures the proportionality
of cancellations in terms of present volume. Therefore, we need to introduce
another term $\barg^n$ describing the actual cancellations, i.e.,
\begin{equation*}
  \barv^{n}(t,x)-v(0,x)-\barV^{n,1}(t,x)-\barV^{n,3}(t,x)
  =\int_0^{\f{\lfloor nt\rfloor }{n}}\barg^n(s,x)ds.
\end{equation*}
Clearly, $\barg^n$ is given by
\begin{equation*}
  \barg^n(t,x) \coloneqq \sum_{k=1}^{\infty}
  \sum_{i=N^n(\ttau^n_{k-1})+1}^{N^n(\ttau^n_{k})}
  \indic{I^n\left(\pi^\bC_{i}+\bar R^n_{\f{k-1}{n}}\right)}(x) \omega_{i}^\bC
  v^n(\tau^n_{i-1}, x)
  \indic{[\f{k}{n},\f{k+1}{n})}(t) \f{\dv^n}{\dx^n} n.
\end{equation*}
We will analyze the impact of order cancellations in the limit in two steps:
first we show that we can replace $\barg^n$ by the (much simpler) expression
$g^n \barv^n$ in the limit (see Lemma~\ref{lem-apprxm-v-a-2}). Then we
characterize the limit of the latter term in the appropriate sense (see
Lemma~\ref{lem-limit-L2-Vn12}, where we also characterize the limiting object
of the order placements).

\begin{remark}\label{rmk-gn}
  From the proof of Lemma \ref{lem:tightness-barV1/2}, it follows that for
  $p\in\{2,4\}$,
  \begin{align*}
    E\left[ \norm{g^n(t)}^p_{L^p} \right]
        +\sup_{x \in \R} E_{\F^n_s}
        E\left[ \left| g^n(t,x) \right|^p \right]
        \le C,
  \end{align*}
  which implies that
  \begin{align*}
    \sup_{x\in\R}E\int_0^{t} \left| g^n(s,x)\right|^p\,ds
    +
    E\int_{\R}\int_0^{t} \left| g^n(s,x)\right|^p\,dsdx
    \le Ct,
  \end{align*}
  with the constants $C$ being independent of $n$ and $ t$.
\end{remark}

\begin{lemma}\label{lem-apprxm-v-a-2}
  For any $t>0$, we have
  \begin{align}
    \lim_{n\rightarrow\infty}
  E\left[\int_{\R}\int_0^{\frac{\lfloor nt\rfloor }{n}}\Big|\barg^n(s,x)-
  g^n(s,x)\barv^{n}(s,x)\Big|^2\,dsdx\right] =0.
  \end{align}
\end{lemma}
\begin{proof}
Using Fubini's theorem and Remark \ref{rmk-difference-vak}, we have
  \begin{align*}
    &E\int_{\R}\int_0^{\frac{\lfloor nt\rfloor }{n}}\Big|\barg^n(s,x)-
    g^n(s,x)\barv^{n}(s,x)\Big|^2\,dsdx
    \\
    &=\,\int_0^{\frac{\lfloor nt \rfloor}{n}} E\int_{\R}\bigg| \sum_{k\in\N^+}
    \sum_{i=N^n(\ttau^n_{k-1})+1}^{N^n(\ttau^n_{k})}
  \indic{I^n\left(\pi^\bC_{i}+\bar R_{\f{k}{n}}\right)}(x)
  \omega_{i}^\bC \Big(v^n(\tau^n_{i-1},x)-
  \barv^n(s,x)\Big)\indic{[\f{k}{n},\f{k+1}{n})}(s)
  \f{\dv^n}{\dx^nn^{-1}}\bigg|^2\, dxds\\
  &\leq\,
  \int_0^{\frac{\lfloor nt\rfloor }{n}}\sum_{k\in\N^+\cup\{0\}}
  \indic{[\frac{k}{n},\f{k+1}{n})}(s)
  \bigg(
  E\int_{\R} |g^n(s,x)|^4dx
  \bigg)^{1/2}
  \bigg(
  E\sup_{i\in [N^n(\ttau^n_{k-1}),
  \,N^n(\ttau^n_{k})]\cap \N^+}\|v^n(\tau_{i})-v^n(\ttau_{k-1}^n)\|_{L^4}^4
  \bigg)^{1/2}ds
  \\
  &\leq\, C \f{1}{\sqrt{n}} \int_0^{\frac{\lfloor nt\rfloor}{n}} \bigg(
  E\int_{\R} |g^n(s,x)|^4dx \bigg)^{1/2}ds\\
  &\leq\, C \f{1}{\sqrt{n}} \Bigg(E\int_0^{\frac{\lfloor nt\rfloor
    }{n}}\int_{\R} |g^n(s,x)|^4dx\Bigg)^{1/2},
  \end{align*}
  which by Remark \ref{rmk-gn} converges to zero as $n$ tends to infinity.
\end{proof}

We can now analyze the limiting objects obtained from order placements and
cancellations. The proof of Lemma~\ref{lem-limit-L2-Vn12} is technical and
rather long and hence postponed to Appendix~\ref{sec:appdx-prf-lem}.

\begin{lemma}\label{lem-limit-L2-Vn12}
  For any $t=\frac{\lfloor nt \rfloor}{n}$ with $n\in \mathbb{N}$,
  \begin{align}
    &\forall\,\alpha\in\{0,1\}: \ \lim_{n\rightarrow\infty}
    \sup_{x\in\R}E\left[\left|\int_0^{t}\left(g^n(s,x)-E[\omega_{1}^\bC]f^\bC(x-R_s)
        \right)\left(1-\alpha+\alpha\barv^n(s,x)\right)\,ds\right|^2\right]
    =0,\label{eq-limit-vn2a}
    \\ &\lim_{n\rightarrow\infty} \sup_{x\in\R}E\left[\left|\int_0^{t}
    \left(\tg^n(s,x)-E[\omega_{1}^\bP]f^\bP(x-R_s)\right)
    \,ds\right|^2 \right] =0.
  \end{align}
\end{lemma}

Combining the characterization of the limit of the fluctuation part of
$\barv^n_r$ obtained in Proposition~\ref{prop-volume-limit-V3} with the
characterization of the limits of order cancellations and placements obtained
in Lemma~\ref{lem-limit-L2-Vn12} together with Lemma~\ref{lem-apprxm-v-a-2},
we are in the position to study the limit of $\barv^n$ itself.

\begin{theorem}\label{thm-volume-va-limit}
 % For $r=a,b$, there exists a Wiener process $\{W_r(t);t\in [0,\infty)\}$ such that
 %  $(\barA^n,\barB^n,\barV^{n,3}_r,\barv_r^n)\Rightarrow (A,B,V^3_r,v_r)$
 %  {along a properly chosen subsequence}, where
 %  $(A,B)$ is the {(not yet identified)} limit of $(A^n, B^n)$
 %  {(along the chosen subsequence)},
 %  $V^3_r$ the limit obtained in Proposition~\ref{prop-volume-limit-V3}, and
 %  \begin{multline*}
 %    v_r(t,x)=v_r(0,x)+
 %    \int_0^t \left(E[\omega^{\textbf{P}}_{1}]f^\textbf{P} (x-R_s)-E[\omega^\bC_{1}]
 %      f^\bC(x-R_s)v_r(s,x)\right)\,ds + \\
 %    +\sqrt{2} \int_0^t E\left[
 %      \omega^{\bN}_1 \right] f^\bN(x-R_s)\,dW_r(s),\quad t\geq 0.
 %  \end{multline*}
    Suppose that (along a properly chosen subsequence)  
    $(\barA^n,\barB^n,\barV^{n,3}_r,\barv_r^n)\Rightarrow (A,B,V^3_r,v_r)$, where
    $(A,B)$ is the (not yet identified) limit of $(A^n, B^n)$ (along the
    chosen subsequence) and $V^3_r$ the limit obtained in
    Proposition~\ref{prop-volume-limit-V3}. Then
    \begin{equation*}
      v_r(t,\cdot)=v_r(0,\cdot)+
      \int_0^t \left(E[\omega^{\textbf{P}}_{1}]f^\textbf{P} (\cdot-R_s) -
        E[\omega^\bC_{1}] f^\bC(\cdot-R_s)v_r(s,\cdot)\right)\,ds 
      +V^3_r(t, \cdot) ,\quad t\geq 0,
  \end{equation*}
  and $V^3_r$ remains a martingale % (in the sense of
  % Proposition~\ref{prop-volume-limit-V3})
  under the filtration generated by
  $(A, B, v_r, V^3_r)$.
\end{theorem}

\begin{proof}
  {First, note that the $v_r$ is already measurable w.r.t.~$R$
    and $V^3_r$, hence the filtration does not change when $v_r$ is added and
    $V^3_r$ trivially stays a martingale.}

  The sequence of price processes is C-tight and converges in distribution to
  some limit $(A,B)$ along a subsequence. The processes $\barV_r^{n,3}$ and
  $\barv_r^{n}$ are tight, due to Proposition \ref{prop:tightness-barv} and
  $\barV_r^{n,3}$ is even $C$-tight, due to Proposition
  \ref{prop-volume-limit-V3}. Hence, the sequence
  $\big(\barA^n,\barB^n,\barV^{n,3}_r,\barv_r^n \big)$ is tight as a sequence
  of processes with sample paths in $\mathcal{D}(0,\infty;\R^2\times
  H^{-1}\times H^{-1})$.  In order to identify the limit of $\barv_r^n$ as a
  function of the (existing, yet still to be identified) limit of the price
  process we use the additive decomposition
 \begin{align}\label{eq-decomposition-bar-v}
    \barv^n_r(t,x)-v_r^n(0,x)
    =\barV^{n,1}_r(t,x)+\widetilde{V}^{n,2}_r(t,x)+\barV^{n,3}(t,x),
    \quad (t,x)\in[0,\infty)\times \R,
  \end{align}
  where
  \[
  \widetilde{V}^{n,2}_r(t,x): =\int_0^{\f{[nt]} {n}}\barg^n(s,x)ds.
  \]
  In view of Skorohod's lemma (see Lemma~\ref{lem:Skorokhod}) we may
  w.l.o.g. assume that all processes are defined on a common probability space
  $(\Omega,{\mathcal F}, \mathbb{P})$ such that, for some process $v_r$ to be
  determined, the sequence $\big( \barA^n,\barB^n,\barV_r^{n,3}, \barv_r^n
  \big)$ converges almost surely to some limit $(A,B,V_r^3,v_r)$ as a sequence
  of processes with sample paths in $\mathcal{D}(0,\infty;\mathbb{R}^2\times
  H^{-1} \times H^{-1})$. In particular, and this will be used below, as a
  sequence in $\mathbb{R}^2\times H^{-1} \times H^{-1}$,
\[
	\lim_{n \to \infty} \big( \barA^n,\barB^n,\barV_r^{n,3}, \barv_r^n \big) = (A,B,V_r^3,v_r) \quad \mathbb{P} \otimes dt\mbox{-a.e.}
\]
Indeed, taking $\barv^n_r$ for example, the sample paths are c\'{a}dl\'ag, and hence they have at most countably many discontinuities. For almost all $\omega \in \Omega$ the convergence $\lim_{n \to \infty}   || \barv^n_r(t,\cdot) - v_r(t,\cdot) ||^2_{H^{-1}}=0$ at each point of continuity can be derived in a similar way to \cite[Prop. VI.1.17]{JacodShiryaev2002}. Then, dominated convergence yields
\[
	\lim_{n \to \infty} E \int_0^T  || \barv^n_r(t,\cdot) - v_r(t,\cdot) ||^2_{H^{-1}} \wedge 1 \,\,dt  = 0 \quad \mbox{for all $T>0$.}
\]	
This allows us to choose a subsequence that is converging a.e.~in $H^{-1}$.

Operating in the probability space $(\Omega,{\mathcal F}, \mathbb{P})$ the decomposition (\ref{eq-decomposition-bar-v}) of the volume process shows that we need to identify the limit of $\widetilde{V}^{n,2}$ in order to identify that of $\barv^n_a$. For this, we first show that it is enough to identify weak limits in the Hilbert space $L^2(\Omega \times [0,T] \times \R)$ for arbitrary $T>0$.
In fact, by Lemma \ref{lem:tightness-barv} the sequence $\barv_r^n$ is uniformly bounded  in $L^2(\Omega \times [0,T] \times \R)$. By Lemma \ref{lem:tightness-barV1/2} and Lemma \ref{lem:tightness-barV3} the same applies to $\barV^{n,1}$ and $\barV^{n,3}$. Hence, the sequence $\big(\barv^n, \barV^{n,1}, \barV^{n,3} \big)$ has a weak accumulation point in $L^2(\Omega \times [0,T] \times \R)$.   
By the Banach-Saks theorem, the weak accumulation point is a strong limit in Cesaro sense of
  a subsequence. Since $L^2(\Omega \times [0,T] \times \R) \subset L^2\left( \Omega
    \times [0,T]; H^{-1} \right)$ this shows that the weak limit 
  coincides with $\big( v_r, V^{1}, V^3)$ as a weak limit in $L^2\left( \Omega \times [0,T]
    \times \R \right)$. As a result, it is enough to identify the weak limit
  $K$ of $\widetilde{V}^{n,2}_r$ in $L^2(\Omega \times [0,T] \times \R)$. By
  Lemma~\ref{lem-apprxm-v-a-2} and~\ref{lem-limit-L2-Vn12} this is equivalent to identifying the weak limit of the process
  \[
  {(t,x) \mapsto} \int_0^tE[\omega_{1}^\bC]f^\bC(x - R_s)
  \barv_r^n(s,x)\,ds.
  \]

 In order to identify $K$ we test against test functions $\psi\in
  L^{\infty}(\Omega \times [0,T])$ and $\phi
  \in L^2(\R)$. Weak convergence of $\barv^n$ and $\widetilde{V}^{n,2}$ in $L^2(\Omega\times[0,T] \times \R)$ yields that
  \begin{align*}
     E\int_0^{T}\int_{\R}\psi(t)K(t,x) \phi(x)\,dx \,dt
    =&\lim_{n\rightarrow\infty}E\int_0^{T}\psi(t)\langle \widetilde{V}^{n,2}_r(t),\,\phi\rangle \,dt
    \\
    =&
    \lim_{n\rightarrow\infty}E\int_0^{T}\psi(t)\int_0^{\f{[nt]}{n}}
    \int_{\R}\barg^n(s,x) \phi(x)\,dxds \,dt
    \\
    &\textrm{(by Lemma \ref{lem-apprxm-v-a-2})}\\
    =&
    \lim_{n\rightarrow\infty}E\int_0^{T}\psi(t)\int_0^{\f{[nt]}{n}}
    \int_{\R}g^n(s,x)\barv^{n}_r(s,x) \phi(x)\,dxds \,dt\\
    &\textrm{(by Lemma \ref{lem-limit-L2-Vn12})}\\
    =&E[\omega^\bC_{r,1}]
     \lim_{n\rightarrow\infty}E\int_0^{T}\psi(t)\int_0^{\f{[nt]}{n}}
    \int_{\R}f^\bC(x-R_s)\barv^{n}_r(s,x) \phi(x)\,dxds \,dt
    \\
    =&E[\omega^\bC_{r,1}]
     \lim_{n\rightarrow\infty}E\int_0^{T}
     \int_{\R}f^\bC(x-R_s)\barv^{n}_r(s,x) \phi(x)\,dx
     E_{\barF_{{\ceil{ns}/n}}} \Big[
     \int_{{\ceil{ns}/n}}^{T}\psi(t) \,dt\Big] \,ds\\
    &\textrm{(by the weak convergence in Hilbert space)}\\
    =&
    E[\omega^\bC_{r,1}]
    E\int_0^{T}
    \int_{\R}f^\bC(x-R_s){v}_r(s,x) \phi(x)\,dx
    E_{\barF_s}\Big[\int_s^{T}\psi(t)
    \,dt\Big] \,ds\\
    =&E[\omega^\bC_{r,1}]
     E\int_0^{T}\psi(t)\int_0^{t}
    \int_{\R}f^\bC(x-R_s){v}_r(s,x) \phi(x)\,dxds \,dt,
  \end{align*}
  where $\barF_t$ denotes the filtration generated by all the processes
  $\overline{A}^n$, $\overline{B}^n$, $A$, $B$, $\barv^n_r$ and $v_r$. Since $\phi\in L^2$ and $\psi\in
  L^{\infty}(\Omega\times[0,T])$ are arbitrary, we
  get
  $$
  K(t,x)=E[\omega^\bC_{1}]\int_0^{t}
    f^\bC(x-R_s)v_r(s,x) \,ds
  $$
  for almost every $(t,\omega,x)\in[0,T] \times \Omega \times \R$. Hence, the
  limit $v_r$ satisfies
    \begin{equation*}
      v_r(t,\cdot)=v_r(0,\cdot)+ \int_0^t
      \left(E[\omega^{\bP_r}_1]f^{\textbf{P}_r}(\cdot-R_s)-E[\omega^{\bC_r}_{1}]
        f^{\bC_r}(\cdot-R_s)v_r(s,\cdot)\right)\,ds
      + V^3_r(t,\cdot), \quad t\geq 0. \qedhere
  \end{equation*}
\end{proof}

%%%%%%%%%%%%%%%%%%%%%%%%%%%%%%%%%%%
%%%%%%%%%%%%%%%%%%%%%%%%%%%%%%%%%%%
%%%%%%%%%%%%%%%%%%%%%%%%%%%%%%%%%%%

\subsection{The limit of the volume density}
\label{sec:limit-volume-density}

With tightness of the sequence of auxiliary processes $\barv^n_{r}$
established in Proposition~\ref{prop:tightness-barv}, we can now turn to the
actual volume densities $v^n_{r}$. To this end, we introduce the processes 
\[
	\hatv^n_{r}(u) := \barv^n_{r} \circ \eta^n_u, \quad
	\hatV^{n,i}_{r}(u) := \barV^{n,i}_{r} \circ \eta^n_u \quad (r=a,b; i=1,2,3)
\]
where the time-change $\eta^n_u$ was defined in~\eqref{eq:bareta-def}. In view
of Kurtz's \cite{Kurtz} strong approximation result for Poisson processes by
Brownian motion, for any $T > 0$ 
\[
	\lim_{n \to \infty} \sup_{0 \leq t \leq T} | \eta^n_t - t | = 0 \quad \mathbb{P}\mbox{-a.s.}
\]
As a result, Lemma \ref{lem:billingsley} and Theorem \ref{thm-volume-va-limit}
imply that the limit of $(A^n,B^n,\hatv^n_r)$ coincides with that of
$(\barA^n,\barB^n,\barv^n_r)$, namely $(A,B,v_r)$ of Theorem
\ref{thm-volume-va-limit}.

Let $\delta v^n_{r} \coloneqq v^n_{r} - \hatv^n_{r}$ and $\delta V^{n,i}_{r}
\coloneqq V^{n,i}_{r} - \hatV^{n,i}_{r}$ $(i = 1, 2, 3)$. Our goal is to prove
that $\delta v^n_{r}$ converges weakly to $0$ as $n \to \infty$. We
shall then deduce that convergence of $ \hatv^n$ implies convergence of $v^n$.
The first step is to establish moment estimates for the processes $V^{n,i}$
$(i=1,2,3)$ similar to Lemmas \ref{lem:tightness-barV3} and
\ref{lem:tightness-barv}. Analogous to Proposition \ref{prop:tightness-barv}
these estimates indicate tightness of $v_r^n$ and thus the tightness of
$(A^n,B^n,v^n_r)$. The rather technical proof is deferred to
Appendix~\ref{sec:appdx-prf-lem}.

\begin{lemma}\label{lem-est-van}
  For $r=a,b$ and $i=1,2,3$ it holds that
  \begin{align*}
    E_{\F^n_s}\left[
    \sum_{i=1}^3\big\|V^{n,i}_{r}(t)-V^{n,i}_{r}(s)\big\|_{L^2}^2
    \right]
    \,\leq&\,C^n_s \left[(t-s)+(t-s)^2\right],\quad 0\leq s\leq t<\infty,
    \\
    E_{\F^n_s}\left[
    \big\|v_{r}^n(t)-v_{r}^n(s)\big\|_{L^2}^2
    \right]
    \,\leq&\,C^n_s \left[(t-s)+(t-s)^2\right],\quad 0\leq s\leq t<\infty,
  \end{align*}
  with $\sup_nE\left[ \sup_{s\in[0,t]} C_s^n  \right]\leq\,C(t^2+t)$,
  $t\in[0,\infty)$, where the constant $C$ is independent of $n$, $s$ and
  $t$.
\end{lemma}

Furthermore, we will show that $\delta v^n_{r}(t)$ converges point-wise to
$0$ in an $L^2$-sense for which we need some elementary results on Poisson processes.

\begin{lemma}
  \label{lem:beta-distribution}
  Let $N_1$ and $N_2$ be two independent Poisson processes with intensities
  $\lambda_1$ and $\lambda_2$, respectively. Moreover, let $T_i$, $i=1,
  \ldots$, denote the jump times of the Poisson process $N_1$. Then we have
  \begin{gather*}
    E\left[ N_2(t) - N_2(T_{N_1(t)}) \right] = \f{\lambda_2}{\lambda_1} \left(
    1 - e^{-\lambda_1 t} \right),\\
  E\left[ \left( N_2(t) - N_2(T_{N_1(t)}) \right) \left( N_2(t) -
      N_2(T_{N_1(t)}) - 1 \right) \right] = 4 \f{\lambda_2^2}{\lambda_1^2}
  \left(1 - (1+t \lambda_1) e^{-\lambda_1 t} \right).
  \end{gather*}
\end{lemma}
\begin{proof}
  Notice that conditional on $N_1(t) = l$, the relative difference $(t -
  T_l)/t$ has a beta distribution with parameters $1$ and $l$, as this is the
  distribution of the differences in the order statistics of $l$ random
  variables distributed uniformly on $[0,1]$. Hence, elementary calculations
  give
  \begin{equation*}
    E\left[ N_2(t) - N_2(T_l) \, | \, N_1(t) = l \right] = \sum_{k=0}^{\infty}
    k \int_0^1 e^{-\lambda_2 t x} \f{(\lambda_2 t x)^k}{k!}
    \f{1-x)^{l-1}}{B(1,l)} dx = \f{\lambda_2 t}{1+l}
  \end{equation*}
  and
  \begin{align*}
    E\left[ \left(N_2(t) - N_2(T_l)\right) \left(N_2(t) - N_2(T_l) -
        1\right)\, | \, N_1(t) = l \right] = & \sum_{k=0}^{\infty}
    k (k-1) \int_0^1 e^{-\lambda_2 t x} \f{(\lambda_2 t x)^k}{k!}
    \f{1-x)^{l-1}}{B(1,l)} dx \\
     = & \f{2 \lambda_2^2 t^2}{2+3l+l^2}.
  \end{align*}
  Multiplying these terms with $P(N_1(t) = l) = e^{-\lambda_1 t}
  \f{(\lambda_1t)^l}{l!}$ and summing over $l$ gives the formulas from above.
\end{proof}

\begin{lemma}
  \label{lem:bound-deltav}
  Let $u = u(t) = u(t,x)$ denote any of the processes $\delta v^n_{r}$,
  $\delta V^{n,i}_{r}$, $i=1,2,3$. Moreover, assume that the sequence
  $v^n_{r}(0)$ is uniformly bounded in $L^2$. Then there is a constant $C$
  independent of $n$ or $t$ such that
  \begin{equation*}
    E\left[ \norm{u(t)}_{L^2}^2 \right] \le C\frac{1}{n}(1+t+t^2),\quad
    \forall\,t\in[0,\infty).
  \end{equation*}
\end{lemma}
\begin{proof}
  Let us first consider $u = \delta V^{n,i}_{r}$ for some $i=1,2,3$, $r=a,b$. Note that for some random variables $\omega_i$ and $\pi_i$ we have
  for some scaling constant $\epsilon$ (either equal to $\dv / \dx$ or equal
  to $\sqrt{\dv}$)
  \begin{equation*}
    u(t,x)^2 = \left( \sum_{i=N\left( \ttau_{\tN(t)} \right)}^{N(t)}
      \indic{I\left( R^n(\ttau^n_{\tN(t)}) + \pi_i \right)}(x) \omega_i
    \right)^2 \epsilon^2,
  \end{equation*}
  as $\txi_{r,i}$ is constant in $i$ and $\txi_{r,i}^2 = 1$. Letting
  $\mathcal{G}$ denote the $\sigma$-algebra generated by all sources of
  randomness \emph{except} $(\omega_i)_{i\in \N^+}$, we have
  \begin{align*}
    E\left[ u(t,x)^2 \right] &= E\left[ \left\{ \sum_{i \neq i' = N\left(
          \ttau_{\tN(t)} \right)}^{N(t)}
          \!\!\!\!
           E_{\mathcal{G}}\left[ \omega_i
        \omega_{i'} \right] \indic{I\left( R^n(\ttau^n_{\tN(t)}) + \pi_i
        \right)}(x) \indic{I\left( R^n(\ttau^n_{\tN(t)}) + \pi_{i'}
        \right)}(x) + \sum_{i = N\left(
          \ttau_{\tN(t)} \right)}^{N(t)}
          \!\!\!\!
           E_{\mathcal{G}}\left[ \omega_i^2
      \right] \indic{I\left( R^n(\ttau^n_{\tN(t)}) + \pi_i
        \right)}(x)\right\} \right] \epsilon^2 \\
    &= E\left[ \left\{ \sum_{i \neq i' = N\left(
          \ttau_{\tN(t)} \right)}^{N(t)} \indic{I\left( R^n(\ttau^n_{\tN(t)}) + \pi_i
        \right)}(x) \indic{I\left( R^n(\ttau^n_{\tN(t)}) + \pi_{i'}
        \right)}(x) E[\omega_1]^2 + \sum_{i = N\left(
          \ttau_{\tN(t)} \right)}^{N(t)} \indic{I\left( R^n(\ttau^n_{\tN(t)}) + \pi_i
        \right)}(x) E\left[\omega_1^2\right] \right\} \right] \epsilon^2.
  \end{align*}
  Furthermore, conditioning on the $\sigma$-algebra generated by
  all sources of randomness except for $\left( \pi_i \right)_{i \in \N^+}$, we
  can bound in a similar way to~\eqref{eq:density-bound}
  \begin{multline*}
    E\left[ u(t,x)^2 \right] \le E\Biggl[ E\left[\omega_1\right]^2
        \norm{f}^2_{L^{\infty}} \dx^2
        \left( N(t) - N\left( \ttau_{\tN(t)} \right) \right) \left( N(t)
      - N\left( \ttau_{\tN(t)} \right) - 1 \right) \indic{\left[R(\ttau_{\tN(t)})-M,
          R(\ttau_{\tN(t)}) + M \right]}(x)  +\\
      + E\left[\omega_1^2\right]
      \norm{f}^2_{L^{\infty}} \dx \left( N(t) - N\left( \ttau_{\tN(t)} \right)
      \right) \indic{\left[R(\ttau_{\tN(t)})-M, R(\ttau_{\tN(t)}) + M \right]}(x)
    \Biggr] \epsilon^2.
  \end{multline*}
  Hence, plugging in Lemma~\ref{lem:beta-distribution}, we obtain
  \begin{align*}
     E\left[ \norm{u(t)}^2_{L^2} \right]
    &\le C \left( \dx^2 E\left[ \left( N(t) - N\left( \ttau_{\tN(t)} \right)
        \right) \left( N(t) - N\left( \ttau_{\tN(t)} \right) - 1 \right)
      \right]  + \dx E\left[ \left( N(t) - N\left( \ttau_{\tN(t)} \right)
        \right) \right] \right) \epsilon^2 \\
    &= C\left( \dx^2 4 \f{\lambda^2}{\mu^2} \left[ 1 - (1+t\mu) e^{-\mu t}
      \right] + \dx \f{\lambda}{\mu} \left[ 1 - e^{-\mu t} \right] \right)
    \epsilon^2 \\
    &\le C\left( \f{1}{n} \f{n^4}{n^2} + \f{1}{\sqrt{n}} \f{n^2}{n} \right)
    \epsilon^2 \\
    &= C \left( n + \sqrt{n} \right) \epsilon^2.
  \end{align*}
  Now we recall that $\epsilon^2 = \f{\dv^2}{\dx^2} = n^{-3}$ in case $i=1,2$
  and $\epsilon^2 = \dv = n^{-2}$ in case $i=3$.

  The proof for the estimate of $\delta v^n_{r}$ works in precisely the same
  way as the proof of Lemma~\ref{lem:tightness-barv}, taking into account the
  appropriate estimates for $\delta V^{n,i}_{r}$ derived above.
\end{proof}

Combining these lemmas with the results in
Theorem~\ref{thm-volume-va-limit} we can now prove convergence of the volume
densities. We denote by $(A,B)$ an accumulation point of the sequence of price processes.
Ex post, we shall see
that the limit is unique, and hence we do not actually need
to work with such a sub-sequence.

\begin{theorem}\label{thm-volume-LM}
  The sequence of processes $\left( A^n, B^n, v_a^n, v_b^n \right)$ is
  tight. Given a subsequence such that
  $(A^n,B^n,v_a^n,v_b^n)\Rightarrow (A,B,v_a,v_b)$ for some volume processes
  $v_a$ and $v_b$. Then
  \begin{multline}
    \label{ODE-v}
    v_r(t,\cdot) = v_{r,0}(\cdot)+
    \int_0^t \left(E[\omega^{\bP_r}_{1}]f^{\bP_r}(\cdot- R_s)-E[\omega^{\bC_r}_1]
      f^{\bC_r}(\cdot-R_s)v_r(s,\cdot)\right)\,ds + V^3_r(t,\cdot),\quad
    t\geq 0.
  \end{multline}
  $V^3_a$ and $V^3_b$ are martingales w.r.t.~the filtration generated by
  $(A,B,v_a,v_b)$, and their quadratic co-variance \emph{diagonalizes}. More
  precisely, given test functions $\phi^1_a, \ldots, \phi^l_a, \phi^1_b,
  \ldots, \phi^k_b \in \mathcal{E}$, then for any $1 \le i \le l$, $1 \le j
  \le k$ we have 
  \begin{equation*}
    \sip{\ip{\phi^i_a}{V^3_a}}{\ip{\phi^j_b}{V^3_b}}_t = 0, \quad t \ge 0.
  \end{equation*}
\end{theorem}
\begin{proof}
  Recall that
  \begin{equation*}
    (A^n,B^n,\hatv^n_{a})(u) = (\barA^n,\barB^n,\barv^n_{a}) \circ \eta^n_u.
  \end{equation*}
  Since the time change process converges almost surely to the identity
  uniformly on compact time intervals, it follows from Lemma
  \ref{lem:billingsley} and Theorem \ref{thm-volume-va-limit} that
  $(A^n,B^n,\hatv^n_{a})\Rightarrow (A,B,v_a)$. On the other hand, in a
  similar way to Proposition \ref{prop:tightness-barv} we derive from Lemma
  \ref{lem-est-van} the tightness of $(A^n,B^n,v^n_{a})$. Additionally, Lemma
  \ref{lem:bound-deltav} implies that the limit of $(A^n,B^n,v^n_{a})$
  coincides with that of $(A^n,B^n,\hatv^n_a)$, namely $ (A,B,v_a)$. This
  implies the $C$-tightness of $(A^n,B^n,v^n_{a})$ and thus the tightness of
  $(A^n,B^n,v^n_{a},v_b^n)$ by Corollary \ref{cor:C-tight-tight}. Finally, we
  verify that $(A^n,B^n,v_a^n,v_b^n)\Rightarrow (A,B,v_a,v_b)$ as
  {in Theorem~\ref{thm-volume-va-limit}, i.e., by once more
    referring to Lemma~\ref{lem:C1}. The diagonalization of the quadratic
    covariation in the limit is
    clear as the quadratic covariation is diagonal at each level $n$}.
\end{proof}

%%%%%%%%%%%%%%%%%%%%%%%%%%%%%%%

\section{Characterization of the limit price process---proof of the main theorem}
\label{new-model}

{So far, we have shown that the sequence of processes $(B^n,A^n,v_a^n,v_b^n)$ is
  C-tight. {As $Y^{r,n}$ is a continuous function of
    $v^{r,n}$ together with $R^n$, it follows} that
  $(B^n,A^n,v_a^n,v_b^n,Y^{a,n},Y^{b,n})$ is tight, {as well}. As a result, any accumulation point $(Y^{r})$ of $(Y^{r,n})$ is of the form:
\begin{equation}
Y^{r}_t=\langle v_r(t,\cdot),\,\varphi^r(\cdot-R_t) \rangle \label{ODE-yab}
\end{equation}
where $(A,B)$ is a weak accumulation point of the sequence of price processes. In this section we first characterize the process $(A,B)$; then we characterize the full limiting dynamics and prove convergence to a unique limit. 

\subsection{Convergence of the limiting price process}

In order to characterize the limiting price dynamics notice that the price processes satisfy
\begin{equation} \label{R1}
\begin{split}
\overline{R}^n_t &= R^n_0+{\Delta x^n}\sum_{i=1}^{\lfloor nt\rfloor} \xi^n_{r,i} \\
& =R^n_0+\int_0^t b_r(\overline{B}^n_{s},\overline{A}^n_s,\overline{Y}^{b,n}_s,\overline{Y}^{a,n}_s)\,ds
+ M^{n}_r(t)
+ S^n_r(t),
\end{split}
\end{equation}
with
\begin{align*}
M^{n}_r(t) & \coloneqq {\Delta x^n}\sum_{i=1}^{\lfloor nt\rfloor} \left(\xi^n_{r,i}-E_{\mathscr{F}^n_{\frac{i-1}{n}}} \xi^n_{r,i}\right) \\
S^n_r(t)
& \coloneqq
\int_0^t b^n_r(\overline{B}^n_{s-},\overline{A}^n_{s-},\overline{Y}^{b,n}_{s-},\overline{Y}^{a,n}_{s-})\,ds
-\int_0^t b_r(\overline{B}^n_{s},\overline{A}^n_s,\overline{Y}^{b,n}_s,\overline{Y}^{a,n}_s)\,ds.
\end{align*}
  Denoting $\barZ^n_s \coloneqq (\barB^n_s, \barA^n_s, \barY^{b,n}_s,
  \barY^{a,n}_s)$, we have
  \begin{align*}
    E\left[\abs{S^n_r(t)}^2 \right] 
    &\le C E\left[ \abs{ \int_0^t
        \left(b^n_r(\barZ^n_{s-}) - b_r(\barZ^n_{s-}) \right) ds }^2 \right] + C
    E\left[\abs{\int_0^t \left( b_r(\barZ^n_{s-}) - b_r(\barZ^n_s) 
        \right) ds}^2 \right] + o(1)\\
    &\le
    % C \norm{\nabla b^n_r}_{L^\infty}^2 \int_0^t \norm{\barZ^n_{s-} -
      %\barZ^n_s}_{L^2(\Omega)}^2 ds + 
      C \norm{b^n_r - b_r}^2_{L^\infty} + o(1),
  \end{align*}
  where in view of the fact $b_r\in C(\R^4;\R)$ and the continuity of limit process $\left( B,A,Y^b,Y^a  \right)$, we apply dominated convergence theorem to the second term  on the right-hand side of the first inequality.
In view of Assumption \ref{ass-ODE}, this implies
$\lim_{n\rightarrow\infty}E|S^n_r(t)|^2=0$ for any $t>0$. By
Lemma~\ref{lem:C1}, the martingale $M^n_r$ converges in distribution to a
martingale $M_r$ $(r=a,b)$ with quadratic co-variation 
\begin{equation}\label{eq:price-covariation}
  \qvar{(M^b,M^a)}_t = 
  \begin{pmatrix}
    \int_0^t \abs{\sigma_b(B_s, A_s, Y^b_s, Y^a_s)}^2 ds & \int_0^t
    \sigma_a\sigma_b^\top(B_s, A_s, Y^b_s, Y^a_s) ds\\
    \int_0^t \sigma_a\sigma_b^\top(B_s, A_s, Y^b_s, Y^a_s) ds &
    \int_0^t \abs{\sigma_a(B_s, A_s, Y^b_s, Y^a_s)}^2 ds
  \end{pmatrix}, \quad t \ge 0.
\end{equation}
Indeed, condition~\eqref{eq:1} % and~\eqref{eq:2}
of the lemma is true by
Assumption~\ref{ass-ODE}, whereas condition~\eqref{eq:3} is clear from the
scaling $\dx^n = 1/\sqrt{n}$. Finally, \eqref{eq:4} is trivial as the jumps
are even uniformly bounded.

Since we have joint tightness of the drift and the martingale part in
(\ref{R1}) we conclude that the limiting price process must be of the form:
\begin{equation} \label{ODE-ab}
 \begin{split}
 A_t=&A_0+\int_0^tb_a(B_s,A_s,Y^b_s,Y^a_s)\,ds+ M^a_t
% \int_0^t\sigma_{a}(B_s,A_s,Y^b_s,Y^a_s)\,d\widetilde{W}(s)
 \\
 B_t=& B_0+\int_0^t b_b(B_s,A_s,Y^b_s,Y^a_s)\,ds+M^b_t
% \int_0^t\sigma_{b}(B_s,A_s,Y^b_s,Y^a_s)\,d\widetilde{W}(s)
 , \quad t\geq 0.
 \end{split}
\end{equation}

\subsection{Characterization of the limiting dynamics}

% Our independence assumptions on the sequences $(\xi^n_{r,i})_{i\in\mathbb{N}}$
% and $(\txi_{r,i})_{i\in\mathbb{N}}$ imply that the Wiener processes
% $\widetilde{W}$, $W_a$ and $W_b$ are independent.

{ It remains to characterize the full limiting dynamics.  As the respective bounded variation parts have already been determined, we need to prove that the martingale parts can be represented in terms of four independent Brownian motions. To this end, we fix finitely many test function $\phi^1_a, ..., \phi^m_a, \phi^1_b, ..., \phi^l_b$, and consider the vector of processes } 
%
%quadratic covariation.}
%
%{
%More precisely, we can once again apply Lemma~\ref{lem:C1} together with
%Corollary~\ref{cor:C11}, this time for the full vector of processes
%$(B,A,Y^b,Y^a,v_a,v_b)$ or, rather, the full vector of processes 
\begin{equation*}
  \left(\barB^n,\barA^n, \ip{\barv^n_a}{\phi_a^1},
    \ldots, \ip{\barv^n_a}{\phi_a^m}, \ip{\barv^n_b}{\phi_b^1}, \ldots,
    \ip{\barv^n_b}{\phi_b^\ell}\right)
\end{equation*}
along with a weak accumulation point 
\begin{equation*}
  \left(B,A, \ip{v_a}{\phi_a^1},
    \ldots, \ip{v_a}{\phi_a^m}, \ip{v_b}{\phi_b^1}, \ldots,
    \ip{v_b}{\phi_b^\ell}\right).
\end{equation*}
Since $\barY^{b,n}$ and $\barY^{a,n}$ are obtained by integrating the volume
densities against test functions, there is no loss in generality in not
including them in the above vector. There is also no loss in generality in
assuming that all test functions are strictly positive.

Let $Z^{i}_r$ denote the martingale part of the process
$\ip{v_r}{\phi_r^i}$. From Proposition~\ref{prop-volume-limit-V3},
Theorem~\ref{thm-volume-LM}  and the independence of the Poisson processes
$N^n$ and $\tilde N^n_r$ $(r=a,b)$ we conclude that for $r,\tilde{r}\in\{a,b\}$,
\begin{equation} \label{QV}
	\sip{Z^{i}_r}{Z^{j}_r}_t = \int_0^t \sigma(\phi_r^i)(R_s)
        \sigma(\phi^j_r)(R_s) ds, \quad  \sip{Z^{i}_a}{Z^{j}_b}_t = 0. 
	\quad \sip{Z^{i}_r}{M^{\tilde r}}_t = 0, \qquad t \geq 0.
\end{equation}

The covariance structure of the martingale parts is as in Corollary \ref{cor:C3+4} with $F_t$  being the matrix with rows $\sigma_a(B_t,A_t,Y^b_t, Y^a_t)$ and $\sigma_b(B_t,A_t,Y^b_t, Y^a_t)$, $\sigma^i_ t := \sigma(\phi^i_a)(A_t)$, and $\tau^l_ t := \sigma(\phi^l_b)(B_t)$. Since the test functions are strictly positive we conclude from that corollary that there exist independent Wiener processes $\widetilde{W}$, $W_a$ and $W_b$ ($\widetilde{W}$ being two-dimensional) such that the weak accumulation point has the same distribution as the the system of coupled SDEs 
\begin{align*}
  dA_t=& b_a(B_t,A_t, Y^b_t, Y^a_t) dt + \sigma_a(B_t,A_t, Y^b_t, Y^a_t) d
         \widetilde{W}_t;\quad A_0=a_0;\\
  dB_t=& b_b(B_t,A_t, Y^b_t, Y^a_t) dt + \sigma_b(B_t,A_t, Y^b_t, Y^a_t) d
         \widetilde{W}_t; \quad B_0=b_0;\\
  v_b(t,\cdot) =& v_{b,0}(\cdot) + \int_0^t\left(E[\omega^{\bP_b}_{1}]
                  f^{\bP_b}(\cdot-B_s) - E[\omega^{\bC_b}_1]
                  f^{\bC_b}(\cdot-B_s)v_b(s,\cdot)\right)\,ds
  \\ 
    & +\sqrt{2} E\left[\omega^{\bN_b}_1\right]\int_0^t
    f^{\bN_b}(\cdot-B_s)\,dW_b(s),\quad t\geq 0;\\
    v_a(t,\cdot) =& v_{a,0}(\cdot)+
    \int_0^t \left(E[\omega^{\bP_a}_{1}]f^{\bP_a}
                    (\cdot-A_s) - E[\omega^{\bC_a}_1]
      f^{\bC_a} (\cdot-A_s) v_a(s,\cdot)\right)\,ds \\
    & +\sqrt{2}E\left[\omega^{\bN_a}_1\right] \int_0^t
    f^{\bN_a}(\cdot-A_s)\,dW_a(s),\quad t\geq 0
\end{align*}
upon integration of the volume density functions with our test functions. In
particular, by Corollary \ref{cor:C3+4} the driving Wiener processes do not
depend on the choice of the test functions.

Standard results on infinite-dimensional stochastic equations
 \cite{DaPrato1992} guarantee that the above coupled system does indeed admit  a unique adapted solution $(B,A,Y^b,Y^a,v_a,v_b)$ in $L^2(\Omega;
 C([0,T];\mathbb{R}^4\times (L^2)^2))$ for any $T>0$. Since two $H^{-1}$-valued random variables have the same distribution if the inner products with respect to any finite collection of test functions have the same distribution, this shows that 
\[
	(B^n,A^n,Y^{b,n},Y^{a,n},v_a^n,v_b^n)\Rightarrow(B,A,Y^b,Y^a,v_a,v_b)
\]
and hence completes the proof of our main result. 

\begin{remark}
The ``volume at the top'' follows a 2-dimensional Brownian motion with drift; for $r=a,b$,
\begin{align*}
	Y^r_t = & \langle v_{r,0}(\cdot),\varphi\rangle+
    \int_0^t \left(E[\omega^{\bP_r}_{1}]\langle f^{\bP_r},\,\varphi^r \rangle+f^r_s\right)\,ds \\
     &+\sqrt{2}E[\omega^{\bN_r}_{1}]\int_0^t \langle f^{\bN_r},\,\varphi^r \rangle
     \,dW_r(s)-
     \int_0^t\langle v_r(s,\cdot),\,D\varphi^r (\cdot - R_s)\sigma_r(B_t,A_t,Y^b_t,Y^a_t)\,d\widetilde{W}(s)\rangle,\quad t\geq 0
\end{align*}
where
\begin{align*}
	f^r_t := & \left\langle v_r(t,\cdot),\, 
	\frac{1}{2}\text{tr}\left\{
	\sigma_r\sigma_r'(B_t,A_t,Y^b_t,Y^a_t)D^2\varphi^r (\cdot - R_t)\right\}- b_r(B_t,A_t,Y^b_t,Y^a_t)D\varphi^r (\cdot - R_t)\right\rangle
	\\
	&-E[\omega^{\bC_r}_{1}]
      \left\langle f^{\bC_r}(\cdot)v_r(t,\cdot),\,\varphi^r
      \right\rangle .
\end{align*}
\end{remark}

% {
%   \begin{remark}
%     \label{rem:reconstruction-BM}
%     From~\eqref{SDE-bar-Y} we can reconstruct the Brownian motion $W$ by
%     \begin{equation*}
%       W_t = \int_0^t \sigma(\phi)(R_s)^{-1} d\barY_s(\phi), \quad t\ge0,
%     \end{equation*}
%     provided that $\phi$ is chosen such that $\sigma(\phi) >
%     0$. Alternatively, the above representation holds \emph{until}
%     $\sigma(\phi)$ vanishes for the first time. Conversely, if we know that
%     $\barY(\phi)$ is a process with quadratic variation
%     \begin{equation*}
%       \langle \barY(\phi) \rangle_t = \int_0^t \sigma(\phi)(R_s)^2 ds, \quad t
%       \ge 0,
%     \end{equation*}
%     then we can use the above reconstruction as a \emph{definition} of the
%     Brownian motion $W$. This line of reasoning is employed in
%     Corollary~\ref{cor:C12}.
%   \end{remark}
% }

\appendix

\section{A result on the characterization of stochastic process limits}

In this appendix we establish a result on the characterization of stochastic process limits in terms of Brownian integrals. Specifically, we assume that we are given a sequence of stochastic processes $(X^n, Z^n)$ (piece-wise
constant with jump times $k/n$) with
\begin{equation*}
  Z^n_t \coloneqq \sum_{k=1}^{\floor{nt}} \dZ^n_k,
\end{equation*}
such that
\begin{align*}
	E_{\mathcal{B}^n_{k/n}}[\dZ^n_k] & = 0 \\
  E_{\mathcal{B}_{k/n}^n}\left[ \dZ^n_k (\dZ^n_k)^\top \right] &= \f{1}{n}
  \sigma_n \sigma_n^\top\left(X^n_{k/n}, \f{k}{n} \right)
\end{align*}
where
$\mathcal{B}^n_{k/n} \coloneqq \sigma\left( X^n_0, \ldots, X^n_{k/n}, \dZ^n_1,
  \ldots, \dZ^n_{k-1} \right)$
and the processes may be multi-dimensional. 

\begin{assumption}
  \label{ass:1}
  Let $\sigma$ be a continuous function and assume that the following
  assumptions hold (for any fixed $t > 0$ where appropriate):
  \begin{gather}
     \label{eq:1}
     \norm{\sigma_n - \sigma}_{L^{\infty}} \xrightarrow{n\to\infty} 0, \quad
     \|\sigma\|_{L^{\infty}} < \infty,\\ 
     \label{eq:3}
    E\left[ \sum_{k=1}^{\floor{nt}} \abs{\dZ^n_k}^4 \right]
    \xrightarrow{n\to\infty} 0,\\
    \label{eq:4}
    \sup_{n \in \N} E\left[ \sup_{k \le \floor{nt}} \abs{\dZ^n_k} \right] <
    \infty.
   \end{gather}
\end{assumption}

Notice that~\eqref{eq:1} directly implies
\begin{equation}
  \label{eq:2}
  E\left[  \f{1}{n} \abs{\sum_{k=1}^{\floor{nt}}
      \sigma_n\sigma_n^\top\left(X_{k/n}^n,\f{k}{n}\right) - 
      \sum_{k=1}^{\floor{nt}} \sigma\sigma^\top\left(X_{k/n}^n,\f{k}{n}\right) }
      \right] \xrightarrow{n \to \infty} 0,
\end{equation}

\begin{lemma}
  \label{lem:C1}
  Suppose that $(X^n)$ is C-tight and that there are stochastic processes $X$
  and $Z$ defined on some probability space such that $(X^n,Z^n) \Rightarrow
  (X,Z)$. If Assumption~\ref{ass:1} is satisfied, then $Z$ has quadratic
  variation
  \begin{equation*}
    \qvar{Z}_t = \int_0^t \sigma(X_s,s) \sigma(X_s,s)^\top ds, \quad t \ge 0.
  \end{equation*}
  Moreover, $Z$ is a martingale w.r.t.~the filtration generated by $X$ and
  $Z$.
\end{lemma}
\begin{proof}
  By~(\ref{eq:4}) the martingales $Z^n$ satisfy the condition of Jacod and
  Shiryaev~\cite[Corollary VI.6.30]{JacodShiryaev2002}. Therefore, we have
  that both $Z^n$ and their quadratic covariation processes $\qvar{Z^n}$
  converge weakly and that the limit of $\qvar{Z^n}$ is the quadratic
  covariation of the limiting process $Z$ of the sequence $Z^n$. Symbolically,
  \begin{equation*}
    \left( Z^n, \qvar{Z^n} \right) \xRightarrow{n\to\infty} \left( Z, \qvar{Z}
    \right).
  \end{equation*}
  By C-tightness of $X^n$, we may add $X^n$ to the convergence and obtain
  \begin{equation*}
    \left( X^n, Z^n, \qvar{Z^n} \right) \xRightarrow{n\to\infty} \left( X, Z,
      \qvar{Z} \right).
  \end{equation*}
  Hence, we are left with identifying
  \begin{equation*}
    \qvar{Z}_t = \lim_{n\to\infty} \sum_{k=1}^{\floor{nt}} \dZ^n_k(\dZ^n_k)^\top.
  \end{equation*}

  To this end, by Skorokhod's lemma, we may assume (changing probability
  spaces as needed) that $(X^n, Z^n, \qvar{Z^n}) \to (X,Z,\qvar{Z})$ a.s.~Note that
    \begin{multline*}
    E\left[ \abs{\qvar{Z^n}_t - \int_0^t (\sigma \sigma^\top)(X_s,s)ds}
    \right] \le\\
 E\left[ \abs{\sum_{k=1}^{\floor{nt}} \dZ^n_k(\dZ^n_k)^\top -
          \f{1}{n} \sum_{k=1}^{\floor{nt}} \sigma \sigma^\top
          \left(X_{k/n},\f{k}{n}\right)} \right] + E\left[ \abs{\f{1}{n}
          \sum_{k=1}^{\floor{nt}} \sigma\sigma^\top\left(X_{k/n},\f{k}{n}\right) -
          \int_0^t \sigma\sigma^\top(X_s,s)ds } \right].
  \end{multline*}
  Convergence of the second term to $0$ follows immediately from dominated
  convergence using continuity of $\sigma$ and (\ref{eq:1}). We continue to
  further split up the first term:
  \begin{align*}
    E\left[ \abs{\sum_{k=1}^{\floor{nt}} \dZ^n_k(\dZ^n_k)^\top - \f{1}{n}
    \sum_{k=1}^{\floor{nt}} \sigma\sigma^\top\left(X_{k/n},\f{k}{n}\right)}
    \right]
    &\le ~~  E\left[ \abs{\sum_{k=1}^{\floor{nt}}  \dZ^n_k(\dZ^n_k)^\top  -
      \f{1}{n} \sum_{k=1}^{\floor{nt}}
      \sigma_n\sigma_n^\top\left(X_{k/n}^n,\f{k}{n}\right)} \right] \\
    &\quad + E\left[ \abs{ \f{1}{n} \sum_{k=1}^{\floor{nt}}
      \sigma_n\sigma_n^\top\left(X_{k/n}^n,\f{k}{n}\right) - \f{1}{n}
      \sum_{k=1}^{\floor{nt}} \sigma\sigma^\top\left(X_{k/n}^n,\f{k}{n}\right) }
      \right] \\
    &\quad+ E\left[ \abs{ \f{1}{n}
      \sum_{k=1}^{\floor{nt}} \sigma\sigma^\top\left(X_{k/n}^n,\f{k}{n}\right)
      - \f{1}{n}
      \sum_{k=1}^{\floor{nt}} \sigma\sigma^\top\left(X_{k/n},\f{k}{n}\right) }
      \right]\\
    &\eqqcolon \text{I+II+III}.
  \end{align*}

  Regarding I, note that the sequence of random variables
  \begin{equation*}
    C^n_k \coloneqq \dZ^n_k(\dZ^n_k)^\top - \f{1}{n}
    \sigma_n\sigma_n^\top\left(X^n_{k/n}, \f{k}{n} \right)
  \end{equation*}
  satisfy $E[C^n_k] = 0$ and $\operatorname{cov}(C^n_k, C^n_l) = 0$ if $k
  \neq l$. Moreover, since $C^n_k$ is obtained from $\dZ^n_k(\dZ^n_k)^\top$ by
  subtracting a conditional expectation, the fourth moment of $\dZ^n_k$ is an
  upper bound of the variance of $C^n_k$---where
  \begin{equation*}
    \operatorname{cov}(C^n_k, C^n_l) \coloneqq E\left[ \ip{C^n_k}{C^n_l}
    \right], \quad \var[C^n_k] \coloneqq E\left[ \ip{C^n_k}{C^n_k}
    \right]
  \end{equation*}
  for the standard inner product on the space of matrices. Hence, by Jensen's inequality and~(\ref{eq:3})
  \begin{equation*}
    \text{I} \leq \sqrt{E\left[ \abs{\sum_{k=1}^{\floor{nt}} C^n_k}^2 \right]} =
    \sqrt{\sum_{k=1}^{\floor{nt}} \var[C^n_k]} 
    \le \sqrt{\sum_{k=1}^{\floor{nt}} E\left[ \abs{\dZ^n_k}^4 \right]} \xrightarrow{n\to\infty} 0.
  \end{equation*}

  II converges to $0$ by~(\ref{eq:2}). For III, note that the integrand
  converges a.s.~by the convergence of $X^n$ to $X$, and convergence of the
  expectation follows by dominated convergence.

  Finally, note that if $(X,Z,\qvar{Z}) = \left(X,Z,\int_0^\cdot
    \sigma\sigma^\top(X_s,s) ds \right)$ \emph{in law}, then we really must
  have $\qvar{Z} = \int_0^\cdot \sigma\sigma^\top(X_s,s) ds$ as random
  variables, i.e., the proposed equality actually also holds on the original
  probability space before applying Skorokhod's lemma.
   
  We are left to prove that the limiting process $Z$ is a (local) martingale
  w.r.t.~the filtration generated by $(X,Z)$. Note that this will follow by a
  combination of \cite[Proposition IX.1.10 and IX.1.12]{JacodShiryaev2002} if
  we can show uniform integrability 
  of the family $(Z^n_t)_{n\in\N^+;\,t\in [0,T]}$ of random variables for
  arbitrary intervals $[0,T]$. This follows from~\eqref{eq:1} as 
  \begin{equation*}
  \sup_n \sup_{t \in [0,T]} E\left[  \abs{Z^n_t}^2 \right] \le 
  \sup_n \sum_{k=1}^{\floor{nT}} E\left[ \abs{\dZ^n_k}^2 \right]
  \le \sup_n \f{1}{n} \sum_{k=1}^{\floor{nT}} \norm{\sigma_n}^2_{L^{\infty}} 
    \le
    \left(\sup_n \norm{\sigma_n}^2_{L^{\infty}}\right) T < \infty. \qedhere
  \end{equation*}
\end{proof}

The preceding lemma suggests that $Z$ can be represented as a Brownian integral. As the quadratic variation of a martingale does not determine its distribution in general we now prove that we can find indeed a multi-dimensional Brownian motion $W$ such that
\begin{equation*}
  Z_t = Z_0 + \int_0^t \sigma(X_s, s) dW_s.
\end{equation*}
While probably standard, we have not been able to find a reference for this
statement directly applicable to our situation. Therefore, we give a formal
proof of the special case needed for the representation step in the main
theorem.
\begin{corollary}
  \label{cor:C3+4}
  Let $Z = (Z^A, Z^B, Z^C)$ be a continuous local martingale taking values in
  $\R^{d+n+m}$ such that the differential quadratic co-variation satisfies
  \begin{equation*}
    d \qvar{Z}_t =
    \begin{pmatrix}
      A_t & 0 & 0\\
      0 & B_t & 0\\
      0 & 0 & C_t
    \end{pmatrix}
    dt,
  \end{equation*}
  where $A_t = F_t F_t^\top$ for a $d\times d$-dimensional invertible process
  $F$ and where
  \begin{gather*}
    B^{i,j}_t = \sigma^i_t \sigma^j_t, \quad i,j = 1, \ldots, n,\\
    C^{l,k}_t = \tau^l_t \tau^k_t, \quad l,k = 1, \ldots, m,
  \end{gather*}
  for processes $\sigma, \tau$ taking values in $\R_{>0}^n$ and $\R^m_{>0}$,
  respectively. Then we can find a $(d+2)$-dimensional standard Brownian
  motion $(W,U,V)$ such that
  \begin{align*}
    Z^A_t &= Z^A_0 + \int_0^t F_s dW_s,\\
    Z^B_t &= Z^B_0 + \int_0^t \diag\left( \sigma^1_s, \ldots, \sigma^n_s
    \right) dU_s,\\
    Z^C_t &= Z^C_0 + \int_0^t \diag\left( \tau^1_s, \ldots, \tau^m_s
    \right) dV_s.
  \end{align*}
\end{corollary}

The proof of Corollary~\ref{cor:C3+4} builds on the following multi-variate
extension of L\'{e}vy's characterization of Brownian motion. The result
appears to be standard; we provide a proof (taken from \cite{lowther}) for
completeness.
\begin{theorem}
  \label{thr:generalized-levy}
  Let $X$ be an $l$-dimensional continuous local martingale with quadratic
  covariation $\langle X \rangle_t = \Sigma_t$ and $X_0 = 0$. Suppose that
  $\Sigma$ is deterministic, $\Sigma_0 = 0$ and for any $a \in \R^d$ we have
  $t \mapsto a^\top \Sigma_t a$ is continuous and increasing. Then for any $0
  \le s < t$ the increment $X_t - X_s$ is independent of $\mathcal{F}_s$ and
  distributed according to $\mathcal{N}\left(0, \Sigma_t - \Sigma_s \right)$.
\end{theorem}
\begin{proof}
Choose ${a\in{\mathbb R}^d}$ and set ${Y=a^{\rm T}X}$, so that ${[Y]_t=a^{\rm T}\Sigma_t a}$. The process
\begin{align*}
	M_t &=f(Y_t,[Y]_t) \equiv\exp\left(iY_t+\frac{1}{2}[Y]_t\right) \\
	&= \exp\left(ia^{\rm T}X_t+ \frac{1}{2}a^{\rm T}\Sigma_t a\right) 
\end{align*}
is bounded by ${\vert M_t\vert\le\exp(a^{\rm T}\Sigma_t a/2)}$. Applying It\^{o}'s lemma for continuous semimartingales to $f$ gives
\begin{align*}
	dM_t &= f_1(Y_t,[Y]_t)\,dY_t+f_2(Y_t,[Y]_t)\,d[Y]_t+\frac{1}{2}f_{11}(Y_t,[Y]_t)\,d[Y]_t \\ 
	&= iM_t\,dY_t.
\end{align*} 
As a bounded local martingale on $[0,T]$, $M$ is a (true) martingale. So,
\begin{align*}
	{E}[\exp(ia^{\rm T}(X_t-X_s))\mid\mathcal{F}_s]& ={E}[M_t\exp(-ia^{\rm T}X_s-a^{\rm T}\Sigma_t a/2)\mid\mathcal{F}_s] \\ 
	&=M_s\exp(-ia^{\rm T}X_s-a^{\rm T}\Sigma_t a/2) \\ 
	&=\exp(a^{\rm T}(\Sigma_s-\Sigma_t)a/2). 
\end{align*} 
This is the characteristic function of the multivariate normal, independently of ${\mathcal{F}_s}$, with mean zero and covariance matrix ${\Sigma_t-\Sigma_s}$, as required.
\end{proof}

\begin{proof}[Proof of Corollary~\ref{cor:C3+4}]
  Define processes $W$, $\tU$, $\tV$ taking values in $\R^d$, $\R^n$, $\R^m$,
  respectively, by
  \begin{gather*}
    W_t \coloneqq \int_0^t F_s^{-1} dZ^A_s,\\
    \tU_t \coloneqq \int_0^t \diag\left( (\sigma^1_s)^{-1}, \ldots,
      (\sigma^n_s)^{-1} \right) dZ^B_s,\\
    \tV_t \coloneqq \int_0^t \diag\left( (\tau^1_s)^{-1}, \ldots,
      (\tau^m_s)^{-1} \right) dZ^C_s.
  \end{gather*}
  We compute the quadratic covariation of the joint process $\left(W, \tU, \tV
  \right)$. For any $1 \le i,j \le d$ we have
  \begin{equation*}
    d\sip{W^i}{W^j}_t = \sum_{\nu,\mu=1}^d (F^{-1}_t)^{i,\nu}
    (F^{-1}_t)^{j,\mu} d\ip{Z^i}{Z^j}_t
    = \sum_{\nu,\mu=1}^d (F^{-1}_t)^{i,\nu} A_t^{i,j} (F^{-1}_t)^{j,\mu} dt % =
    % \left(F^{-1}_t F_t F_t^\top (F^{-1}_t)^\top  \right)^{i,j} dt
    = \delta^{i,j} dt.
  \end{equation*}
  On the other hand, using the structure of $B_t$ and $C_t$, respectively, we
  obtain for any $1 \le i,j \le n$ and $1 \le l,k \le m$
  \begin{equation*}
    d \sip{\tU^i}{\tU^j}_t = d \sip{\tV^l}{\tV^k}_t = dt.
  \end{equation*}
  The cross terms $\sip{W^i}{\tU^j}$, $\sip{W^i}{\tV^l}$, $\sip{\tU^j}{\tV^l}$
  vanish. Hence, the quadratic covariation of the process $(W, \tU, \tV)$ is
  the deterministic matrix-valued process
  \begin{equation*}
    \Sigma_t = t
    \begin{pmatrix}
      I_d & 0 & 0\\
      0 & E_n & 0\\
      0 & 0 & E_m
    \end{pmatrix},
  \end{equation*}
  where $E_k$ denotes the $k\times k$ matrix with all entries equal to $1$. As
  one can immediately see that $t \mapsto a^\top \Sigma_t a$ is continuous and
  increasing for any $a \in \R^{d+n+m}$, Theorem~\ref{thr:generalized-levy}
  implies that $\left( W, \tU, \tV \right)$ is a Gaussian process with
  increments distributed according to $\mathcal{N}(0, \Sigma_t-\Sigma_s)$.

  The special structure of the matrices $\Sigma_t$ implies that $W$ is a
  $d$-dimensional standard Brownian motion, whereas all the components of
  $\tU$ and $\tV$ are, respectively, identical one-dimensional Brownian
  motions. Hence, we may choose $U \coloneqq \tU^1$, $V \coloneqq \tV^1$, and
  obtain the conclusion.
\end{proof}

\begin{remark}
  The conditions of Corollary~\ref{cor:C3+4} can clearly be relaxed. For
  instance, it is enough that for any time $t$ at least one of the
  non-negative processes $\sigma^1, \ldots, \sigma^n$ is strictly positive. On
  the other hand, if all of them vanish identically, then we may not be able
  to find a suitable Brownian motion on the same probability space. In the
  non-regular case, we therefore need to weaken the statement to an equality
  in distribution, and use techniques similar to~\cite{Kushner-1974} to derive
  the result.
\end{remark}

\section{Technical proofs}
\label{sec:appdx-prf-lem}

\begin{proof}[Proof of Lemma~\ref{lem-limit-L2-Vn12}]
We prove \eqref{eq-limit-vn2a}; the second assertion follows similarly.
  Without any loss of generality, we assume $E[\omega^\bC_1]=1$.
  For each $s\in (\frac{1}{n},t)$ with $n\in\mathbb{N}^+$, we choose $k^n_s\in\mathbb{Z}$ such that $s\in[\f{k^n_s+1}{n},\f{k^n_s+2}{n})$. For $s\in(0,\frac{1}{n})$, put $k_s^n=0$. For notational simplicity, we set
  $\widetilde{v}^n(s,x)=1-\alpha+\alpha\barv^n(s,x)$, with $\alpha\in\{1,0\}$. Then
  \begin{align*}
    &\sup_{x\in\R}E\left|\int_0^{t}\left(g^n(s,x)-E[\omega_{1}^\bC]f^\bC(x-R_s)
    \right)\widetilde{v}^n(s,x)\,ds\right|^2  \\
    \leq\,&2\sup_{x\in\R}
        E\left|\int_0^{t} \left(
    f^\bC(x-\bar R^n_{\frac{k^n_s}{n}})-f^\bC(x-R_s)  \right)\widetilde{v}^n(s,x) \,ds\right|^2
    +
    2\sup_{x\in\R}
        E\left|\int_0^{t} \left(
    g^n(s,x)-f^\bC(x-\bar R^n_{\frac{k^n_s}{n}})
    \right)\widetilde{v}^n(s,x)ds\right|^2
    \\
    :=\,&2 (\Gamma_1+\Gamma_2).
  \end{align*}
  Since $f^\bC$ is Lipschitz continuous and vanishes outside a compact interval there exists a constant $C < \infty$ such that
  \begin{align*}
    \Gamma_1
    =\,& \sup_{x\in\R}
    E\left|\int_0^{t} \left(
    f^\bC(x-\bar R^n_{\frac{k^n_s}{n}})-f^\bC(x-R_s)  \right)\widetilde{v}^n(s,x) \,ds\right|^2
    \\
    \leq \,&C \sup_{x\in\R}E\int_0^t |\widetilde{v}^n(s,x)|^2ds\, E \int_0^t |R_s-\bar R^n_{\f{k_s^n}{n}}|^2\wedge 1 \,ds .
  \end{align*}
  Hence, by Lemma \ref{lem:tightness-barv}, $\Gamma_1 \to 0$ as $n \to \infty$ by dominated convergence, due to the a.s.~continuity $A$.
  Using independence of cancellation price levels and volumes, a direct computation yields:
  \begin{align*}
    \Gamma_2
    =\,&\sup_{x\in\R}
    E\left|\int_0^{t}
    \left(
    g^n(s,x)-f^\bC(x-\bar R^n_{\frac{k^n_s}{n}})
    \right)\widetilde{v}^n(s,x)\,ds\right|^2
    \\
        =\,&  \sup_{x\in\R} E\,
         \left|\int_0^{t}
    \left(
             \sum_{i=N^n(\ttau^n_{k_s^n})+1}^{N^n(\ttau^n_{k_s^n+1})}
            \sum_{j\in\Z}
            \indic{[x_j^n,x_{j+1}^n)}(\pi^\bC_{i}+\bar R^n_{\f{k_s^n}{n}})
            \omega_{i}^\bC\indic{[x_j^n,x_{j+1}^n)}(x)
            \f{\dv^nn}{\dx^n}
    -f^\bC(x-\bar R^n_{\f{k_s^n}{n}})  \right)\widetilde{v}^n(s,x)\,ds\right|^2
    \\
  \leq \,&
  3\sup_{x\in\R} E\,
         \left|\int_0^{t}\!\!
            \sum_{i=N^n(\ttau^n_{k_s^n})+1}^{N^n(\ttau^n_{k_s^n+1})}
        \sum_{j\in\Z}\Big(
        \indic{[x_j^n,x_{j+1}^n)}(\pi^\bC_{i}+\bar R^n_{\f{k_s^n}{n}})
        \omega^\bC_{i}
        -\int_{[x_j^n,x_{j+1}^n)}\!\!\!f^{\bC}(y-\bar R^n_{\f{k_s^n}{n}})\,dy
        \Big)\indic{[x_j^n,x_{j+1}^n)}(x)
        \f{\dv^n\widetilde{v}^n(s,x)n}{\dx^n}  ds\right|^2
  \\
  &+
  3 \sup_{x\in\R} E\,
         \left|\int_0^{t}\!\!
             \sum_{i=N^n(\ttau^n_{k_s^n})+1}^{N^n(\ttau^n_{k_s^n+1})}
            \left(\sum_{j\in\Z}\frac{1}{\dx^n}
            \int_{[x_j^n,x_{j+1}^n)}\!\!f^{\bC}(y-\bar R^n_{\f{k_s^n}{n}})\,dy
        \indic{[x_j^n,x_{j+1}^n)}(x)
    -f^\bC(x-\bar R^n_{\f{k_s^n}{n}})  \right) n\dv^n\widetilde{v}^n(s,x)  ds\right|^2
    \\
    &+3 \sup_{x\in\R} E\,
         \left|\int_0^{t} \left(
         \left(N^n(\ttau^n_{k_s^n+1})-N^n(\ttau^n_{k_s^n})\right)
         {n\dv^n}-1\right) f^\bC(x-\bar R^n_{\f{k_s^n}{n}})\widetilde{v}^n(s,x)
         \,ds\right|^2
    \\
  :=\,&3\Big(\gamma_0+\gamma_1+\gamma_2\Big).
  \end{align*}
 To estimate $\gamma_0$ we use again independence of involved random variables, the fact that
 \[
 	E_{\mathcal{F}^n_{\f{k_s^n}{n}}}
            \Big[\indic{[x_j^n,x_{j+1}^n)}(\pi^\bC_{i}+\bar R^n_{\f{k_s^n}{n}})
           \omega^\bC_{i} \Big]
           = \int_{[x_j^n,x_{j+1}^n)}f^{\bC}(y+\bar R^n_{\f{k_s^n}{n}})\, dy
 \]
along with Lemmas \ref{lem:aux-poisson-1} and \ref{lem:tightness-barv} and the properties of the scaling constants to conclude that:
\begin{align*}
    \gamma_0
    \leq \,& C t^{{2}}\sup_{x\in\R}E\sup_{s\in[0,t]}|\widetilde{v}^n(s,x)|^2
  \f{\lambda^n}{\mu^n}\bigg(\frac{n\dv^n}{\dx^n}\bigg)^2 \|f^\bC \|_{L^{\infty}}\dx^n \leq\,C t^{{2}}\left(t^{2}+t+1\right)\dx^n\longrightarrow 0,\quad \textrm{as }n\rightarrow \infty.
\end{align*}

To estimate $\gamma_1$ we first deduce from Lipschitz continuity of $f^\bC$ for $x \in [x^n_{j}, x^n_{j+1})$ that
  \begin{align*}
    \frac{1}{\dx^n}\int_{[x_j^n,x_{j+1}^n)}
    \big|f^\bC(y-\bar R^n_{\f{k_s^n}{n}})-f^\bC(x-\bar R^n_{\f{k_s^n}{n}})\big|\,dy
    \leq\,L \frac{1}{\dx^n}\int_{[x_j^n,x_{j+1}^n)} |\dx^n|\,dy = L \dx^n.
  \end{align*}
Thus,  using again Lemmas \ref{lem:aux-poisson-1} and \ref{lem:tightness-barv}, the properties of the scaling constants and the fact that $f^\bC$ vanishes outside a compact interval we find a constant $C < \infty$ such that:
\begin{align*}
    \gamma_1
  \leq\,
  &
  C t^{{2}} n^{-1} \sup_{x\in\R}E\sup_{s\in[0,t]}|\widetilde{v}^n(s,x)|^2
  \leq\,C t^{{2}}\left(t^{2}+t+1\right)n^{-1}
  \longrightarrow 0,\quad \textrm{as }n\rightarrow\infty.
\end{align*}

In view of Lemma \ref{lem:aux-poisson-1}, boundedness of $f^\bC$ and independence of involved random variables,  we have
\begin{align*}
  \gamma_2
        =&\,
        \sup_{x\in\R} E\,
         \left|\int_0^{t} \left(
         \left(N^n(\ttau^n_{k_s^n+1})-N^n(\ttau^n_{k_s^n})\right)
         {n\dv^n}-1\right) f^\bC(x-\bar R^n_{\f{k_s^n}{n}})\widetilde{v}^n(s,x)
         \,ds\right|^2\\
         \leq
         &\,2\sup_{x\in\R} E\,\left|
         \sum_{l=1}^{\lfloor nt\rfloor} \left(
         \left(N^n(\ttau^n_{l})-N^n(\ttau^n_{l-1})\right)
         {n\dv^n}-1\right) f^\bC(x-\bar R^n_{\f{l-1}{n}})\widetilde{v}^n(\frac{l-1}{n},x)
         \frac{1}{n}\right|^2
         \\
         &\,
         +
         2\sup_{x\in\R} E\,\left|\int_0^{\frac{1}{n}}
          f^\bC(x-\bar R^n_0)\widetilde{v}^n(0,x)\,ds
         \right|^2
         \\
         \leq & \,2(\dv^n)^2 E\,\sup_{x\in\R}
         \sum_{l=1}^{\lfloor nt\rfloor}\left|\left(
        N^n(\ttau^n_{l})-N^n(\ttau^n_{l-1}) -E[ N^n(\ttau^n_{l})-N^n(\ttau^n_{l-1})] \right)
          f^\bC(x-\bar R^n_{\f{l-1}{n}})\widetilde{v}^n(\frac{l-1}{n},x)
         \right|^2
         +C/n\\
         \leq&\,
          \frac{C}{n}\left(1+ t \sup_{x\in\R} E \sup_{s\in[0,t]}|\widetilde{v}^n(s,x)|^2\right)  \longrightarrow 0,\quad \textrm{as }n\rightarrow \infty.\qedhere
\end{align*}
\end{proof}

\hspace{5mm}

\begin{proof}[Proof of Lemma \ref{lem-est-van}]
Without any loss of generality, we take $s=0$ and drop the index $r$. 
First, we have
\begin{align}
  &E\left\| V^{n,1}(t)  \right\|_{L^2}^2
  \nonumber \\
  &
  =\left(\frac{\dv^n}{\dx^n}\right)^2\int_{\R}E\bigg|
  \sum_{i=1}^{N^n(t)}\sum_{j\in\Z}
  \omega^{\textbf{P}}_{i}\indic{[x_j^n,x_{j+1}^n)}(x) \indic{[x_j^n,x_{j+1}^n)}(R^n(\widetilde{\tau}^n_{\tN^n({\tau_{a,i}^n})})
  +\pi^{\textbf{P}}_{i})
   \bigg|^2 \,dx
  \nonumber \\
  &
  =\left(\frac{\dv^n}{\dx^n}\right)^2\int_{\R}
  \sum_{l=1}^{\infty}\f{(\lambda^nt)^l}{l!}e^{-\lambda^n t}
E_{N^n(t)=l}\Bigg[
  \sum_{i>i';i,i'=1}^{l}\nonumber\\
  &  2\left(\sum_{j\in\Z}
  \omega^{\textbf{P}}_{i}\indic{[x_j^n,x_{j+1}^n)}(x) \indic{[x_j^n,x_{j+1}^n)}(R^n(\widetilde{\tau}^n_{\tN^n({\tau_{a,i}^n})})
  +\pi^{\textbf{P}}_{i})\right)
  \left(\sum_{j\in\Z}
  \omega^{\textbf{P}}_{i'}\indic{[x_j^n,x_{j+1}^n)}(x) \indic{[x_j^n,x_{j+1}^n)}(R^n(\widetilde{\tau}^n_{\tN^n({\tau_{a,i'}^n})})
  +\pi^{\textbf{P}}_{i'})\right)\nonumber\\
  &
 +
 \sum_{i=1}^{l}\sum_{j\in\Z}
 E|\omega^{\textbf{P}}_{i}|^2
   \indic{[x_j^n,x_{j+1}^n)}(x) \indic{[x_j^n,x_{j+1}^n)}(R^n(\widetilde{\tau}^n_{\tN^n({\tau_{a,i}^n})})
   +\pi^{\textbf{P}}_{i})
   \Bigg]\,dx
  \nonumber\\
  &
  \leq \,C
  \left(\frac{\dv^n}{\dx^n}\right)^2\int_{\R}
  \sum_{l=1}^{\infty}\f{(\lambda^nt)^l}{l!}e^{-\lambda^n t}
E_{N^n(t)=l}\Bigg[
  \sum_{i<i';i,i'=1}^{l}
  2\big(E\omega^{\textbf{P}}_{1}\big)^2
  \indic{[-M+R^n(\widetilde{\tau}^n_{\tN^n({\tau_{a,i}^n})}),
  M+R^n(\widetilde{\tau}^n_{\tN^n({\tau_{a,i}^n})})]}(x)
  \|f^{\textbf{P}}\|^2_{L^{\infty}}(\dx^n)^2
  \nonumber\\
  &
  ~~ +E|\omega^{\textbf{P}}_{1}|^2
 \sum_{i=1}^{l}
\indic{[-M+R^n\widetilde{\tau}^n_{\tN^n({\tau_{a,i}^n})}),
M+R^n(\widetilde{\tau}^n_{\tN^n({\tau_{a,i}^n})}))}(x)
  \|f^{\textbf{P}_{a}}\|_{L^{\infty}}\dx^n
      \Bigg]\,dx
  \nonumber\\
  &
  \leq C \left(\frac{\dv^n}{\dx^n}\right)^2
  \sum_{l=1}^{\infty}\f{(\lambda^nt)^l}{l!}e^{-\lambda^n t}
\Bigg[
  l(l-1)
  \|f^{\textbf{P}}\|^2_{L^{\infty}}(\dx^n)^2
  +
 l  \|f^{\textbf{P}}\|_{L^{\infty}}\dx^n
      \Bigg]
  \nonumber\\
  &
  \leq C \left(\frac{\dv^n}{\dx^n}\right)^2\bigg[ (\lambda^nt\dx^n)^2+\lambda^nt\dx^n \bigg]
  \nonumber\\
  &
  \leq C (t^2+t), \nonumber
\end{align}
and similarly, we have
  $E\left\| V^{n,2}(t)  \right\|_{L^2}^2
  \leq C (t^2+t)$,
where the constants $C$ are independent of $n$. Taking the supremum norm $\|\cdot\|_{L^{\infty}}$ instead,  we obtain
$$
\sup_{x\in\R}E_{\F^n_s}|V^{n,1}(t)-V^{n,1}(s)|^2
+\sup_{x\in\R}E_{\F^n_s}|V^{n,2}(t)-V^{n,2}(s)|^2
\leq
C [t-s+(t-s)^2],\quad 0\leq s \leq t<\infty.
$$
On the other hand,
\begin{align*}
  &E\sup_{s\in[0,t]}\|V_{a}^{n,3}(s)\|_{L^2}^2
  \nonumber\\
    =\,&
  E\sup_{s\in[0,t]}
  \bigg\|
  \sum_{i=1}^{N^n(s)}\sum_{j\in\mathbb{Z}}
   \indic{[x_j^n,x_{j+1}^n)}(\cdot)
    \indic{[x_j^n,x_{j+1}^n)}(R^n(\widetilde{\tau}^n_{\tN^n({\tau_{a,i}^n})})
    +\pi^{\textbf{P}}_{i})
    \txi_{a,\tN^n({\tau_{a,i}^n})+1}\sqrt{\dv^n}
  \bigg\|_{L^2}^2\nonumber\\
  =\,&
  E\sup_{s\in[0,t]}
  \Bigg\|
  \sum_{k=1}^{\tN^n(s)}
    \sum^{N^n(\widetilde{\tau}^n_k)}_{i=N^n(\widetilde{\tau}^n_{k-1})+1}\sum_{j\in\mathbb{Z}}
    \indic{[x^n_j,x^n_{j+1})} (\pi^{\textbf{P}}_{i}+R^n(\widetilde{\tau}^n_{k-1}))
     \indic{[x^n_j,x^n_{j+1})}(\cdot)\,
    \txi_{a,k} \sqrt{\dv^n}\nonumber\\
    &+
    \sum_{i=N^n(\widetilde{\tau}^n_{\tN^n(s)})+1}^{N^n(s)}\sum_{j\in\mathbb{Z}}
    \indic{[x^n_j,x^n_{j+1})} (\pi^{\textbf{P}}_{i}+R^n(\widetilde{\tau}^n_{\tN^n(s)}))
     \indic{[x^n_j,x^n_{j+1})}(\cdot)
    \txi_{a,\tN^n(s)+1} \sqrt{\dv^n}
    \Bigg\|_{L^2}^2
    \nonumber\\
    \leq \,&
  C\dv^n E\Bigg[
  \sum_{k=1}^{\tN^n(t)}\Bigg\|
    \sum^{N^n(\widetilde{\tau}^n_k)}_{i=N^n(\widetilde{\tau}^n_{k-1})+1}\sum_{j\in\mathbb{Z}}
    \indic{[x^n_j,x^n_{j+1})} (\pi^{\textbf{P}}_{i}+R^n(\widetilde{\tau}^n_{k-1}))
      \indic{[x^n_j,x^n_{j+1})}(\cdot)
    \Bigg\|_{L^2}^2\nonumber\\
    &+\Bigg\|
    \sum_{i=N^n(\widetilde{\tau}^n_{\tN^n(t)})+1}^{N^n(t)}\sum_{j\in\mathbb{Z}}
    \indic{[x^n_j,x^n_{j+1})} (\pi^{\textbf{P}}_{i}+R^n(\widetilde{\tau}^n_{\tN^n(t)}))
     \indic{[x^n_j,x^n_{j+1})}(\cdot)
    \Bigg\|_{L^2}^2\Bigg]
    \nonumber\\
  =\,&
  C\dv^n
  \sum_{l=0}^{\infty}\frac{(\mu^nt)^l}{l!}e^{-\mu^nt}
  E_{\tN^n(t)=l}
  \Bigg[
  \sum_{k=1}^{l}\Bigg\|\sum_{j\in\mathbb{Z}}
    \sum^{N^n(\widetilde{\tau}^n_k)}_{i=N^n(\widetilde{\tau}^n_{k-1})+1}
    \indic{[x^n_j,x^n_{j+1})} (\pi^{\textbf{P}}_{i}+R^n(\widetilde{\tau}^n_{k-1}))
      \indic{[x^n_j,x^n_{j+1})}(\cdot)
    \Bigg\|_{L^2}^2\nonumber\\
    &+\Bigg\|
    \sum_{i=N^n(\widetilde{\tau}^n_{l})+1}^{N^n(t)}\sum_{j\in\mathbb{Z}}
    \indic{[x^n_j,x^n_{j+1})} (\pi^{\textbf{P}}_{i}+R^n(\widetilde{\tau}^n_{l}))
    \indic{[x^n_j,x^n_{j+1})}
    \Bigg\|_{L^2}^2    \Bigg]
    \nonumber\\
  \leq \,&
  C\dv^n
  \sum_{l=0}^{\infty}\frac{(\mu^nt)^l}{l!}e^{-\mu^nt}
  E_{\tN^n(t)=l}\Bigg[\sum_{k=1}^l
  (N^n(\widetilde{\tau}^n_k)-N^n(\widetilde{\tau}^n_{k-1}))
  (N^n(\widetilde{\tau}^n_k)-N^n(\widetilde{\tau}^n_{k-1})-1)
  \|f^{\textbf{P}}\|_{L^{\infty}}^2 \\
  & \qquad (\dx^n)^2 \|\indic{[-M+R^n(\widetilde{\tau}^n_{l}),M+R^n(\widetilde{\tau}^n_{l}))} \|_{L^2}^2
  \nonumber
  \\
  &+\sum_{k=1}^l(N^n(\widetilde{\tau}^n_k)-N^n(\widetilde{\tau}^n_{k-1}))
  \|f^{\textbf{P}}\|_{L^{\infty}} \dx^n
  \| \indic{[-M+R^n(\widetilde{\tau}^n_{l}),M+R^n(\widetilde{\tau}^n_{l})]}\|_{L^2}^2\nonumber\\
&
+(N^n(t)-N^n(\widetilde{\tau}^n_{l}) ) (N^n(t)-N^n(\widetilde{\tau}^n_{l})-1)
  \|f^{\textbf{P}}\|_{L^{\infty}}^2 (\dx^n)^2 \|\indic{[-M+R^n(\widetilde{\tau}^n_{l}),M+R^n(\widetilde{\tau}^n_{l})]}\|_{L^2}^2
  \nonumber
  \\
  &+(N^n(t)-N^n(\widetilde{\tau}^n_{l}) )\|f^{\textbf{P}}\|_{L^{\infty}} \dx^n
  \| \indic{[-M+R^n(\widetilde{\tau}^n_{l}),M+R^n(\widetilde{\tau}^n_{l})]}\|_{L^2}^2
  \Bigg]
  \nonumber\\
  \leq\,&
   C\dv^n
  \sum_{l=0}^{\infty}\frac{(\mu^nt)^l}{l!}e^{-\mu^nt}
  E_{\tN^n(t)=l}\Bigg[
  \sum_{k=1}^l
  (N^n(\widetilde{\tau}^n_k)-N^n(\widetilde{\tau}^n_{k-1}))(N^n(\widetilde{\tau}^n_k)-N^n(\widetilde{\tau}^n_{k-1})-1)
  \ (\dx^n)^2
\nonumber\\
  &+\sum_{k=1}^l(N^n(\widetilde{\tau}^n_k)-N^n(\widetilde{\tau}^n_{k-1}))+(N^n(t)-N^n(\widetilde{\tau}^n_{l}) ) (N^n(t)-N^n(\widetilde{\tau}^n_{l})-1)
   (\dx^n)^2
 +(N^n(t)-N^n(\widetilde{\tau}^n_{l}) )
  \Bigg]\nonumber\\
 =\,&
 C\dv^n
  \sum_{l=0}^{\infty}\frac{(\mu^nt)^l}{l!}e^{-\mu^nt}
  E\Bigg[\nonumber\\
  &lN^n(\beta(1,l))(N^n(\beta(1,l))-1)\ (\dx^n)^2
  +lN^n(\beta(1,l))
+N^n(\beta(1,l))
   (\dx^n)^2
 +N^n(\beta(1,l))
\Bigg]\nonumber\\
=\,&
C\dv^n
  \sum_{l=0}^{\infty}\frac{(\mu^nt)^l}{l!}e^{-\mu^nt}
  \sum_{m=0}^{\infty}\left[ lm(m-1)(\dx^n)^2+m^2(\dx^n)^2+(l+1)m \right]
  \int_0^1\frac{(\lambda^ntz)^m}{m!}e^{-\lambda^ntz}\frac{(1-z)^{l-1}}{B(1,l)}dz
  \nonumber\\
  \leq\,&
  Ct\dv^n\left[ \frac{(\lambda^n\dx^n)^2}{\mu^n}+\lambda^n   \right]\nonumber\\
  \leq\,&
  Ct,\nonumber
\end{align*}
with the constant $C$ independent of $n$ and $t$.

The estimate of $ v^n_{a/b}$ follows precisely in the same
  way as the proof of Lemma~\ref{lem:tightness-barv}, taking into account the
  appropriate estimates for $ V^{n,i}_{a/b}$ derived above.
\end{proof}

\section{Classical tightness results}
\label{sec:two-class-tightn}

For the convenience of the reader, we recall some classical results
on tightness which the derivations of Section~\ref{sec:scaling-limit-volume}
are based on. We first note that though the following theorems and lemmas may
be originally established on finite time intervals, we state them on the half
line $[0,\infty)$ since there is no essential difficulty to make such
extensions in the spirit of Jacod and Shiryaev \cite{JacodShiryaev2002}.

The first result is a sufficient condition for tightness in the Skorokhod
space $\mathcal{D}([0,\infty); E)$ for a complete separable metric state space
$(E, \rho)$ due to Aldous and Kurtz. We take it from \cite[Th.~6.8]{wal86}.
\begin{theorem}
  \label{thr:aldous}
  Let $X_n$ be a sequence of processes taking values in $\mathcal{D}([0,\infty);
  E)$ such that the family $\left( X_n(t) \right)_{n\in\N^+}$ of random
  variables is tight (in $E$) for any rational $t$. Moreover, assume that for each $N\in\mathbb{N}^+$,
  there is a number $p > 0$ and processes $\left( \gamma_n(\delta) \right)_{\delta
    \in [0,\infty)}$, $n \in \N^+$, such that
  \begin{gather*}
    E\left[ \left. \rho\left( X_n(t+\delta), X_n(t) \right)^p \, \right| \,
      \F^n_t \right] \le E\left[ \gamma_n(\delta) \, | \, \F^n_t \right],\quad \forall\,t,t+\delta\in[0,N],\\
    \lim_{\delta\to0} \limsup_{n\to\infty} E\left[ \gamma_n(\delta) \right] = 0,
  \end{gather*}
  where the filtration $\F^n$ is generated by $X^n$. Then $(X_n)_{n\in\N^+}$ is
  tight in $\mathcal{D}\left([0,\infty); E \right)$.
\end{theorem}
\begin{proof}
  See \cite[Th.~6.8]{wal86}. Note that Walsh assumes one joint filtration
  $\F_t$, whereas we allow for filtrations depending on $n$. This difference
  is, however, inconsequential, e.g., by choosing $X^n$ to be defined on a
  common probability space in an independent way and then choosing $\F_t$ to
  be the filtration generated by all the filtrations $\F^n_t$.
\end{proof}

The following lemma on C-tightness is borrowed from \cite[Proposition 3.26, Page 351]{JacodShiryaev2002}.
\begin{lemma}\label{lem-c-tight}
For a sequence $X^n$ with paths in $\mathcal{D}([0,\infty);\mathbb{R}^d)$ $(d\in\mathbb{N}^+)$, it is C-tight if and only if it is tight and for all $N\in\mathbb{N}^+$, $\epsilon>0$, there holds
$$
\lim_{n\rightarrow\infty} \mathbb{P}^n\left( \sup_{t\leq N}|\Delta X^n_t|>\epsilon \right)=0.
$$
\end{lemma}

The main theoretical tool in this paper is Mitoma's theorem, on
basis of \cite[Th.~6.13, Lem. 6.14, Cor 6.16, Note on Page~365]{wal86}, which relates
tightness of distribution-valued processes to real-valued processes obtained
by applying test-functions. We specialize the general formulation given in
\cite{wal86} so that the theorem can be directly applied to our setting.

\begin{theorem}[Mitoma's theorem]
  \label{thr:mitoma}
  For any positive integer $d$, let $X^n:=(X^n_1,\cdots,X^n_d)$ be a sequence
  of processes with sample paths lying in
  $\mathcal{D}\left([0,\infty); \left(\mathcal{E}'\right)^d \right)$. The
  sequence $X^n$ is tight as processes with paths in
  $\mathcal{D}\left([0,\infty); \left(\mathcal{E}'\right)^d \right)$, if and
  only if for any $\phi_1,\cdots,\phi_d \in \mathcal{E}$ we have tightness of
  the sequence of $\mathcal{D}\left([0,\infty); \R \right)$-valued processes
  $\sum_{i=1}^d\ip{X_i^n}{\phi_i}$. If, furthermore, for any
  $\epsilon,N\in(0,\infty)$ there exists $\widetilde{N}\in (0,\infty)$ such
  that
  $\sup_{n}\mathbb{P} \left( \sup_{t\in[0,N]} \sum_{i=1}^d\|X^n_i(t)\|_{L^2} >
    \widetilde{N} \right) <\epsilon$,
  then $X^n$ is tight as a sequence of processes with paths in
  $\mathcal{D}\left([0,\infty); \left(H^{-1}\right)^d \right)$.
\end{theorem}

Here we choose $H^{-1}$ for convenience. Indeed, in view of the arguments in \cite[Page 335, Example 1a]{wal86}, we can replace the space $H^{-1}$ by $H^{-m}$ for any $m>1/2$. On the other hand, an immediate application of Theorem \ref{thr:mitoma} is the following corollary, which states that joint tightness of a pair of sequences of
stochastic processes follows from individual tightness assuming that at least
one of the involved sequences is $C$-tight, i.e., all its accumulation points
are continuous processes.

\begin{corollary}
  \label{cor:C-tight-tight}
  Let $Y^n$ and $Z^n$ be sequences of stochastic processes taking values in
   $\left(\mathcal{E}'\right)^d $ and $\left(\mathcal{E}'\right)^l$ respectively, with $d,l\in\mathbb{N}^+$. If $Y^n$ is
  $C$-tight with paths  in $\mathcal{D}\left([0,\infty); \left(\mathcal{E}'\right)^d \right)$ and $Z^n$ is tight with paths  in $\mathcal{D}\left([0,\infty); \left(\mathcal{E}'\right)^l\right)$, then the pair of processes
  $(Y^n, Z^n)$ is tight with paths in
  $\mathcal{D}\left([0,\infty); \left(\mathcal{E}'\right)^{d+l} \right)$.
\end{corollary}
\begin{proof}
  We fist note that for the finite-dimensional case where $\left(\mathcal{E}'\right)^d$ and $\left(\mathcal{E}'\right)^l$ are replaced by Euclidean spaces, Corollary  \ref{cor:C-tight-tight} coincides with
  \cite[Cor.~VI.3.33]{JacodShiryaev2002}. Obviously the $C$-tightness of $Y^n$   with paths  in $\mathcal{D}\left([0,\infty); \left(\mathcal{E}'\right)^d \right)$ implies that of $\sum_{i=1}^d\langle Y^n_i,\,\phi_i\rangle$   with paths  in $\mathcal{D}\left([0,\infty); \mathbb{R} \right)$ for any $\phi_1,\cdots,\phi_d\in\mathcal{E}$.    As Theorem \ref{thr:mitoma} allows us to prove the tightness of distribution-valued processes by verifying that of the  real-valued processes obtained by applying test-functions, there follows the tightness of pair of processes   $(Y^n, Z^n)$  with paths in
  $\mathcal{D}\left([0,\infty); \left(\mathcal{E}'\right)^{d+l} \right)$.
\end{proof}

 We remark that the method of proof for the finite-dimensional case (see \cite[Page 353, Cor.~VI.3.33]{JacodShiryaev2002}) cannot directly be applied to Corollary \ref{cor:C-tight-tight}, as the compactness of the unit ball is key to their proof of the finite-dimensional case. On the other hand, if we replace $\left(\mathcal{E}'\right)^d$ for $Y^n$ by $\R^m\times\left(\mathcal{E}'\right)^d$ with $m\in\N^+$, then Corollary \ref{cor:C-tight-tight} still holds, since the finite-dimensional space is isomorphic as well as homeomorphic to some subspace of $\mathcal{E}'$.

 We also use a lemma of Billingsley about weak limits under time-changes.
\begin{lemma}
  \label{lem:billingsley}
  Let $X^n$ be a sequence of processes taking values in
  $\mathcal{D}([0,\infty); E)$ for some separable metric space $E$ and let
  $\Phi^n$ be a sequence of non-decreasing processes with paths in
  $\mathcal{D}([0,\infty); [0,\infty))$. Assume that $(X^n, \Phi^n)$ converge
  weakly to a pair of processes $(X,\Phi) \in \mathcal{D}\left( [0,\infty); E
    \times [0,\infty) \right)$ such that $X \in C\left([0,\infty); E\right)$
  with probability $1$. Then
  \begin{equation*}
    X^n \circ \Phi^n \Rightarrow X \circ \Phi.
  \end{equation*}
\end{lemma}
\begin{proof}
  The proof in Billingsley~\cite[p.~151]{Billingsley_1999} (for the special
  case $E = \R$) can be immediately adapted to this more general setting.
\end{proof}

{
Finally, we recall Skorokhod's lemma (\cite[Theorem 6.7 on page
70]{Billingsley_1999}).
\begin{lemma}
  \label{lem:Skorokhod}
  Let $\mu_n \Rightarrow \mu$ be a weakly converging sequence of probability
  measures on a metric space such that the support of $\mu$ is separable. Then
  there is a probability space $\left(\Omega, \F, P\right)$ and a sequence of
  random variables $X_n$ with distribution $\mu_n$ together with a random
  variable $X$ with distribution $\mu$ such that
  \begin{equation*}
    \forall \omega \in \Omega:\ \lim_{n\to\infty} X_n(\omega) = X(\omega).
\end{equation*}
\end{lemma}
}

%\end{appendix}

\bibliographystyle{plain}

\begin{thebibliography}{10}

\bibitem{AJ}
Fr\'ed\'eric Abergel and Aymen Jeddi 
\newblock Long time behaviour of a Hawkes process-based limit order book.
\newblock {\em Preprint}, 2015.

%\bibitem{BayraktarHorstSircar}
%Erhan Bayraktar, Ulrich Horst, and Ronnie Sircar.
%\newblock Queuing theoretic approaches to financial price fluctuations.
%\newblock {\em Handb. Oper. Res. Manag. Sci.}, 15:637--677, 2007.

\bibitem{Biais}
Bruno Biais, Pierre Hillion, and Chester Spatt.
\newblock An empirical analysis of the limit order book and the order flow in
  the {P}aris bourse.
\newblock {\em J. Financ.}, 50(5):1655--1689, 1995.

\bibitem{Billingsley_1999}
Patrick Billingsley.
\newblock {\em Convergence of probability measures}.
\newblock Wiley Series in Probability and Statistics: Probability and
  Statistics. John Wiley \& Sons Inc., New York, second edition, 1999.
\newblock A Wiley-Interscience Publication.

%\bibitem{Carmona-Webster}
%Ren\'e Carmona and Kevon Webster 
%\newblock A belief-driven order book model
%\newblock Preprint, 2014.

\bibitem{Cebiroglu-Horst}
G\"okhan Cebiroglu and Ulrich Horst.
\newblock Optimal order display in limit order markets with liquidity competition.
\newblock {\em J. Econ. Dyn. Con,}, 58: 81-100, 2015.

\bibitem{ContDeLarrard2012b}
Rama Cont and Adrien {De Larrard}.
\newblock Order book dynamics in liquid markets: Limit theorems and diffusion
  approximations.
\newblock {\em Available at SSRN: http://ssrn.com/abstract=1757861}, 2012.

%\bibitem{ContDeLarrard2012a}
%Rama Cont and Adrien {De Larrard}.
%\newblock Price dynamics in a {Markovian} limit order market.
%\newblock {\em SIAM J. Financ. Math.}, 4(1):1--25, 2013.

\bibitem{ContStoikovTalreja}
Rama Cont, Sasha Stoikov, and Rishi Talreja.
\newblock A stochastic model for order book dynamics.
\newblock {\em Oper. Res.}, 58(3):549--563, 2010.

\bibitem{DaPrato1992}
Guiseppe Da~Prato and Jerzy Zabczyk.
\newblock {\em Stochastic equations in infinite dimensions}.
\newblock Encyclopedia of Mathematics and its Applications. Cambridge
  University Press, 2008.

%\bibitem{DuffieProtter}
%Darrell Duffie and Philip Protter.
%\newblock From discrete-to continuous-time finance: Weak convergence of the
%  financial gain process.
%\newblock {\em Math. Financ.}, 2(1):1--15, 1992.

\bibitem{EasleyOHara}
David Easley and Maureen O'Hara.
\newblock Price, trade size, and information in securities markets.
\newblock {\em J. Financ. Econ.}, 19(1):69--90, 1987.

%\bibitem{FoellmerSchweizer}
%Hans F{\"o}llmer and Martin Schweizer.
%\newblock A microeconomic approach to diffusion models for stock prices.
%\newblock {\em Math. Financ.}, 3(1):1--23, 1993.

%\bibitem{Garman}
%Mark~B Garman.
%\newblock Market microstructure.
%\newblock {\em J. Financ. Econ.}, 3(3):257--275, 1976.

\bibitem{Farmer}
J. Doyne Farmer and Laszlo Gillemot and Fabrizio Lillo and Szabolcs Mike and Anindya Sen.
\newblock What really causes large price changes?
\newblock {\em Quantitative Finance}, 4(4): 383-397, 2004.


\bibitem{GlostenMilgrom}
Lawrence~R. Glosten and Paul~R. Milgrom.
\newblock Bid, ask and transaction prices in a specialist market with
  heterogeneously informed traders.
\newblock {\em J. Financ. Econ.}, 14(1):71--100, 1985.

\bibitem{Gao}
Xuefeng Gao, J.G. Dai, A.B. Dieker and S.J. Deng.
\newblock Hydrodynamic limit of order book dynamics.
\newblock Preprint, 2014.

%\bibitem{Guo}
%Guo, X. and Ruan, Z. and Zhu, L.
%\newblock Dynamics of order positions and related queues in a limit order book.
%\newblock ArXiv e-print 1505.04810v1. 2015.

\bibitem{Hautsch-Huang}
Nikolaus Hautsch and Ruihong Huang.
\newblock The market impact of a limit order.
\newblock {\em J. Econ. Dyn. Control}, 36:501--522, 2012.

\bibitem{Horst-Kreher}
Ulrich Horst and D\"orte Kreher.
\newblock A weak law of large numbers for a limit order book model with fully state dependent order dynamics.
\newblock Preprint, 2015.

\bibitem{Horst-Paulsen}
Ulrich Horst and Michael Paulsen.
\newblock A law of large numbers for limit order books.
\newblock {\em Math. Oper. Res.}, to appear.

%\bibitem{HorstRothe}
%Ulrich Horst and Christian Rothe.
%\newblock Queuing, social interactions, and the microstructure of financial
%  markets.
%\newblock {\em Macroecon. Dyn.}, 12(02):211--233, 2008.

\bibitem{Lakner1}
Peter Lakner, Josh Reed and Sasha Stoikov.
\newblock High Frequency Asymptotics for the Limit Order Book.
\newblock Preprint, 2014.

\bibitem{Lakner2}
Peter Lakner, Josh Reed and Florian Simatos.
\newblock Scaling limit of a limit order book via the regenerative characterization of L\'{e}vy trees.
\newblock ArXiv e-print 1312.2340v2. 2014.

\bibitem{Rosenbaum1}
Weibing Huang, Charles-Albert Lehalle and Mathieu Rosenbaum.
\newblock Simulating and analyzing order book data: The queue-reactive model.
\newblock ArXiv e-print 1312.0563v2. 2014.

\bibitem{Rosenbaum2}
Weibing Huang and Mathieu Rosenbaum.
\newblock Ergodicity and diffusivity of Markovian order book models: a general framework.
\newblock ArXiv e-print 1505.04936v1, 2015.

\bibitem{JacodShiryaev2002}
Jean Jacod and Albert~N. Shiryaev.
\newblock {\em Limit theorems for stochastic processes}, volume 288 of {\em
  Grundlehren der Mathematischen Wissenschaften}.
\newblock Springer-Verlag, Berlin, second edition, 2003.

\bibitem{KangWilliams2007}
Weining Kang and Ruth~J. Williams.
\newblock An invariance principle for semimartingale reflecting {B}rownian
  motions in domains with piecewise smooth boundaries.
\newblock {\em Ann. Appl. Probab.}, 17(2):741--779, 2007.

\bibitem{KRM}
Martin Keller-Ressel and Marvin M\"uller. 
\newblock A Stefan-type stochastic moving boundary problem.
\newblock Preprint, 2015.

\bibitem{Kruk2003}
Lukasz Kruk.
\newblock Functional limit theorems for a simple auction.
\newblock {\em Math. Oper. Res.}, 28(4):716--751, 2003.

\bibitem{Kurtz}
Thomas Kurtz.
\newblock Strong approximation theorems for density dependent Markov chains.
\newblock {\em Stoch. Process. Appl.}, 6:223 -- 240, 1978.

\bibitem{Kushner-1974}
Harold~J. Kushner.
\newblock On the weak convergence of interpolated {M}arkov chains to a
  diffusion.
\newblock {\em Ann. Probab.}, 2:40--50, 1974.

\bibitem{LasryLions}
Jean-Michel Lasry and Pierre-Louis Lions.
\newblock Mean-field games.
\newblock {\em Jap. J. Math., 2(1):229--260, 2007}. 

\bibitem{lehalle}
Aim\'e Lachapelle, Jean-Michel Lasry, Charles-Albert Lehalle and Pierre-Louis Lions.
\newblock Efficiency of the price formation process in presence of high frequency participants: a mean-field game analysis.
\newblock {\em Math. Fin. Econ.}, to appear.

\bibitem{lowther}
George Lowther.
\newblock L\'{e}vy's Characterization of Brownian Motion.
\url{https://almostsure.wordpress.com/2010/04/13/levys-characterization-of-brownian-motion/}, 2010.

\bibitem{Osterrieder}
J\"{o}rg Osterrieder.
\newblock {\em Arbitrage, Market Microstructure and the Limit Order Book}.
\newblock Ph.D. thesis, ETH Zurich. DISS. ETH Nr 17121, 2007.

\bibitem{Rosu}
Ioanid Ro{\c{s}}u.
\newblock A dynamic model of the limit order book.
\newblock {\em Rev. Financ. Stud.}, 22(11):4601--4641, 2009.

%\bibitem{Simatos}
%Florian Simatos.
%\newblock Coupling limit order books and branching random walks.
%\newblock {J. Appl. Probab., 51 (3), 625-639, 2014.}

\bibitem{wal86}
John~B. Walsh.
\newblock An introduction to stochastic partial differential equations.
\newblock In {\em \'{E}cole d'\'et\'e de probabilit\'es de {S}aint-{F}lour,
  {XIV}---1984}, volume 1180 of {\em Lect. Notes Math.}, pages 265--439.
  Springer, Berlin, 1986.


\end{thebibliography}

\end{document}